\documentclass[11pt]{article}
\usepackage{amsthm,amsmath,amssymb}%bbm}
\usepackage{natbib}
\usepackage{multirow}
\usepackage[pdftex]{graphicx}
\usepackage{subfigure}
\usepackage{booktabs}
\usepackage{array}
\usepackage{url}
\usepackage{algorithm}
\usepackage{algorithmic}
\usepackage{bm}
\usepackage{mathrsfs}
\usepackage{color}
\usepackage[utf8]{inputenc}
\usepackage{longtable}
\usepackage{graphicx}
\usepackage{caption}
\usepackage{float}
\usepackage{threeparttable}

\usepackage{multirow}

\usepackage[usenames,dvipsnames,svgnames,table]{xcolor}
\usepackage[colorlinks,
linkcolor=red,
anchorcolor=blue,
citecolor=blue
]{hyperref}

\usepackage{star}

\newcommand \ru{\mathrm{u}}
\newcommand \rv{\mathrm{v}}

\newcommand{\nn}{\nonumber}

\def\##1\#{\begin{align}#1\end{align}}
\def\$#1\${\begin{align*}#1\end{align*}}

\newcommand {\vc}{\textnormal {vc}}
\newcommand {\diam}{\textnormal {diam}}
\renewcommand {\d}{\textnormal {d}}

\newcommand {\opt}{\textnormal {opt}}
\newcommand {\median}{\textnormal {median}}

\def\T{\mathrm{\scriptscriptstyle T}} %%%transpose operator
\def\sn{\sum_{i=1}^n}

\newcommand {\ro}{\text {o}}
\newcommand {\rO}{\text {O}}

\newcommand{\Rom}[1]{\text{\uppercase\expandafter{\romannumeral #1\relax}}}

\def\bet{\begin{theorem}}
	\def\eet{\end{theorem}}
\def\bel{\begin{lemma}}
	\def\eel{\end{lemma}}

\usepackage{geometry}
\usepackage{enumitem}
\usepackage{mathtools}

\newcommand{\scolor}[1]{{\color{magenta}#1}}

\newcommand{\scomment}[1]{\scolor{$\dagger$}\marginpar{\tiny\scolor{S:\ #1}}\hspace{-3pt}}

\begin{document}
	
\title{\LARGE Adaptive Capped Least Squares}
	
\author{Qiang Sun\footnotemark[1], ~Rui Mao\thanks{Department of Statistical Sciences, University of Toronto, Toronto, ON M5S 3G3, Canada. E-mail:{\texttt{qsun@utstat.toronto.edu, rrui.mao@mail.utoronto.ca}}.} ~and ~Wen-Xin Zhou\thanks{Department of Mathematics, University of California, San Diego, La Jolla, CA 92093, USA. E-mail:  {\texttt{wez243@ucsd.edu}}.}}
	\date{} 
	\maketitle	
	\vspace{-0.25in}
	% typeset the title of the contribution
\begin{abstract}
This paper proposes the capped least squares regression with an adaptive resistance parameter, hence the name, adaptive capped least squares regression.    The key observation is, by taking the resistant parameter to be data dependent, the proposed estimator achieves full asymptotic efficiency without losing the resistance property: it achieves the maximum breakdown point asymptotically. Computationally, we formulate the proposed regression problem as a quadratic mixed integer programming problem, which  becomes computationally expensive when the sample size gets large. The data-dependent resistant parameter, however, makes the loss function more convex-like for larger-scale problems. This makes a fast randomly initialized gradient descent algorithm possible for global optimization. Numerical examples indicate the superiority of the proposed estimator compared with classical methods. Three data applications to cancer cell lines, stationary background recovery in video surveillance, and blind image inpainting showcase its broad applicability.  
\end{abstract}
	
\noindent {\bf Keywords}: Breakdown point, capped least squares, data-dependent, efficiency, localized empirical process, $M$-estimator, resistance.
	
\section{Introduction}\label{sec:1}
%{\bf background}	
	
Suppose we  collect data points $\{(x_i, y_i) \}_{i=1}^n$ that follow a linear model
\$
y_i=x_i^\T \beta^* +\epsilon_i, \ \ i=1,\ldots, n, 
\$
where $y_i$ is a univariate response,  $x_i$ is a $d$-dimensional predictor, $\beta^*$ is the regression coefficient vector, and $\epsilon_i$ is a random error. A standard approach for estimating $\beta^*$ is to solve the  empirical risk minimization problem of the form
\$
\widehat \beta \in \argmin_{\beta\in \RR^d} \frac{1}{n} \sum_{i=1}^n  f(\beta; x_i , y_i),
\$
where $f$ is a risk or loss function.  When data are collected without contamination, the risk function  $f$ is often taken as $f(\beta; x_i, y_i):= \ell (y_i-x_i^\T\beta ) = (y_i - x_i^\T \beta )^2/2$. The resulting estimator $\hat \beta$ is known as the ordinary least squares (OLS) estimator, which is also the maximum likelihood estimator when the errors $\epsilon_i$ follow independent and identical  normal distributions. 
Albeit having desirable statistical and computational properties, the ordinary least squares estimator  is highly sensitive to outliers in both the feature and the response space.

%\scolor{
%The least median of squares (LMS) estimator and the least trimmed squares (LTS) estimator \citep{rousseeuw1984least} were among the first  equivariant regression estimators that attain maximum asymptotic breakdown point of $1/2$. The LMS estimator, however, only has a cubic rate of convergence and thus has zero asymptotic efficiency. The LTS estimator is asymptotically normal but has a low efficiency of $7\%$ under the normal errors \citep{maronna2019robust}. 
%}

Since the seminal work of  \cite{huber1964robust}, many authors have  proposed robust regression methods by replacing the square loss by some loss that grows slowly at tails. For example, the Huber loss exhibits a linear growth while away from zero.  Although it  protects against outlying $y_i$, it is not robust to outliers in the feature space:  one single outlier in the feature space can have arbitrary large effect  on the estimate. To characterize this phenomenon, \cite{hampel1971general} introduced the notion of breakdown point, which is defined as the smallest percentage of contaminated data that can cause  the estimator to take on arbitrarily large aberrant values.  Because robustness can be somewhat a broad concept, we prefer to call an estimator with high breakdown point a resistant estimator. % lHowever,  Huber loss is  This type of loss functions is less sensitive to large deviations of $y_i$ from $x_i^\T \beta $. 
\cite{rousseeuw1984least} proposed the least median of squares (LMS) estimator and the least trimmed squares (LTS) estimator, which were among the first  equivariant regression estimators that can attain asymptotic maximum breakdown point of $1/2$. The LMS estimator, however, only has a cubic rate of convergence and thus has zero asymptotic efficiency. The LTS estimator is asymptotically normal but has a low efficiency of $7\%$ under the normal errors \citep{maronna2019robust}. Other resistant regression estimators include the S-estimator  \citep{rousseeuw1984robust}, the $\tau$-estimator \citep{yohai1988high}, the MM-estimator \citep{yohai1987high}, and the rank-based estimator \citep{wang2020tuning}. 
 \cite{maronna2019robust} have argued that the redescending $M$-estimators, such as the Tukey's biweight $M$-estimator, achieve better balance between  resistance and efficiency.

%\scolor{\today.} 

For most resistant $M$-estimators, there is often a tuning constant $\tau$, which we refer to as the resistance parameter,  governing the tradeoff between resistance and efficiency. %\scolor{Following \cite{sun2019adaptive}, we refer to $\tau$ as the {\it robustification parameter}.} 
A common practice is to pick the resistance parameter based on the $95\%$ asymptotic efficiency rule \citep{western1995concepts} %-- picking the resistance parameter such that the resulting estimator is $95\%$ efficient at normals -- 
so that resistance can be introduced at a manageable cost of efficiency. %by sacrificing efficiency. 
%Throughout the paper, we interchangeably use resistance and robustness for the same meaning.   
Despite being resistant, such an approach  makes resulting estimator biased % away from the model generating parameter $\beta^*$ 
when error distributions are asymmetric \citep{fan2017estimation, sun2019adaptive}.  %and thus lose statistical accuracy . 
Moreover, a fixed-$\tau$ loss is asymptotically different from that derived from the likelihood principal, so that the resulting robust estimator is less efficient. %while tuning the resistance parameter towards more efficiency causes loss of resistance. Resistance and efficiency seem to be  two contrasting properties that can not be achieved simultaneously for a single-step estimator, as conjectured by \cite{riani2014consistency}. }
Several authors have worked on improving the efficiency of an initial resistant estimator  using  multi-stage procedures.  %A typical choice for  the initial estimator is the LTS estimator. %which improve an initial estimator that is already resistant but not efficient.
 %to improve the efficiency by exploiting multi-stage estimators: the first step produces a robust yet inefficient estimator and later steps improve the efficiency of the initial robust estimator. 
%For example, \cite{yohai1987high} proposed a three-stage estimator, the MM-estimator, in which the third stage uses a one-step estimation procedure to improve efficiency. Another example is the $\tau$-estimator \citep{yohai1988high}. 
%However, tuning up these estimator for high efficiency will be accompanied by an increase in bias as a side-effect. 
%However,  neither  cannot achieve full asymptotic efficiency and positive breakdown point simultaneously. 
%Other efficient and high breakdown point regression methods include the S-estimator  \citep{rousseeuw1984robust} and the $\tau$-estimator \citep{yohai1988high}, among others. %\scomment{Cite Doug's two papers here.} 
 %\cite{agostinelli1998one} and 
For example, \cite{gervini2002class} proposed a fully asymptotically efficient multi-stage estimator with high breakdown point: the later stages use a sequence of reweighted least squares to improve the  efficiency.  %in which the weights are based on an adaptive measure of discrepancy between the empirical distribution of the errors from an initial robust estimator and the assumed normal distribution. 
%The finite-sample efficiency of these estimators are sometimes less desirable possibly due to their sensitivity to empirical distribution estimates and the assumed normal distribution. 
%To further close the gap between efficiency and resistance, 
 \cite{bondell2013efficient} proposed an empirical-likelihood framework that down-weights outlying observations by measuring the divergence between the empirical likelihood and the normal error distribution. Both work assumes symmetric errors.

Given all the aforementioned works, it seems that efficiency and resistance can only be achieved simultaneously via multi-stage procedures. 
A natural question that arises is as follows:
{Is there a simple, one-stage estimator,  without estimating any empirical distribution, that can achieve full asymptotic efficiency and maximum breakdown point?}
This paper gives a positive answer by proposing such a one-stage estimator. Unlike most of the literature which only considers symmetric errors with unimodal densities, we accommodate the practice where errors can be asymmetric, and characterize the tradeoff between the statistical accuracy (bias), asymptotic efficiency and robustness.   The key observation is that the robustification parameter  $\tau=\tau(n)$ should  grow as the sample size grows, in which case  the estimation bias diminishes asymptotically.  %, thus making the estimator asymptotically unbiased. 
From the statistical efficiency perspective, the loss function approaches to that derived from the likelihood principal under normals as the resistance parameter grows to infinity in the asymptotic limit, and thus the resulting estimator can be fully asymptotically efficient. Theoretically, we prove its asymptotic properties  under mild moment conditions with increasing dimensions.  Counterintuitively, by taking $\tau=\tau(n)$ to be data-dependent, the ensuing estimator maintains the robustness property, that is, it achieves the maximum high breakdown point of $1/2$ asymptotically. %This data-dependent idea has been explored  in our earlier work \citep{sun2019adaptive, fan2019farmtest,liu2019robust} to improve the nonasymptotic performance for the classic Huber estimators in the presence of heavy-tailed errors  but without any contamination to the features. 

Computationally, we formulate the corresponding optimization program as a mixed quadratic integer programming (MIP) problem which can be readily solved by \texttt{CPLEX}. However, solving MIP for the global optimum is NP-hard and thus computationally intractable for large scale problems. %Due to nonlinearity, even certifying that a point is a local optimum can be NP-hard.
%is generally NP-hard in the worst case.
%solving it is computational expensive. % The proposed estimator involves solving a nonconvex problem, which is  generally NP-hard using the worst case analysis.       
By allowing the resistance parameter to grow with the sample size, the problem becomes more convex-like as $n\rightarrow \infty$, and thus become easier computationally when the sample size increases. Indeed, since the empirical loss function has a growing quadratic region, first-order algorithms will converge as long as the starting point is not in the flat region. This motivates us to propose a randomly initialized gradient descent algorithm, which is able to  find the global optimum with high probability. An R package that implements our algorithm can be found at \url{https://github.com/rruimao/ACLS}.%This property will also allow us to design other convex-relaxation type method, which I will discuss in detail in the second detailed research objective. 

The rest of this paper proceeds as follows. In Section \ref{sec:2}, we introduce the adaptive capped least squares regression and prove its resistance property. Section \ref{sec:4} is devoted to asymptotic properties where we show the proposed estimator achieves full asymptotic efficiency. Section \ref{sec:5} presents algorithms. Simulation studies and real data applications are provided in Sections  \ref{sec:6} and \ref{sec:7} to support our method and theory.  We close this paper with a discussion in Section \ref{sec:8}.\\

%\scomment{{\bf Related work}: Depth-based methods. }
	
\noindent{\sc Notation} We summarize here the notation that will be used throughout the paper. For any vector $\ru = (\ru_1, \ldots, \ru_d)^\T \in \RR^d$ and $q \geq 1$, $\|\ru\|_q=\big(\sum_{j=1}^d |\ru_j|^q\big)^{1/q}$ is the $\ell_q$ norm. For any vectors $\ru , \rv \in \RR^d$, we write $\langle \ru, \rv \rangle = \ru^\T \rv$. We  use $C$ to denote a generic constant which may change from line to line. For two sequences of real numbers $\{ a_n \}_{n\geq 1}$ and $\{ b_n \}_{n\geq 1}$, $a_n \lesssim b_n$ denotes $a_n \leq C b_n$ for some constant $C>0$, $a_n \gtrsim b_n$ if $b_n \lesssim a_n$, and $a_n \sim b_n$ indicates that $a_n \lesssim b_n$ and $b_n \lesssim a_n$. %For  random variables $X_n$ and $X$, $X_n \overset{\PP}{\to}  X$ indicates convergence in probability. 
If $A$ is an $m\times n$ matrix, we use $\| A \|_q$ to denote its $\ell_q$ operator norm, defined by $\| A \|_q = \max_{ \ru \in \RR^n} \| A \ru \|_q/\|\ru\|_q$. 
\section{Methodology}\label{sec:2}
%We assume that $\epsilon_i$'s are independently generated such that 
%\$
%v_{1+\xi}=\max_{1\leq i \leq n }\EE|\epsilon_i|^{1+\xi}<\infty. 
%\$

We start with definitions of the capped least squares loss and  the resistance parameter. 
\begin{definition}[Capped Least Squares, CLS]
The  capped least squares loss  $ \ell_{\tau}(\cdot) $ is defined as 
\begin{eqnarray*}
\ell_{\tau}(x)=\begin{cases}
                           x^2/2, & \textnormal{if}\  |x|\leq \tau; \\     
                           \tau^2/2, & \textnormal{if}\  |x|>\tau,
                    \end{cases}   
\end{eqnarray*}
where $\tau=\tau(n)>0$ is referred to as the resistance parameter. 
\end{definition}
 %This distinguishes our  from the classical redescending $M$-estimator \citep{maronna2019robust}. 
 The loss $\ell_\tau(\cdot)$ is quadratic for small values of $x$ and stays flat when $x$ exceeds $\tau$ in magnitude. The parameter $\tau$ therefore controls the blending of the quadratic and flat regions, and the flat region brings resistance. %We will use the term, adaptive capped least squares (ACLS), to emphasize the fact that $\tau=\tau(n)$ should be data-dependent to achieve  efficiency and resistance simultaneously. An increasing $\tau$ can also help reduce the estimation bias when the errors are asymmetric. This distinguishes our framework from the classical setting. 
Define the empirical loss function $ \cL_{n,\tau}(\beta )=n^{-1} \sum_{i=1}^n \ell_{\tau}(y_i-x _i^\T \beta ) $.  The adaptive capped least squares estimator is then defined as 
	\begin{eqnarray}\label{eq:flathuber}
	\hat \beta _{\tau} \in \argmin_{\beta \in \mathbb{R}^d}\cL_{n,\tau}(\beta )=\argmin_{\beta\in \RR^d}\left\{\frac{1}{n} \sum_{i=1}^n \ell_{\tau}(y_i-x _i^\T \beta )\right\}.
	\end{eqnarray}
With a growing resistance parameter $\tau=\tau(n)$, the empirical loss $\cL_{n,\tau}(\beta)$ approaches the least squares loss function, so that the ensuing estimator is asymptotically unbiased and efficient. Perhaps surprisingly, this does not cause any loss of resistance. We will use the term, adaptive capped least squares (ACLS), to emphasize the fact that $\tau=\tau(n)$ should be data-dependent.  This distinguishes our framework from classical redescending-type estimators.  %to achieve  efficiency and resistance simultaneously. 

%\scolor{\today}

Comparing with the ordinary least squares estimator, outliers are completely removed in the  estimation procedure \eqref{eq:flathuber} which leads to resistance. To see this,  we first fix $\tau$ and write
	$
	\psi_\tau(x)=x1(|x|\leq \tau),
	$
	which can be  thought as the first order derivative of $\ell_\tau(x)$, except at the point of $\tau$ and $-\tau$. Heuristically,  $\widehat\beta_\tau$ solves the estimating equations 
	\$
\frac{1}{n}\sum_{i=1}^n \psi_\tau (y_i -x_i^\T\beta ) x_i=0.
	\$
Note that $\psi_\tau(r_i)=0$ if the $i$-th residual $r_i=y_i-x_i^\T\beta $ exceeds $\tau$ in magnitude, which corresponds to a ``bad" fit possibly due to an outlying observation for the $i$-th individual. As a result, this sample does not contribute to the above estimating equations. This implies that the capped least squares estimator could be more robust than the Huber estimator \citep{huber1964robust}. Comparing with other redescending-type losses such as the Tukey's biweight loss, the CLS estimator is more efficient. This is because the CLS loss is exactly quadratic in the center while, for example, the Tukey's biweight loss is not.  Figure \ref{fig:eff} shows the relative efficiency of the Tukey's biweight estimator over the  CLS  estimator over a sequence of same tuning parameters $\tau$. We shall mention that the CLS estimator is also referred to as the  trimmed mean estimator under location models \citep{huber1964robust}. 

\begin{figure}[t]
	\centering
	\includegraphics[width=.6\textwidth]{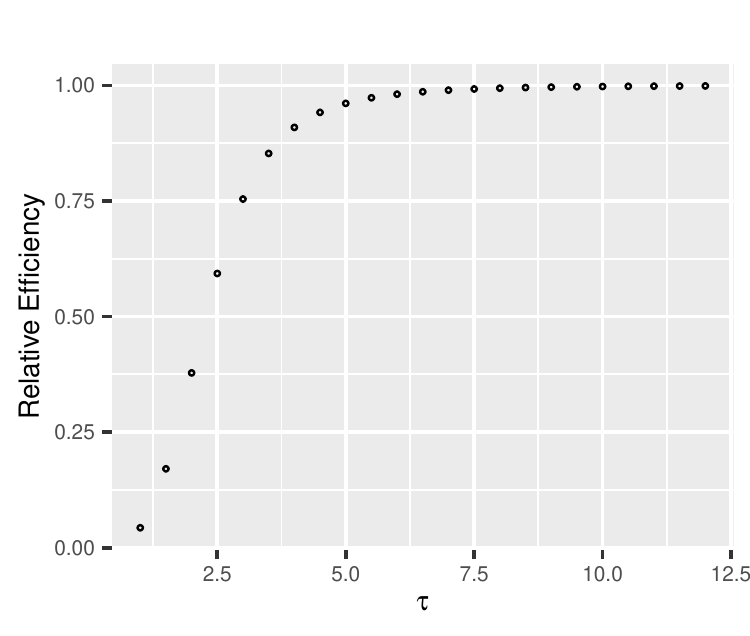} 
	\caption{Relative efficiency of the Tukey's biweight estimator vs the CLS estimator.}\label{fig:eff}
\end{figure}%\scomment{We need some discussions here about comparison with Redescending loss.}

%to that resulting from the likelihood principal when Gaussian distributions are assumed

While $\tau$ serves as a capping parameter that encourages resistance, at first sight it seems that $\widehat \beta_\tau$ will lose resistance when $\tau=\tau(n)$ tends to infinity with $n$. Perhaps surprisingly, at least to us,  we will show in what follows: the resistance property is preserved  as long as $\tau=\tau(n)$ is  fixed for every $n$. Denote the collected $n$ data vectors by $\cZ=\{z_1, \ldots, z_n\}$, where $z_i=(x_i, y_i)$.  For any estimator $\widehat \beta$, as a function of $\cZ$, the finite-sample breakdown point $\varepsilon^*(\widehat\beta , \cZ)$ of $\widehat\beta$  \citep{donoho1983notion,hampel1971general} is defined as 
	\$
	\varepsilon^*(\widehat\beta , \cZ)=\frac{1}{n}\min\Bigg\{m: \sup_{\tilde\cZ\in \cP_m(\cZ)}\big\|\beta (\tilde\cZ)\big\|_2=\infty\Bigg\},
	\$
	where $\cP_m(\cZ)=\big\{\widetilde\cZ:\,  \tilde z_i\ne z_i,\, i\in \cI~\textnormal{s.t.} ~ |\cI|\leq m\big\}$. Denote by $\cG$ the set of uncontaminated/clean samples. The following result is on the breakdown point of the Huber estimator, taken from \cite{maronna2019robust}. 
	
	\begin{proposition}\label{pro:huber}
	The Huber estimator has breakdown point at most $1/n$. 
	\end{proposition}

	%\subsection{A Motivation: The Breakdown Point of Huber Regression}
	%XX
	The above proposition demonstrates that the Huber estimator, and thus the adaptive Huber estimator,  is not robust to outliers in the feature space. Our first main result is on the high breakdown point of  the proposed ACLS estimator, under the general position assumption \citep{mili1996robust}. %We first impose the following general position assumption \citep{mili1996robust} on the data.  
	\begin{assumption}[General Position]\label{ass:gp}
		We assume the data  are in general position, that is, any $d$ of them give a unique determination of $\beta$. 
	\end{assumption}
		
	\begin{theorem}[High Breakdown Point]\label{thm:2}
	Assume Assumption \ref{ass:gp} and let $c_r=\sum_{i\in \cG} \epsilon_i^2/(2|\cG|)$.  Then, the adaptive capped least squares estimator has breakdown point of at least 
		\$
		\frac{1}{n}\left(\frac{n(1-2c_r/\tau^2)-d+1}{2(1-c_r/\tau^2)}\right)\rightarrow \frac{1}{n}\left\lfloor\frac{n-d+2}{2}\right\rfloor  ~\textnormal{as}~\tau^2/c_r\rightarrow \infty. 
		\$
	\end{theorem}

	Note that the adaptive capped least squares regression estimator is regression equivariant \citep{rousseeuw1987}, that is,  
	\$
	\widehat\beta_\tau\big(\{x _i, y_i+x _iv), i=1,\ldots n\}\big)=\widehat\beta_\tau\big(\{x _i, y_i), i=1,\ldots n\}\big)+v
	\$
	for any $v\in \RR^d$.   {The maximum breakdown point for regression equivariant estimator is $\left\lfloor(n-d+2)/2\right\rfloor/ n$} \citep{muller1995breakdown,mizera2002breakdown}.  Therefore, provided that $ \tau^2/c_r\rightarrow\infty $, the adaptive capped least squares regression estimator achieves the maximum breakdown point asymptotically. 
	
%\scolor{\today}

A common practice is to determine $\tau$ by the $95\%$ asymptotic efficiency rule \citep{maronna2019robust} to introduce resistance at the cost of efficiency. %while maintaining resistance.
%It is a common belief is to fix $\tau$, usually by the $95\%$ efficiency rule \citep{maronna2019robust}, in order to introduce  resistance. 
Yet our result suggests that, by letting $\tau=\tau(n)$ to be data-dependent and  diverging, %as long as nonasymptotically, 
 high breakdown point can be preserved. At the same time, full asymptotic efficiency is achieved as the CLS loss is exactly quadratic in an increasing region (with $n$) around the origin.
% behaves like least squares in the center. 
This is somewhat counter-intuitive at first glance. A careful examination, however, reveals that these two properties do not necessarily contradict each other. This is because the breakdown point is a finite-sample property while the statistical efficiency is an asymptotic notion.  %\scomment{Cite doug's paper here.}

To illustrate the intuition, let us consider estimating the mean of a random variable. Suppose the uncontaminated data $\{ y_i \}_{i=1}^n$ follow the local model  
\$
y_i=\mu + \epsilon_i, ~ i =1,\ldots, n ,
\$
while we only observe the contaminated versions $\tilde y_i= y_i + \gamma_i$ with $\gamma_i$ denoting a mean shift parameter to indicate the deterministic contamination.  %In the literature, the resistance (or robustness) of  a procedure is often measured by the breakdown point \citep{hampel1971general}, which is roughly defined as the largest proportion of contaminated $\tilde y_i$'s with $\gamma_i$'s being arbitrarily large  (we can let $\gamma_i\rightarrow \infty$ independently  of $n$) 
The breakdown point of a procedure can be roughly understood as how many arbitrary $\gamma_i$'s the procedure can tolerate before it produces arbitrarily large estimates. For any fixed $n$ and $\tau(n)<\infty$, the procedure $\argmin_{\mu} (1/n)\sum_{i=1}^n\ell_\tau(\widetilde y_i-\mu)$ achieves high breakdown point by eliminating the effects of arbitrarily large $\gamma_i$'s when $|\widetilde y_i-\mu|>\tau$. In other words, the breakdown point property does not depend  on the actual value of $\tau$ as long as it is finite for every $n$. 
On the other hand, a data-dependent and diverging $\tau$ brings full asymptotic efficiency. We prove this rigorously in the next section.

Finally, we point out that, in the processing of finishing this manuscript, the authors  were pointed out a previous paper by \cite{he2000breakdown} showed that weighted $t$-type regression estimators of a can have high breakdown points and are efficient for certain heteroscedastic $t$-models due to being maximum likelihood estimators. \cite{simpson1987minimum}  showed full efficiency and resistance for minimum Hellinger distance estimator but only  at a target model like a Poisson model for count data. These work showed the promise of developing a single-stage fully efficient and highly resistant estimator in the general case, as provided by our paper. %Rigorous arguments will be provided in the following section. 

% \scomment{rewrite this paragraph}
	%Theorem \ref{thm:2} look rather counterintuitive at first glance as $\tau$ has to go to $\infty$ so that the adaptive resistant regression estimator can achieve  the highest breakdown point. This is, however, justifiable, because  a small $\tau$ would damage the estimator by cutting off the good and  uncontaminated samples.  %should not be removed by a too small $\tau$. %We are not taking $\tau$ to be infinity at first, but instead, $\tau$ must be much larger than absolute residuals produced from good observations, in which case the good observations and contaminated samples can be  separate
	
%	Another intuitive explanation why the robustification parameter $\tau$ should grow is that  $\tau$ should be at least in the order of $\max_{i\in \cG} |\epsilon_i|$ so that the estimator uses all the original and good samples. Now if $\epsilon_i$ has bounded $(1\!+\!\xi)$-th moment, that is $m_{1+\xi}:=\max_{1\leq i \leq n}|\epsilon_i|^{1+\xi}$. Then to make $\tau^2/c_r \rightarrow \infty$, $\tau$ should be at a large order of 
%\$
%c_r=\EE \sum_{i\in \cG} \epsilon_i^2/(2|\cG|)
%\$

	\section{Asymptotic  Properties}\label{sec:4}
%\scomment{Restructure this section.}	
This section establishes the asymptotic properties for the proposed robust estimator.  For any prespecified $\tau>0$,  $\hat \beta_\tau$ is an $M$-estimator of 
	\$
	\beta_\tau^*=\argmin_{\beta \in \RR^d}  \cL_\tau(\beta) ~\mbox{ with }~  \cL_\tau(\beta):=  \frac{1}{n} \sum_{i=1}^n\EE \ell_\tau (y_i-x_i^\T\beta). 
	\$
We call $\beta_\tau^*$ the capped least squares regression coefficient, which is not necessarily equal to $\beta^*$ unless the error distributions are symmetric. By setting $\tau=\infty$, we have $\beta^*_\infty = \beta^*$ as long as $\EE[\epsilon_i|x_i]=0$. Throughout the paper, we assume that the capped least squares regression coefficient approaches $\beta^*$ as $\tau\rightarrow \infty$, that is, $ \| \beta_\tau^* -   \beta^* \|_2 \rightarrow 0$ as $\tau\rightarrow \infty$. To prove consistency, we impose the following assumption, which essentially assumes that the population global optimum is unique.  
	
%Let 
%$
%\cL(\beta)=n^{-1}\sum_{i=1}^n\EE\{ \ell_\tau(y_i-x_i^\T\beta)\}.
%$ 
	
	\begin{assumption}[Separability]  \label{ass:sep}
		For  every $\varepsilon>0$, $\beta \to \cL_\tau(\beta)$ with $\tau>0$ sufficiently large satisfies 
		\$
		\inf_{\beta: \|\beta- \beta_\tau^*\|_2\geq \varepsilon} \cL_\tau(\beta)>\cL_\tau(\beta_\tau^*). 
		\$
	\end{assumption}
	
	\begin{theorem}[Consistency]\label{thm:consistency}
		Suppose Assumption \ref{ass:sep} holds. Provided that the triplet $(d,n,\tau)$ satisfies $\tau^{4}   d/n\rightarrow 0$ as $n\to \infty$, we have $\widehat\beta_\tau  \overset{\mathbb{P}}{\to} \beta^*_\tau$. In addition, if $\tau \rightarrow \infty$, then $\widehat\beta_\tau \overset{\mathbb{P}}{\to} \beta^*.$
	\end{theorem}

	%For notational simplicity, let $m_k=\EE |\epsilon|^k$ be the $k$-th absolute moment of $\epsilon$. 
	
	The above theorem shows that as long as  $\tau$ does not grow too fast, the capped least squares estimator  $\widehat\beta_\tau$ is consistent. Technically, if we take the parameter space to be a compact subset of $\RR^d$, then the scaling condition in Theorem~\ref{thm:consistency} can be relaxed to $\tau = o(\sqrt{n})$. 
	To obtain the convergence rate and  the asymptotic normality, we need the following anti-concentration type property for random errors.

		\begin{assumption}[Anti-concentration]\label{ass:anti}
	There exist positive constants $L_0$, $\eta$ and  $\tau_0$ such that for any $\Delta \in [0, \tau/2]$ and $\tau\geq \tau_0$, 
	 \begin{align*}
		%\EE \left\{\epsilon_i^2 1(\tau-\Delta\leq |\epsilon_i |\leq \tau+\Delta )\right\}\leq L_0\Delta/\tau,\\
		\max_{1\leq i \leq n}\EE \left\{|\epsilon_i |^{4+\eta}1(\tau-\Delta\leq |\epsilon_i |\leq \tau+\Delta )\right\}\leq \frac{L_0\Delta}{\tau}.
		\end{align*}
	\end{assumption}

	Assumption~\ref{ass:anti} is an anti-concentration-type condition on the distributions of $\epsilon_i$. It is satisfied if the density of $\epsilon_i$ decays sufficiently fast. For example, assume $\epsilon_i$'s are independent and identically distributed, and denote by $p(\cdot)$ the density function of $|\epsilon_i|$. Moreover, assume that $p(x ) \lesssim x^{-(5+\eta)}$ as $x\to \infty$. Then, for any $\Delta \in [0, \tau/2]$ with $\tau>0$ sufficiently large,
$$
	\EE \left\{|\epsilon_i |^{4+\eta}1(\tau-\Delta\leq |\epsilon_i |\leq \tau+\Delta )\right\} = \int_{\tau-\Delta}^{\tau+\Delta} t^{4+\eta} p(t) {\rm d} t \lesssim \int_{\tau-\Delta}^{\tau+\Delta}  t^{-1} {\rm d} t \lesssim \frac{\Delta }{\tau}.
$$
%basically states that the $(4+\eta)$-th absolute moment of $\epsilon_i$ constrained in a short interval should be small. For example,  it holds for random errors with exponential decaying densities.  Indeed, Assumption \ref{ass:anti} holds under much milder conditions where $\epsilon_i$ has a polynomial decay density $p(x)$ such that $p(x) \sim x^{-(5+\eta)}$ as $|x|\rightarrow \infty$. 

To prove asymptotic normality, we  need the following condition as a stronger version of Assumption \ref{ass:sep}.

	\begin{assumption}[Local Strong Convexity]\label{ass:inv}
		There exist some radius $r$, a curvature parameter $\kappa_\ell $ and a tolerance parameter $\eta$ such that, for any $\beta\in \BB_2(\beta^*_\tau, r)\coloneqq\{\beta: \|\beta-\beta^*_\tau\|_2\leq r\}$,
		\$
		\cL_\tau(\beta) - \cL_\tau(\beta^*_\tau)  \geq \kappa_\ell \|\beta-\beta_\tau^*\|_2^2. 
		\$
	\end{assumption}

The above assumption basically requires the population loss function to be locally strongly convex. This is true for the adaptive capped least squares loss, along with other nonconvex losses such as Tukey's biweight loss \citep{maronna2019robust}. For the predictors $x_i$'s, we impose the following boundedness assumption, which is standard in regression analysis with fixed designs.

\begin{assumption}\label{ass:design}
 There exists some constant $C_x$ such that $\max_{1\leq i\leq n}\|x_i\|_2\leq C_x \sqrt{d}$. 
\end{assumption}

Assumption \ref{ass:design} can be further relaxed by sacrificing the scaling condition $n\gtrsim d^3$ in the theorems.  Our first result characterizes the order of bias induced by robustification. For $k\geq 1 $, define the moment parameter $m_k:= \max_{1\leq i \leq n} \EE |\epsilon_i|^k$.	Finally, we are ready to present the main theorem of this section.  Let   	$
	\psi_\tau(x)=x1(|x|\leq \tau)$, which can be  viewed as the first order derivative of $\ell_\tau(x)$ except at the points $\pm \tau$. Define the  matrices 
$$
	\Sigma=  \frac{1}{n}\sum_{i=1}^n x_ix_i^\T ~~\mbox{ and }~~ \Sigma_\tau = \frac{1}{n} \sn \PP(|\epsilon_i| \leq \tau ) x_ix_i^\T .
$$

\begin{theorem}[Asymptotic Normality]\label{thm:asyn}
Suppose Assumptions \ref{ass:anti} -- \ref{ass:design} hold, and $m_{4+\eta}$ exists. Further assume that the eigenvalues of $\Sigma_\tau$ are bounded away from zero and infinity for $n$ large enough. Let $n\gtrsim d^3$ and $\tau\geq n^{\frac{1}{2(1+\eta)}}$. Provided $\widehat\beta_\tau \overset{\PP}{\to}\beta_\tau^*$, we have
\$
\sqrt{n}(\widehat\beta_\tau-\beta^*)
&= \frac{1}{\sqrt{n}}\sum_{i=1}^n \Sigma_\tau^{-1}\psi_\tau(\epsilon_i)x_i+\Delta_n\leadsto \cN(0, \sigma^2\Sigma^{-1}),
\$
where $\|\Delta_n\|_2=\ro_{\PP}(1)$ and  $\sigma^2= \lim_{n\to \infty} n^{-1}\sum_{i=1}^n \sigma_i^2 $ with $\sigma_i^2= \EE (\epsilon_i^2)$.
\end{theorem}	

Theorem \ref{thm:asyn} simply states  the ACLS estimator achieves full asymptotic efficiency. This is a direct consequence of the data-dependent and growing  $\tau$.   Lastly, we verify Assumption \ref{ass:inv} under the following assumption. 
 \begin{assumption}\label{ass:mineigen}
Assume that, for $\tau,\, n$ large enough,
 \$
\rho_\tau:=\lambda_{\min}\left(\frac{1}{n}\sum_{i=1}^nx_ix_i^\T\EE 1(|\epsilon_i|\leq 3\tau/4)\right)>0.
 \$
 \end{assumption}

  \begin{lemma}\label{lemma:sc}
Assume Assumptions \ref{ass:mineigen} and \ref{ass:design}. Suppose  $\beta\in \BB_2(\beta^*_\tau, r_0)$ for some $r_0=\tau/(4\sqrt{d})$. Take $\tau$ such that $\sqrt{d}/\tau^{3+\eta}\leq c\|\beta-\beta_\tau^*\|_2$ and $\sqrt{d^3}/\tau^{5+\eta}\leq c \|\beta-\beta_\tau^*\|_2^{-1}$ for sufficiently small constant $c$. If $\|\beta^*-\beta_\tau^*\|_2\lesssim \|\beta-\beta_\tau^*\|_2$, then Assumption \ref{ass:inv} holds with $\kappa_\ell=\rho_{\tau}/4$, that is, 
\$
\cL_\tau(\beta) - \cL_\tau(\beta^*_\tau)  \geq \frac{\rho_\tau}{4} \|\beta-\beta_\tau^*\|_2^2.
\$
 \end{lemma}

 Heuristically, Assumption \ref{ass:mineigen} can be easily satisfied as long as $n$ and $\tau$ are large enough. This is because $\rho_\tau\rightarrow \lambda_{\min}\left(n^{-1}\sum_{i=1}^nx_ix_i^\T\right)$ as $\tau\rightarrow \infty$.  Lemma \ref{lemma:sc} above indicates  that the expected ACLS loss is strongly convex in the growing region of $\BB_2(\beta_\tau^*, \tau/(4\sqrt{d}))$, and thus Assumption \ref{ass:inv} holds. The lemma also indicates, that when $\tau=\tau(n)$ diverges to infinity, the problem becomes more and more convex-like.  This justifies our intuition on the landscape of the loss function in the introduction and serves as an inspiration for a randomly initialized gradient descent algorithm for fast computation especially when the scale of the problem gets larger.

\section{Algorithm}\label{sec:5}

This section  addresses the computational aspect of the proposed capped least squares regression method.
Notably, the optimization problem in \eqref{eq:flathuber}  can be formulated as, for some sufficiently large $M>0$, 
\begin{align} \label{eq:QMIP}
\min \quad	& \frac{1}{2}\sum_{i=1}^n t_i^2 + \frac{1}{2}\tau^2 \sum_{i=1}^n  z_i \\ 
\textnormal{s.t.} \quad	& |y_i-{x}_i^\T \beta| \leq t_i + M z_i,  \nonumber \\
& z_i \in \{0,1\}, \qquad \qquad \qquad \qquad i = 1, \ldots, n, \nonumber
\end{align}
where $z_i$ is a  binary decision variable indicating whether observation $i$ is an outlier, $t_i$ can be viewed as the absolute residual of a non-contaminated sample,   and  $M > 0$ is an upper bound for all residuals.   To see the equivalence, first note that the objective function is decomposable. If $z_i = 1$, the inequality constraint holds for $t_i = 0$, and the contribution from observation $i$ to the cost function is $\tau^2/2$; otherwise if $z_i = 0$, the smallest $t_i$ is attained at $|{x}_i^\T \beta - y_i|$, thereby contributing $({x}_i^\T \beta - y_i)^2$ to the cost function. Formulation \eqref{eq:flathuber} is often referred to as the big-$M$ formulation  \citep{Griva2008LinearAN}. It  is a quadratic mixed integer program (QMIP), and can be readily solved  using  \texttt{IBM ILOG CPLEX Optimization Studio}, or  \texttt{CPLEX} for short.

The QMIP problem \eqref{eq:QMIP} can be solved efficiently for small-scale problems. The computational complexity grows exponentially with the sample size. Nevertheless, due to the increasing resistance parameter, the quadratic component also with the sample size and thus the loss function becomes more convex-like. Intuitively, a randomly initialized first-order algorithm will more likely be able to find the global optima because a random initialization has a growing chance to fall in the strongly convex region as $n$ increases.  Another option is to run \texttt{CPLEX} on a smaller sub-sampled dataset to provide a coarse initialization, followed by a first-order algorithm. %\scomment{Add random initializations too.} The key observation is that with an adaptively growing robustification parameter $\tau$, the cost function for the adaptive resistant regression becomes more convex-like. Thus the gradient descent method with good initializations could find the global optimum with large probability.

In the following, we describe a randomized gradient descent algorithm starting at iteration 0 a random initialization $\beta^0\sim \textnormal{Unif}(\BB_2(\tau))$, where $\textnormal{Unif}(\BB_2(\tau))$ is a uniform distribution on the $\ell_2$-ball $\BB_2(\tau)=\{x: \|x\|_2\leq \tau\}$. At iteration $k=0,1,2,\ldots$, we define the update
%any working solution $\beta^{k}$, the gradient descent update is given by
\begin{equation} \label{eq:GD}
\beta^{k+1}=g(\beta^{k},\eta_k):=\beta^{k}-\eta_k\nabla \cL_{n,\tau}(\beta^{k}),
\end{equation}
where $\eta_k>0$ is the step size or learning rate. We adopt  an inexact line search method to search for the best possible $\eta_k$.  This method starts from a  small step size $\eta_0$, say $0.001$, successively inflates it by a factor of $\gamma_u>1$, say $2$, and computes the corresponding gradient descent update until the loss function is no longer decreasing. Once stopped, we record the step size as $\eta_k$ and compute the ensuing gradient descent update $\beta^{k+1}$. We then repeat the gradient descent update \eqref{eq:GD} until convergence, that is, until $\|\beta^{k+1}-\beta^{k}\|_2\leq \varepsilon_\opt$ for some pre-specified optimization error $\varepsilon_\opt$. We summarize the details in Algorithm \ref{alg:RGD}.

\begin{algorithm}[t]
	\caption{A randomized gradient descent (RGD) algorithm for solving problem \eqref{eq:flathuber}.}\label{alg:RGD}
	\begin{algorithmic}[1]
		\STATE{\textbf{Algorithm}: $\widehat\beta_\tau \leftarrow \text{RGD}\big(\{(x_i, y_i) \}_{i=1}^n,  \beta^0,\tau, \eta_0, \gamma_u,\varepsilon_{\opt}\big) $}\\
		\STATE{\textbf{Initialize}: $\beta^0\sim \textnormal{Unif}(\BB_2(\tau))$}
		\STATE{\textbf{Input}: $\tau, \eta_0, \gamma_u>0$}
		\STATE{\bf for $k=0,1, \ldots$}  $\textbf{until} \; \| \beta^{k+1}-\beta^{k}\|_2\leq \varepsilon_\opt\;{\bf do}$
		\begin{flalign*}
		&m\leftarrow 0\\
		&\;\;\;\;\;\;\;\;\textbf{while}\;\cL_{n,\tau}(g(\beta^k,\gamma_u^{m}\eta_0))>\cL_{n,\tau}(g(\beta^k,\gamma_u^{m+1}\eta_0))\; {\bf do }\\
		&\;\;\;\;\;\;\;\;m \leftarrow m+1\\
		&\;\;\;\;\;\;\;\;\textbf{end while}		\\
		&\eta_k=\gamma_u^m   \\
		&\beta^{k+1}\leftarrow g(\beta^{k},\eta_k)
		%& \;\;\;\;\;\;\;\;{\bf if } \;\cL_{n,\tau}(\beta^{k})< \cL_{n,\tau}(\beta^{k+1}) \;\\
		%&\;\;\;\;\;\;\;\;\beta^k=\beta^{k-1}\\
		%&\;\;\;\;\;\;\;\;\textbf{end if}	
		\end{flalign*}
		\hspace{10pt}{\bf end for}
		\STATE{{\bf Output}: $\widehat\beta_\tau=\beta^{k+1}$}
	\end{algorithmic}
\end{algorithm}

\iffalse
\begin{algorithm}[t]
	\caption{A randomized gradient descent algorithm for problem \eqref{eq:flathuber}.}\label{alg:RGD}
	\begin{algorithmic}[1]
		\STATE{\textbf{Algorithm}: $\widehat\beta_\tau \leftarrow \text{RGD}\big(X, y, \tau, K,\eta_0, \alpha,\varepsilon_{\opt}\big)$ }
		\STATE{\textbf{Parameter}: $\tau>0$, $K>0$, $\eta_0>0$, $\alpha>0$}
		\STATE{}  
		\begin{flalign*}
		&\textbf{for }k=1,2, \ldots,K\\
		&\quad \quad\text{Random generalize }\beta^0_k\\
		&\quad \quad\widehat\beta_{\tau,k}=\text{GD}\big(X, y, \tau, \beta^0_k, \eta_0, \alpha,\varepsilon_{\opt}\big)\\
		&\textbf{end for}\\
		&\widehat\beta_\tau=\argmin_k \cL_{n,\tau}(\widehat\beta_{\tau,k})
		\end{flalign*}
		\STATE{{\bf Output}: $\widehat\beta_\tau$}
	\end{algorithmic}
\end{algorithm}
\fi

\section{Numerical Studies}\label{sec:6}

This section assesses numerically the finite sample performance of the proposed method  under various settings. 
The landscape of the empirical loss function under different contamination models is also examined.
%We use simulation studies to examine the finite sample performance and the landscape of the loss functions. 

\subsection{Finite Sample Performance}

%This section examines the finite-sample performance of the proposed estimator. 
In the following numerical studies, we set the sample size  $n=50$, dimension $d=6$, and generate uncontaminated data from
\#\label{eq:truemodel}
y_i=\alpha^*+x_i^\T \theta^* +\epsilon_i, \ \ i=1,\ldots, n,
\#
where $\alpha^*$ is the intercept, $x_i$'s are independently and identically distributed (i.i.d.) as $\cN(0,\Sigma)$ with $\Sigma=(0.5^{|j-k|})\in\RR^{5\times 5}$,  and $\epsilon_i$'s are i.i.d. $\cN(0,1)$ random errors. We take $\beta^*=(\alpha^*, (\theta^*)^\T)^\T =(0,3,4,1,2,0)^\T $. In the simulations, we only have accesss to a contaminated dataset  %and consider to contaminate the the data $\{(y_i, x_i): 1\leq i\leq n\}$ 
under one of the following three scenarios:
\begin{enumerate}
	\item Scenario 1: The data are clean data $\{(y_i, x_i) \}_{i=1}^n$ with no outliers.
	
	\item Scenario 2: There are outliers only in the response space ($y$-outliers). Specifically, we generate contaminated random errors from a mixture of normal distribution
	\$
	(1-\lambda) \cN(0,1)+\lambda \cN (a,1),
	\$
	where $\lambda$ indicates the the contamination proportion and $\cN (a,1)$ is the distribution of outliers. We take $\lambda=10\%$ and $a\in \{ 10, 20, \ldots, 100\}$ in the contaminated samples.
	
%	contaminate model errors $\epsilon_i$' using the distribution $\cN(a,1)$ for some $a$. 	There are only  $y$-outliers.  The observed covariates are $x_i$'s. We contaminate $\epsilon_i$'s, and thus $y_i$'s, using the distribution $\cN(a,1)$, such that the 

	\item Scenario 3: There are outliers in both the response and predictors ($y$-outliers and $x$-outliers).  
	%There are both $x$-outliers and $y$-outliers.
	In the linear model \eqref{eq:truemodel}, we first generate contaminated random errors from a mixture of normal distribution 
	\$
	(1-\lambda) \cN (0,1)+\lambda \cN (a,1),
	\$ 
	where $\lambda=10\%$ and $a\in \{ 10, 20, \ldots, 100\}$.  
	We then add a random perturbation vector  $z_i \sim \cN(a1_{d-1},I_{d-1} )$
	to each covariate $x_i$ in the contaminated samples. 
%	The response variables $y_i$'s are then generated  according to model \eqref{eq:truemodel} using the contaminated random errors instead. For each $x_k$ in these contaminated samples, we add a random perturbation vector  $z\sim \cN(a1_{d-1},I_{d-1} )$ to further corrupt the samples. 
\end{enumerate}

In all three scenarios, we compute the adaptive capped least squares estimator via three different algorithms.  The first one uses gradient descent with random initialization as described in Algorithm \ref{alg:RGD}. In each run, we randomly initialize and then run the algorithm $200$  times, and pick the estimator with the smallest capped least squares loss. We name the first method as ACLS.
%This estimator is denoted by ARR. 
The second one uses  a more carefully-designed initialization: it first runs \texttt{CPLEX} on subsamples of size $\lceil 0.3 n \rceil$ $10$ times to compute a coarse estimate that has the smallest capped least squares loss, and then runs gradient descent initialized from this estimate. 
%The subsample size is taken to be $30\%$ of the total size. 
We call the second method as ACLS-hybrid or ACLS-h for short. %\footnote{An R package that implements ACLS, ACLS-h can be found at \url{https://github.com/rruimao/ACLS}.} 
The third one runs \texttt{CPLEX} on the full dataset. This method, denoted by ACLS-C, serves as a benchmark.  For all implementations, we set $\tau=\sqrt{n}/\log\log n\gtrsim n^{1/(2+2\eta)}$ as implied by the theory. We take $M=10^4$ in \eqref{eq:QMIP}, and $\alpha=2, \eta_0=0.001$ in Algorithm \ref{alg:RGD}.

%the gradient descent method with initial's coefficients randomly chosen from $\mathcal{N}(0,nI_d)$ and we set initial's intercepts 0. We denote the method by ARR. We randomly generate the initials for 1000 times, then we apply GD with each initial to obtain the estimator. We pick the initial which will provide the estimator with the smallest adaptive resistant loss.

%\iffalse
%\scolor{
%{\bf ARR with random cplex.}\\
%For all three scenarios, we also implement the proposed estimator using the gradient descent method with initials chosen from \eqref{eq:QMIP} on part of samples using \texttt{CPLEX}. We denote this method by ARR-hybrid. We apply two different initializations to ARR-hybrid estimator, and refer them as ARR-h-1, ARR-h-2, respectively. We set the number of iteration 10 for both methods. ARR-h-1 is that we randomly pick $30\%$ of samples for each iteration, use \texttt{CPLEX} to obtain the global solution to this subquestion and calculate its loss. The initial is chosen as the estimate who has the smallest loss. ARR-h-2 is we again sample $30\%$ data randomly and obtain initials using \texttt{CPLEX} for 10 times, then we apply GD with each initial to obtain the estimator. We pick the initial which will provide the estimator with the smallest adaptive resistant loss. We use these two kinds of initials to further implement GD algorithm described in Algorithm \ref{alg:RGD}.}
%\fi

We compare proposed estimators, ACLS, ACLS-h and ACLS-C, with three existing methods: the ordinary least squares (OLS) estimator, the least trimmed squares (LTS) estimator,  and adaptive Huber regression (AHR) estimator. These three estimators are given, respectively, by
\begin{gather*}
\widehat \beta_{\text{OLS}}=\argmin_{\beta\in \RR^d}  \sum_{i=1}^n  (y_i-x_i^{\T}\beta)^2/2,~  \widehat \beta_{\text{LTS}}
=\argmin_{\beta\in \RR^d} \ \sum_{i=1}^h  r^2(\beta; x_i,y_i )_{(i)},\\
\widehat \beta_{\text{AHR}}=\argmin_{\beta\in \RR^d} \sum_{i=1}^n  \ell_{\tau}^{{\rm H}}(\beta; x_i,y_i ),
\end{gather*}
where $r^2(\beta; x_i,y_i )_{(i)}$ is the $i$-th order statistic of the squared residuals $r^2(\beta; x_i,y_i )_i=(y_i-x_i^{\T}\beta)^2$, $h$ is the number of residuals used, and $ \ell_{\tau}^{{\rm H}}(\beta; x_i,y_i )$ is the Huber loss \citep{huber1973robust}.    
%of adaptive Huber's estimator and adaptive resistant estimator the same $\tau=\sqrt{n}/\log{\log{n}}=\sqrt{50}/\log{\log{50}} \approx 5.18$.

%\begin{eqnarray*}
%	\ell_{\tau}^H(\beta; x_i,y_i )=\begin{cases}
%		\frac{1}{2}(y_i-x _i^\T \beta )^2, & |y_i-x _i^\T \beta|\leq \tau; \\     
%		\tau\lvert y_i-x _i^\T \beta\rvert-\tau^2/2, & |y_i-x _i^\T \beta|>\tau.
%	\end{cases}
%\end{eqnarray*}
The LTS estimator is computed by the {FAST-LTS} algorithm \citep{rousseeuw1999fast}, implemented in the R package \texttt{robustbase}, with $h={(n+d+1)}/{2}$.  The AHR estimator is calculated using the iteratively reweighted least square algorithm, with the same robustification parameter $\tau=\sqrt{n}/\log{\log{n}}$ as that for ACLS.

%\iffalse
%We compare the mean square error (MSE) of ARR-hybrid estimator with MSEs of OLS estimator, LTS estimator using FAST-LTS algorithm introduced in \cite{rousseeuw1999fast} and adaptive Huber's estimator using iteratively reweighted least square (IRLS) algorithm. The MSE is defined as $ \|\hat\beta _I -\beta^{*} \|_2^2$, where $ \hat\beta _I $ is the estimate of $\beta^{*} $ obtained from some fit. FAST-LTS is implemented in the function ltsReg of the R package robustbase. In addition, we use \texttt{CPLEX} to obtain the global solution to QMIP (\ref{eq:QMIP}) as a standard best estimator.
%We also compare the standard deviations $\sigma$ among these estimators in scenario 1. We estimate standard deviation $\sigma^2$ by $\sum_{1\leq i \leq n}\widehat{\sigma_i}^2/n $, where $\widehat{\sigma_i}^2=(y_i-x_i^{\T}\beta)^2$.
%\fi

All the above estimation procedures are repeated $100$ times for $100$ randomly generated datasets. We record  the median of mean square errors (MSEs), defined as $\|\hat\beta -\beta^{*} \|_2^2$, and average CPU time (in seconds)  under all three scenarios, as well as the median of the standard deviations under the first scenario. The residual variance $\sigma^2$ is estimated by $(1/n)\sum_{i=1}^n \widehat{\sigma_i}^2 $, where $\widehat{\sigma}_i^2=(y_i-x_i^{\T}\hat \beta_{{\rm OLS}})^2$. Figures \ref{fig:1} and \ref{fig:2} show how MSEs change as the outlier mean $a$ increases in Scenario 2 and Scenario 3, respectively. Table \ref{table:Scenario1} collects the results for Scenarios $1$ -- $3$, respectively, where  the outlier mean is taken to be $a=50$ in Scenario 2, while $a=100$ in Scenario 3.

In Scenario 1, all estimators other than LTS achieve competitive mean square errors, while LTS has a slightly higher median MSE possibly due to fact that the LTS estimator has a relatively low ($7\%$) efficiency. In Scenario 2, both OLS's  and AHR's performances get worse rapidly as the outlier mean $a$ increases, while all other estimators achieve steady and satisfactory performances.  In both Scenario 2 and Scenario 3, all three of our proposed  estimators outperform LTS uniformly for every $a$. This  demonstrates the high efficiency and resistance of our proposed estimators.  %and those two estimators obtain significantly better mean square errors than OLS and AHR estimators.

{We also compare the average CPU time among ACLS, ACLS-h and ACLS-C. The results are summarized in Table \ref{table:Time}. On a laptop with a 2.9 GHz Core6 Duo processor and 32 GB RAM,  the average CPU time for ACLS and ACLS-h implemented by our R-package is computed for one random initialization per iteration.} 

\begin{figure}[t]	
	\includegraphics[width=\textwidth]{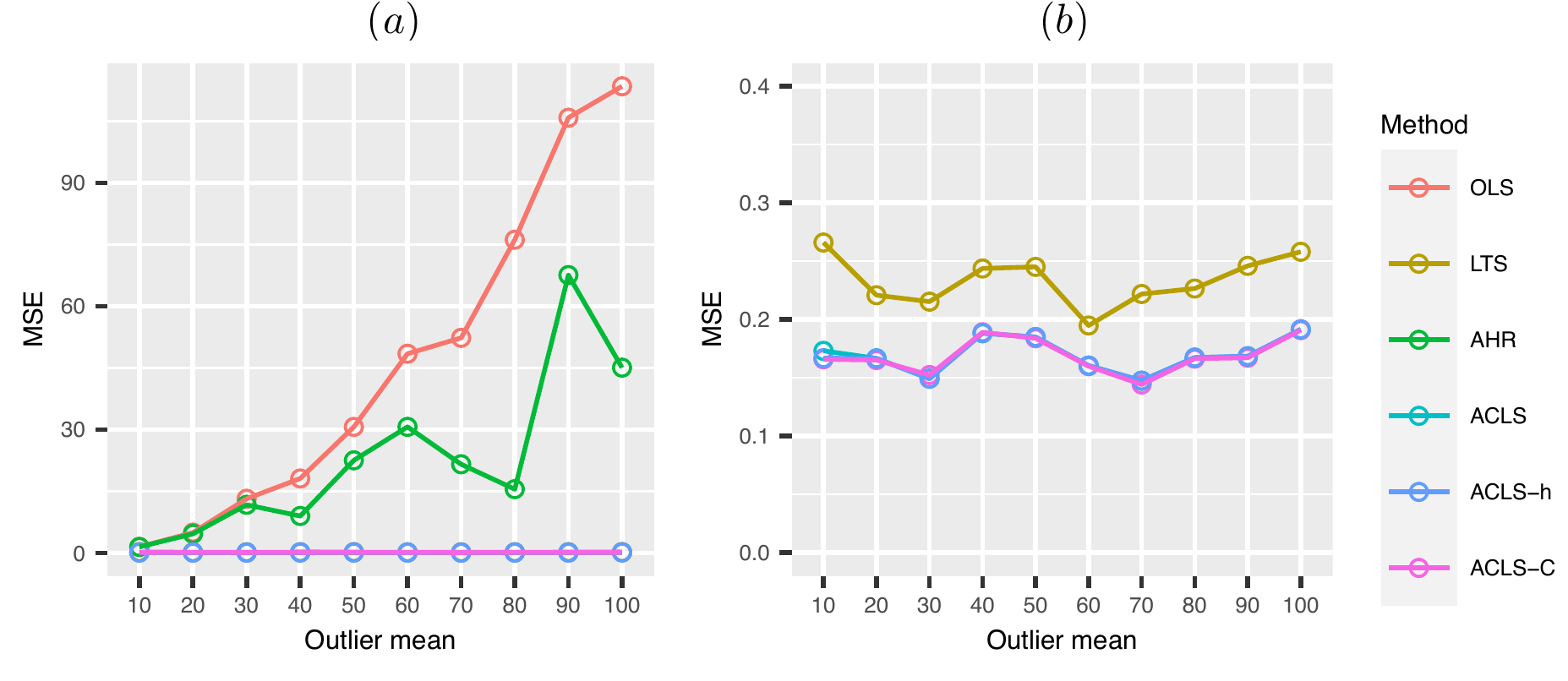}
	\caption{(a) Scenario 2:  median of mean square errors vs the outlier mean $a$ for all estimators. (b) Scenario 2: median of mean square errors vs the outlier mean $a$ for  LTS, ACLS, ACLS-h and ACLS-C.}\label{fig:1}
\end{figure}

\begin{figure}[h]	
	\includegraphics[width=\textwidth]{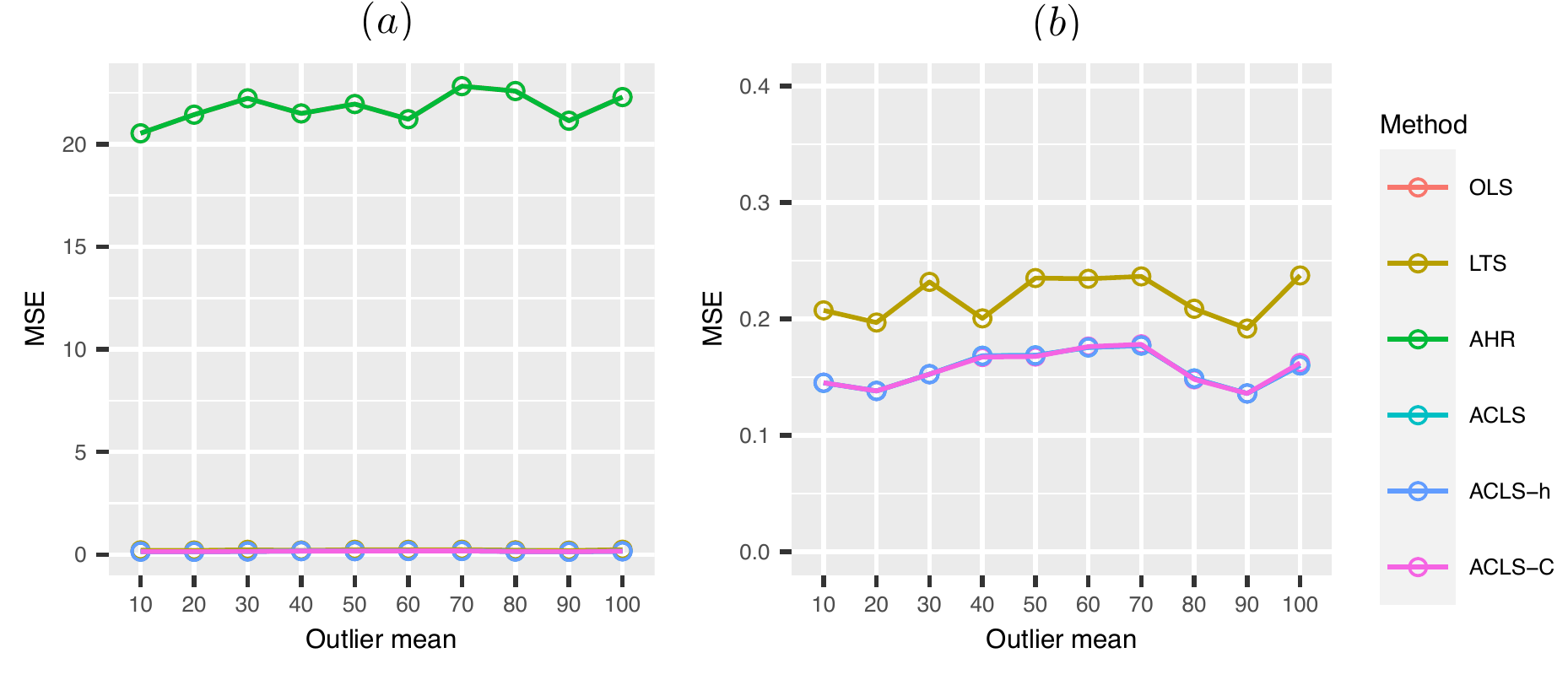}
	\caption{(a) Scenario 3:  median of mean square errors vs the outlier mean $a$ for all estimators. (b) Scenario 3: median of mean square errors vs the outlier mean $a$ for  LTS, ACLS, ACLS-h and ACLS-C.}\label{fig:2}
\end{figure}

\begin{table}[H]
	\center
	\begin{threeparttable}		
		\caption{{Simulation results for Scenario 1-3.}}
		\begin{tabular}[c]{lccccccc}
			
			& &OLS & AHR & LTS  &   ACLS&ACLS-h&  ACLS-C  \\
			MMSE & S1& 0.1302&  0.1302& 0.2315 & 0.1302&0.1303&0.1306\\ 
			& S2 ($a=50$) &30.6784 &  22.5611 & 0.2451 & 0.1848&0.1849 &0.1838\\ 
			&S3 ($a=100$) & 22.2955 &  22.2955& 0.2374  & 0.1601&0.1600 &0.1622\\ 
			SD  & S1 & 0.8453 & 0.8453  & 0.9417 & 0.8453&0.8453&0.8453  \\ 
			%Time& 0.0001  & 0.0000 & 0.0145 & 0.4678&5.2971&2.0134\\ 			
		\end{tabular}	
		\label{table:Scenario1}
		\begin{tablenotes}[flushleft]
			\footnotesize
			\item  MMSE,  median mean square error; S1, Scenario 1; S2, Scenario 2; S3, Scenario 3. 
		\end{tablenotes}
	\end{threeparttable}
\end{table}

\begin{table}[H]
	\center
	\begin{threeparttable}		
		\caption{{Average CPU time (in seconds) for Scenario 1-3.}}
		\begin{tabular}[c]{lccc}
			
			   & ACLS&ACLS-h&  ACLS-C  \\
		     S1&  0.1472&0.1480&2.0134\\ 
			 S2 ($a=50$) &0.4139  &0.4322&7.9087\\
			 S3 ($a=100$) &0.4811 &0.5590 &4.7098\\ 
		\end{tabular}	
	\label{table:Time}
		\begin{tablenotes}[flushleft]
			\footnotesize
			\item S1, Scenario 1; S2, Scenario 2; S3, Scenario 3.
		\end{tablenotes}
	\end{threeparttable}
\end{table}

%In Scenario 1, all estimators obtain satisfactory median of MSEs, while LTS obtain higher median MSE for LTS only use part of samples. 

%In Senario 2, our proposed estimator obtains quite satisfactory median of MSE compared with the median of MSE from \texttt{CPLEX}, while the performance of OLS estimator is poor and adaptive Huber's, LTS estimators obtain slightly  higher median of MSE than the median of MSE from our estimator.

%In Senario 3, performances of OLS and adaptive Huber's estimators are similarly poor. LTS estimator obtains slightly  higher median of MSE than the median of MSE from our estimator, which is very close to the global optimal solution from \texttt{CPLEX}.

\subsection{The Landscape of Loss Functions}\label{sec:6.2}

In this section, we visualize the adaptive capped squares loss in the univariate case with one covariate. As expected, the empirical loss is more convex-like as the sample size increases.
%This section shows that the loss function for the adaptive resistant regression becomes more convex-like as the sample size, and thus the problem scale, increases. We only examine the case of $d=1$. 
To be specific, we set the true coefficient $\beta^*=5$, and plot the empirical loss $\beta \mapsto  \sum_{i=1}^n \ell_{\tau}(y_i-x_i^{\T}\beta)/n$
%where the response $y_i$ is generated from \eqref{eq:truemodel} and we set 
with  $\tau=\sqrt{n}/\log\log n$ in the following four cases: 
\begin{enumerate}
	\item Case 1: There are no outliers as in Scenario 1;
	\item Case 2: There are only $y$-outliers as in Scenario 2 with $a=10$;
	\item Case 3: There are both $y$- and $x$-outliers occur as in Scenario 3 with $a=10$, and the contaminated sample size is fixed at $n_1=5$;
	\item Case 4: There are both $y$- and $x$-outliers occur as in Scenario 3 with $a=10$.
\end{enumerate}

%For both Case 1 and  Case 2, we plot empirical loss function versus $\beta$ for $n=50,$ $100,$ $200,$ and $400$. For case 3, we increaes the sample sizes until $n=6400$. %Specifically in case 2, we set $n_1=5$ and keep the ouliers the same in all $n$'s settings. For the 5 outliers, $x$-outliers are added from $\cN (10,1)$, and the error($\epsilon_i$) are from $\cN (10,1)$. In the setting $n=50$, we then generate 45 normal samples. As $n$ increases, we keep the former samples unchanged and add new normal samples. For example, in the setting $n=100$,  the first 50 samples are samples in the setting $n=50$, then we generate 50 more normal samples. Figure 3, 4, and 5 plots   the loss functions in case 1, 2 and 3, respectively.
%The results are presented in Figures \ref{fig:loss1} -- \ref{fig:loss3}. 

As shown by Figures~\ref{fig:loss1} -- \ref{fig:loss4}, the empirical capped squares loss function is locally strongly convex around the global optimum. In Case 1, as $n$ increases the convex region grows and the global optimum is getting closer to $\beta^*=5$. Figure~\ref{fig:loss3} shows that although there exists a local minimum when $n$ is small, the loss becomes more convex-like as $n$ increases, and finally this local minimum diminishes. An interesting phenomenon shown in Figure~\ref{fig:landscape} is that this local minimum corresponds to the OLS and the Huber estimator. In other words, both the OLS and the Huber estimators are sensitive to $x$-outliers. %Comparing with Case 1, this demonstrates that $x$-ourliers has    
In Cases 2 and 4 where the number of outliers also grows, there exists a local minimum around $\beta=1$ when $n$ is small. When $n$ is large,  the influence of outliers starts to prevail, and the global optimum is shifted to somewhere around $\beta=1$.
%Figure~\ref{fig:loss3} shows that, in case 3, there exists a local minimum at around $\beta=1$ when $n$ is small.  As $n$ increases, the influence of outliers starts to dominate, and the global optimum shifts to somewhere around $\beta=1$ from $5$. But 
In all four cases, the loss function becomes more convex so that any  first-order algorithm can identify the global optimum unless initialized very far away, namely, outside  the  ball  $\BB_2(\tau)$.  

\begin{figure}[H]
	
	\includegraphics[width=\textwidth]{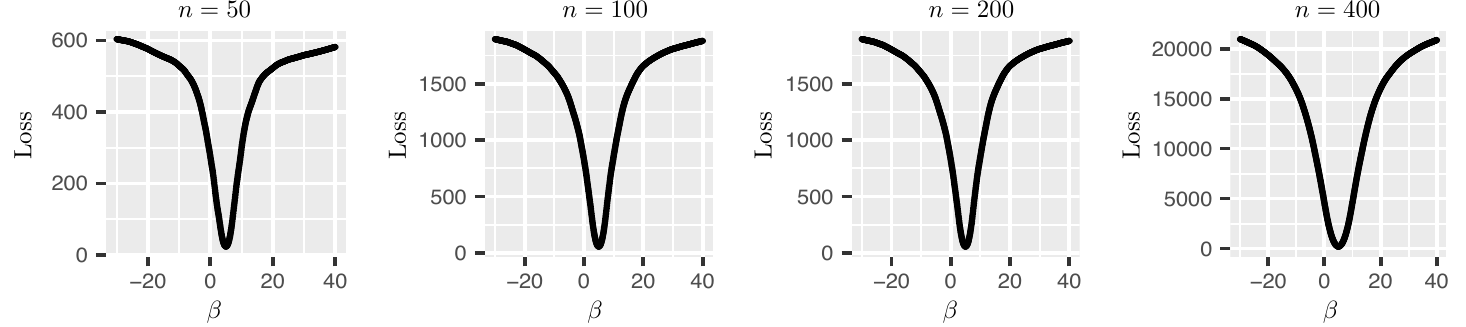}
	\caption{Case 1: Landscape for the adaptive capped least squares loss with (a) $n=50$, (b) $n=100$, (c) $n=200$, (d) $n=400$.}\label{fig:loss1}
\end{figure}

\begin{figure}[H]
	
	\includegraphics[width=\textwidth]{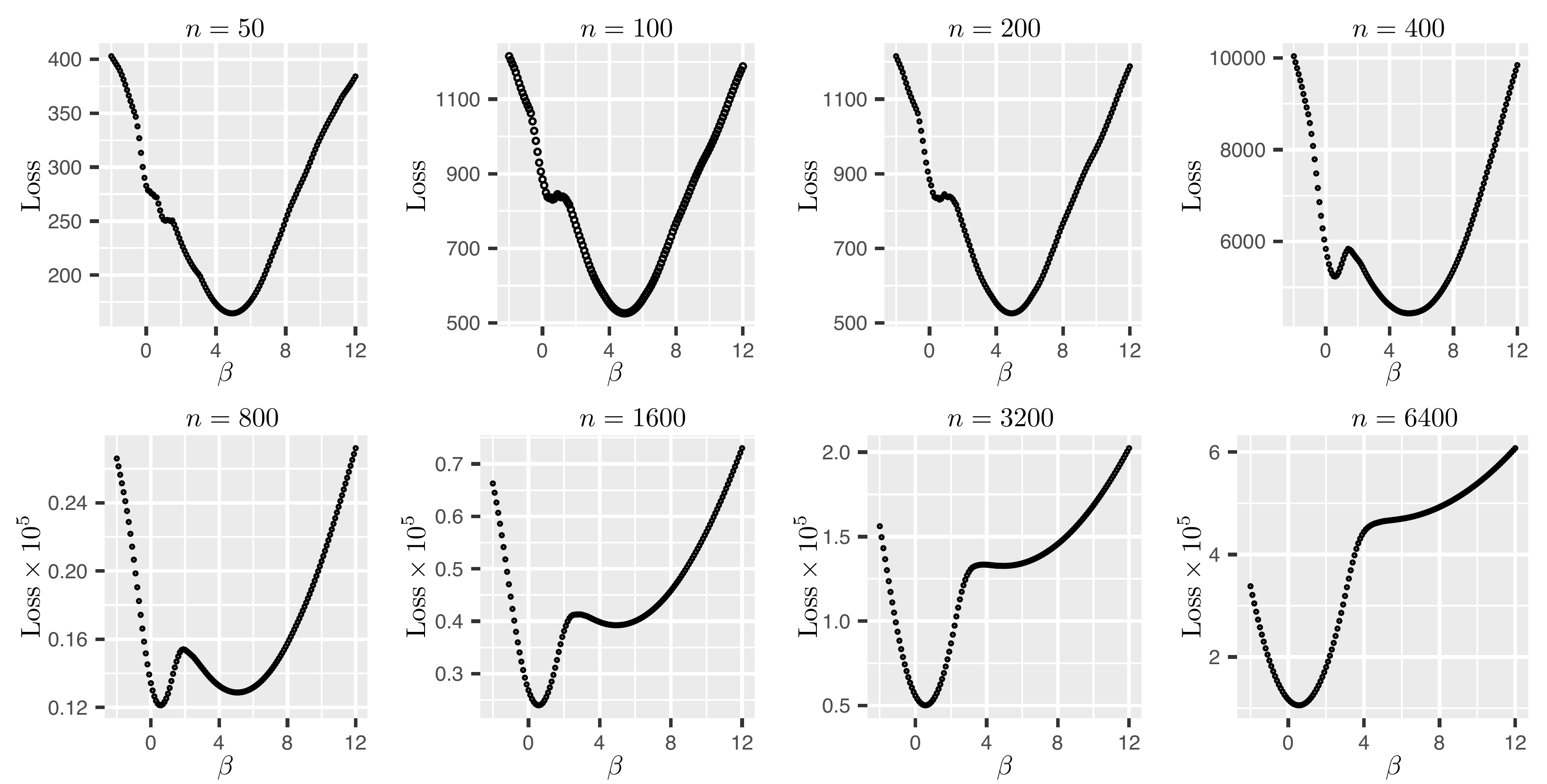}
	\caption{Case 2: Landscape for the adaptive capped least squares loss in Scenario 2. }\label{fig:loss2}
	
\end{figure}

\begin{figure}[H]
	
	\includegraphics[width=\textwidth]{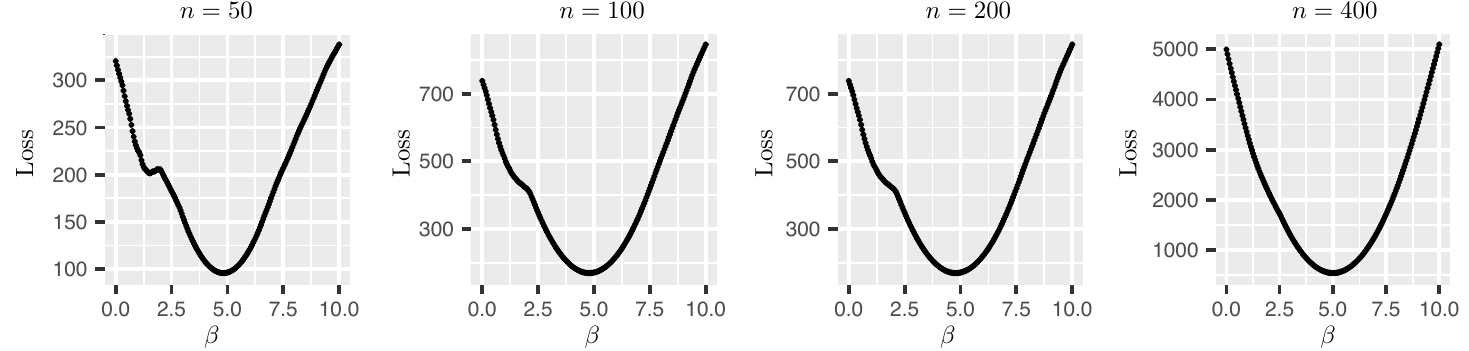}
	\caption{Case 3: Landscape for the adaptive capped least squares loss with (a) $n=50$, (b) $n=100$, (c) $n=200$, (d) $n=400$. }\label{fig:loss3}
	
\end{figure}
\begin{figure}[h]	
	\centering
	\includegraphics[width=0.5\textwidth]{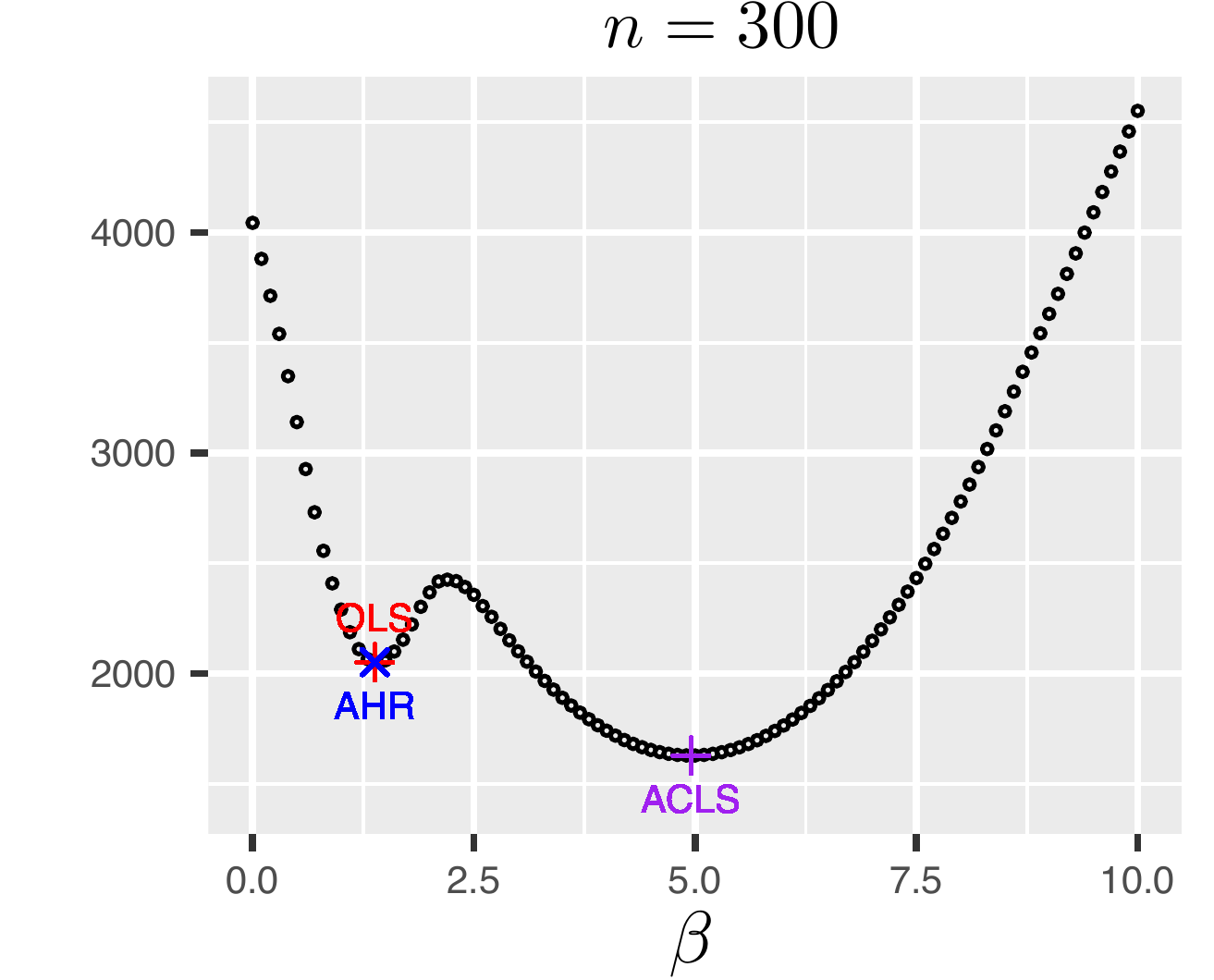}
	\caption{Case 4: OLS and AHR estimators fall in the local local minimum of the adaptive capped least squares loss with $n=300$.}\label{fig:landscape}
\end{figure}

\begin{figure}[H]
	
	\includegraphics[width=\textwidth]{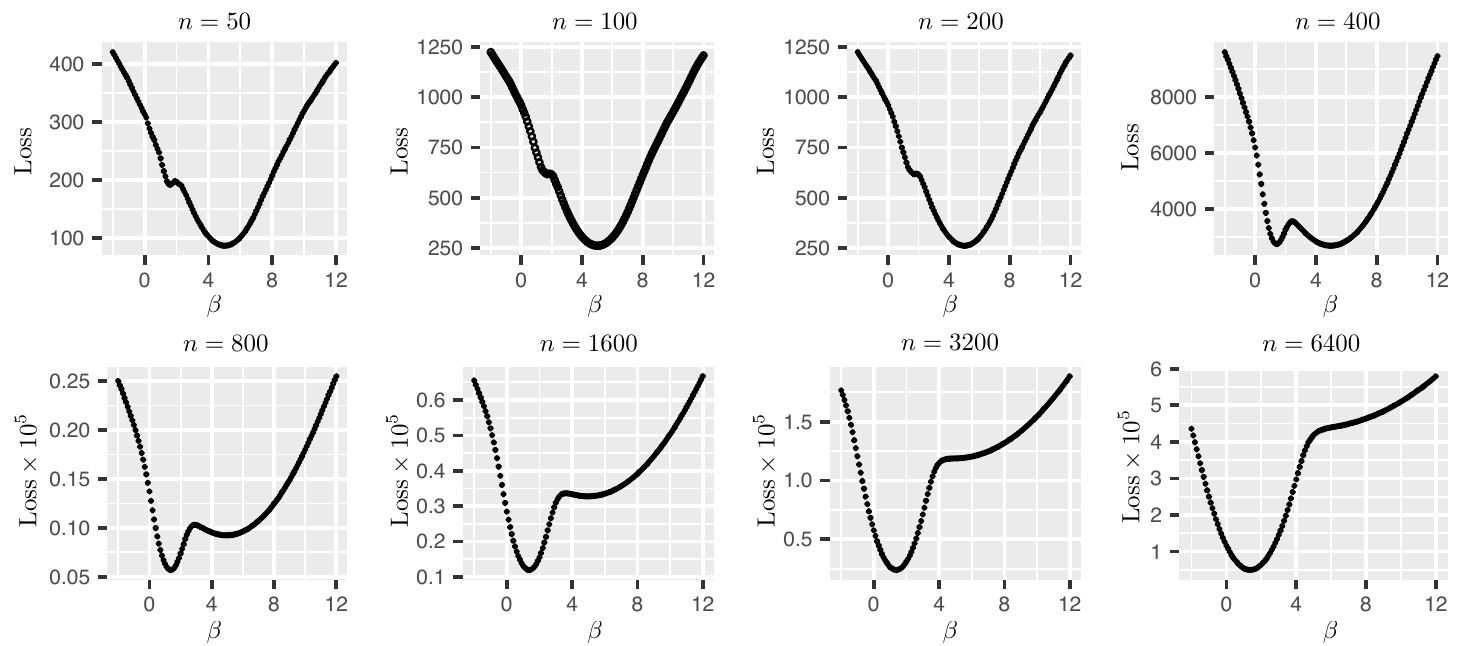}
	\caption{Case 4: Landscape for the adaptive capped least squares loss.}\label{fig:loss4}
	
\end{figure}

\section{Real Data Applications}\label{sec:7}
\subsection{NCI-60 Cancer Cell Lines}\label{sec:7.1}
We apply the proposed method to the NCI-60, a panel of 60 diverse human cancer cell lines. We use two  NCI-60 transcript profile datasets, the gene expression dataset and the protein profile dataset. Both datasets can be downloaded from http://discover.nci.nih.gov/cellminer/. The gene expression data were obtained on Affymetrix HG-U133(A-B) chips, %$\log_2$ transformed \scomment{really?} 
and normalized using the guanine
cytosine robust multi-array analysis \citep{wu2004A}. The protein profile data based on $162$ antibodies were obtained on reverse-phase protein lysate arrays. One observation had to be removed since all values were missing in the gene expression data, reducing the number of observations to $n=59$. We center all the protein and the gene expression variables to have mean zero. 

We pick the KRT19 antibody, which has the largest standard deviation among $162$ antibodies, as the dependent variable.  The KRT19 antibody, a type I keratin, also known as Cyfra 21-1, is encoded by the KRT19 gene. Due to its high sensitivity, the KRT19 antibody is the most used biomarker for the tumor cells disseminated in lymph nodes, peripheral blood, and bone marrow of breast cancer patients \citep{nakata2004serum}. \cite{sun2019adaptive} has identified seven genes, i.e., \textit{MT1E}, \textit{ARHGAP29}, \textit{MALL}, \textit{ANXA3}, \textit{MAL2}, \textit{BAMBI} and \textit{KRT19}, that were possibly associated with the KRT19 antibody. We use these seven genes as predictors.

\begin{figure}[H]
	\centering
	%	\subfigure[width=3in]{{"yhist"}.pdf}
	\includegraphics[width=2.3in]{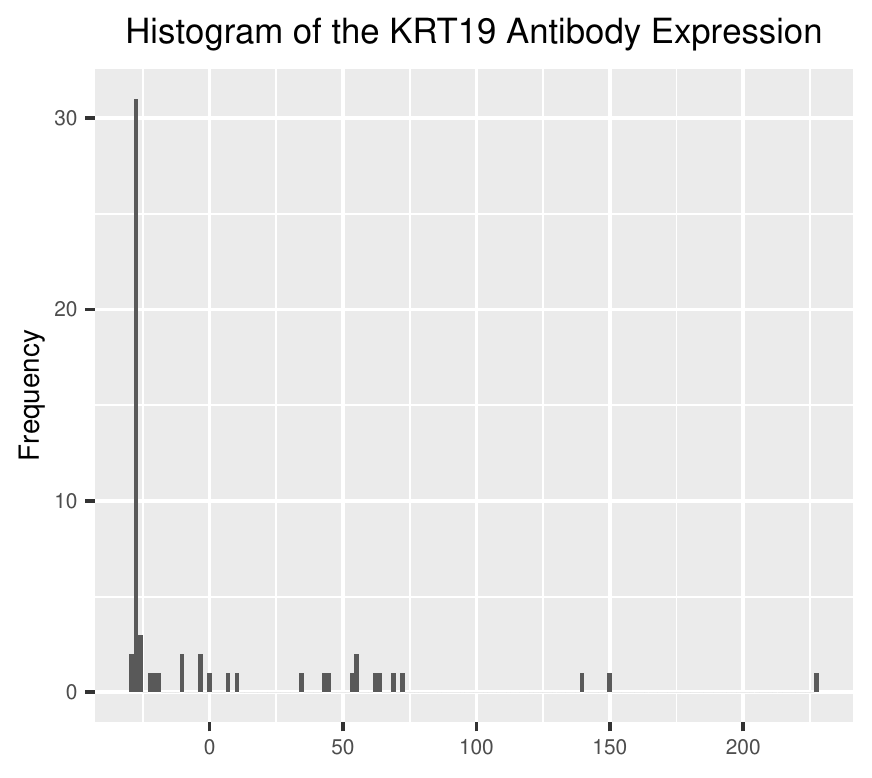}	
    \includegraphics[width=3.1in]{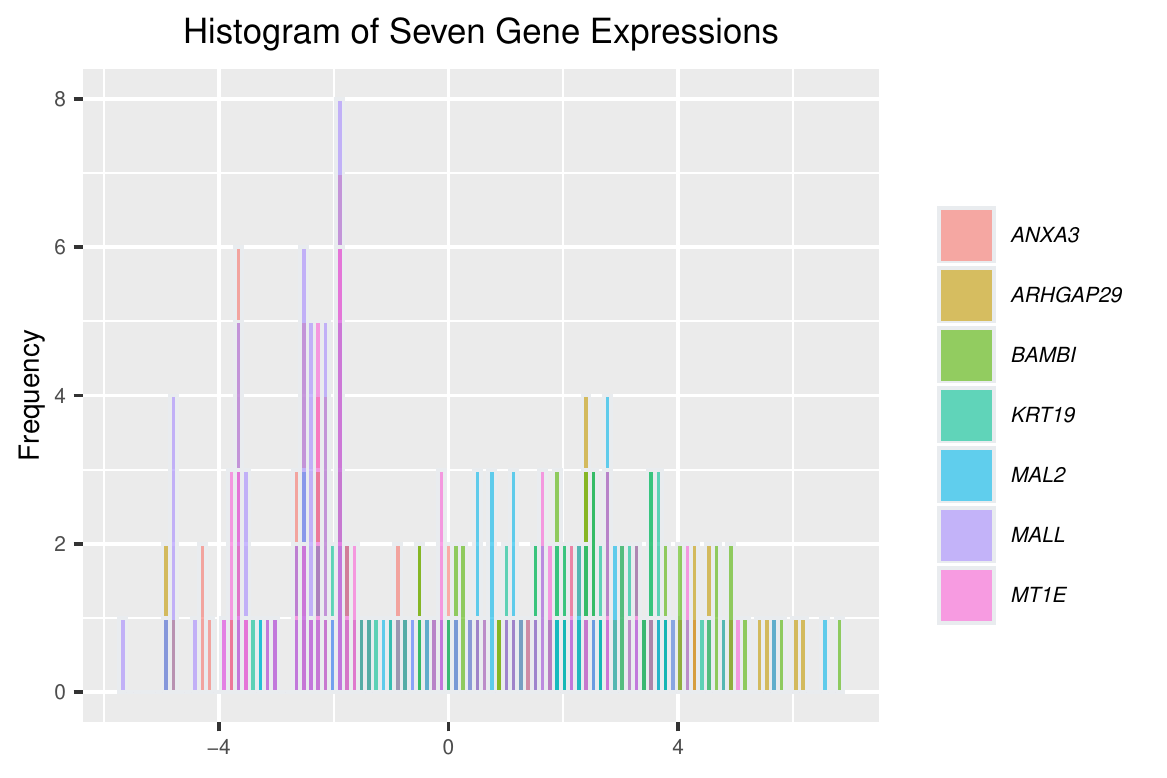}
	\caption{Histograms of the KRT19 antibody expression levels and gene expression levels.}\label{fig:hist}
\end{figure}

We first plot the histograms of the KRT19 antibody expression levels and seven gene expression levels in Figure \ref{fig:hist}. The histograms show that the distributions are asymmetric and there are possible outliers in the protein expression data. This could make results based on non-robust methods invalid. Therefore, we apply our methods to examine the predictive performance and statistical significance of the seven  genes on predicting the KRT19 antibody. We compare our methods with OLS, AHR and LTS, and report mean absolute prediction errors (MAPEs), the coefficient estimates and the corresponding $p$-values. The MAPE for  $\widehat\beta$ is defined as ${n}^{-1}\sum_{i=1}^n|y_i-x_i^{\T}\hat{\beta}|$.    For both our methods and AHR, we set the resistant and robustification parameter as $\tau=\sqrt{59}/\log\log 59$. For ACLS, we run Algorithm  \ref{alg:RGD} 100 times %randomly generate $\beta^0 \sim \text{Unif}(\mathbb{B}_2(\tau))$ and run gradient descent for 500 times 
and pick the estimator with the smallest adaptive resistant loss. To initialize ACLS-h, we run \texttt{CPLEX} on a subsample of size $\lceil0.5 \times 59 \rceil$, followed by gradient descent on the whole data. We then repeat this process $50$ times and pick the estimator with the smallest adaptive resistant loss.  We report the average performance of ACLS and ACLS-h from $20$ experiments.

Compared to  ACLS-C who achieved  loss $5.77$, ACLS and ACLS-h achieved losses of $5.77$ and  $5.86$, which are close to the global optimal loss value. This demonstrates the effectiveness of random initialization in ACLS and  subsampled initialization in ACLS-h.  % both achieved this global optimal loss. 
The MAPEs on the whole dataset are not representative due to the existence of possible outliers. {Thus, we calculate MAPEs from a benign subsample of data points, obtained by removing those data with $y$-outliers outside the $25$th and $75$th quantiles.} The results are collected  in Table \ref{table:7genespartest}, and  possibly imply that our methods and  LTS have favorable  predictive performance compared with OLS and AHR. %\scomment{combine the two tables.}

%\begin{table}[H]
%	\center
%	\begin{threeparttable}[flushleft]
%	\caption{{MAEs, MSEs and Loss of 6 methods, initials of ARR-h are obtained from $(n_{\text{ratio}},t)=(0.5,50)$.}}
%		\begin{tabular}[c]{ccccccc}
%			 & OLS & LTS & AHR &  ARR &ARR-h &ARR-C\\
%			MAPE&24.1305  & 28.1755 &	24.0502  &	22.4072 &27.9244&27.9248\\
%			%MSE& 1597.7514 &	3429.8961 &	1597.7514 &		2371.8980 & 3353.4560&3353.4843\\
%			Loss&11.9989 &	5.9647&	12.0789&5.8562 & 5.7659 & 5.7659 \\
%		\end{tabular}	
%		\label{table:7genes}
%	\end{threeparttable}
%\end{table}

\begin{table}[h]
	\center
	\begin{threeparttable}[flushleft]
		\caption{{MAPEs of six methods on a subsample of good data points obtained by removing $y$-outliers outside the $25$th and $75$th quantiles.}}
		\begin{tabular}[c]{ccccccccc}
			& OLS & LTS & AHR &  ACLS&ACLS-h &ACLS-C\\
			%MAPE & 18.2879 &	8.7986 &	18.0861	 &	8.9805 &8.7760&8.7766\\
			MAPE & 18.29 &	8.80 &	18.09	 &	8.78 &8.98&8.78\\
			%MAPE \scolor{$\log_2$}& 17.78&8.81&17.42&8.77 &9.05&8.77\\
			%MSE& 628.9747 &	406.5714 &	610.2459	&  383.2283 &403.6523& 403.6835\\
		\end{tabular}	
		\label{table:7genespartest}
	\end{threeparttable}
\end{table}

%\begin{table}[H]
%	\center
%	\begin{threeparttable}[flushleft]
%		\caption{{MAPEs of 6 methods on a subsample of good data points obtained by removing data having with residuals outside the $25$th and $75$th quantiles.}}
%		\begin{tabular}[c]{ccccccccc}
%			& OLS & LTS & AHR &  ACLS&ACLS-h &ACLS-C\\
%			MAPE& 25.67&36.01&35.71&35.71 &34.50&35.71\\
%			MAPE \scolor{$\log_2$}& 17.78&8.81&17.42&8.77 &9.05&8.77\\
%			%MSE& 628.9747 &	406.5714 &	610.2459	&  383.2283 &403.6523& 403.6835\\
%		\end{tabular}	
%		\label{table:7genespartest}
%	\end{threeparttable}
%\end{table}

Table \ref{table:pvalue} collects the estimates and the corresponding $p$-values. For ACLS and ACLS-h, the median estimates and the corresponding $p$-values, out of $20$ experiments, are reported.   %collects the estimates and $p$-values from OLS, LTS, ARR, ARR-h and ARR-C. 
The $p$-values are computed using the asymptotic normal distributions, according to Theorem \ref{thm:asyn}. To compute the $p$-values, %For the unknown covariance $\sigma^2\Sigma_\tau$,   %Following Theorem \ref{thm:asyn}, %the finite sample variance of $\sqrt{n}(\hat{\beta}_{\tau}-\beta^*)$ is ${\sigma}^2\Sigma_\tau^{-1}$. 
we use the finite sample estimators of $\sigma^2$ and $\Sigma_\tau$,
\$
\widehat{\sigma}^2=\frac{1}{n_{e}}\sum_i (y_i-x_i^{\T}\hat{\beta}_{\tau})^21(|y_i-x_i^{\T}\hat{\beta}_{\tau}|\leq \tau), ~
\widehat{\Sigma}_{\tau}=\frac{1}{n_{e}}\sum_i 1(|y_i-x_i^{\T}\hat{\beta}_{\tau}|\leq \tau)x_i x_i^{\T},
\$
 where $n_{e}$ is the number of samples whose absolute residuals are smaller than $\tau$.

The results indicate that \textit{MALL}, \textit{ANXA3}, \textit{MAL2}, \textit{BAMBI} and \textit{KRT19} are statistically significant in predicting the KRT19 antibody expression based on the $p$-values of ACLS, ACLS-h and ACLS-C. %\rcomment{The $p$-value does not indicate \textit{KRT19} are significant?} 
\textit{ANXA3} is shown to be upregulated in breast cancer tissues and correlated with poor overall survival \citep{du2018downregulation}. It has been reported that \textit{MAL2} promotes proliferation, migration, and invasion through regulating epithelial-mesenchymal transition in breast cancer cell lines \citep{2018MAL2}. 
%It has been reported that knocked down \textit{MAL2} could decrease the ability of proliferation, migration, and invasion of breast cancer cell lines \citep{2018MAL2}.
\cite{shangguan2012inhibition} has reported that \textit{BAMBI} transduction disrupted the cytokine network mediating the interaction between mesenchymal stem cells and breast cancer cells. Consequently, \textit{BAMBI} transduction abolished protumor effects of bone marrow mesenchymal stem cells in vitro and in an orthotopic breast cancer xenograft model, and instead significantly inhibited growth and metastasis of {coinoculated cancer}.
All three implementations of our method have identified \textit{MALL} as a statistically significant gene in predicting KRT19 antibody expression and thus is possibly associated with breast cancer. This has not been reported by previous studies and is worth further investigation. %the effect of ARHGAP29 and MALL on breast cancer remains unclear and is worth further investigation. \scomment{How about KRT19 in the reference row and the last column?}

%\begin{table}[H]
%	\center
%	\begin{threeparttable}[flushleft]
%		\caption{{We reprot the estimates and the corresponding $p$-values calculated by $6$ methods. }}
%		\begin{tabular}[c]{lccccccc}
%			Genes&\rotatebox{90}{\textit{MT1E}}&\rotatebox{90}{\textit{ARHGAP29}}&\rotatebox{90}{\textit{MALL}}&\rotatebox{90}{\textit{ANXA3}}&\rotatebox{90}{\textit{MAL2}}&\rotatebox{90}{\textit{BAMBI}}&\rotatebox{90}{\textit{KRT19}}\\
%			$\beta_{\text{OLS}}$ & 4.22 & -0.32 & -0.13 &  1.48&5.36 &-1.88&4.89\\
%			$p\text{ value}$& 0.17 &	0.89 &	0.96 &	0.61 &0.06&0.42&0.32\\
%			$\beta_{\text{AHR}}$ & 4.25 & -0.42 & 0.05 &  1.37&5.19 &-1.96&4.96\\
%			$p\text{ value}$& $0.00^*$ &	0.12 &	0.86 &	$0.00^*$ &$0.00^*$&$0.00^*$&$0.00^*$\\
%			$\beta_{\text{LTS}}$ & -0.12 & 0.01 & -0.01 &  0.03&-0.14 &-0.11&0.03\\
%			$p\text{ value}$& 0.12 &	0.91 &	0.91 &	0.59 &0.25&$0.03^*$&0.84\\
%			$\beta_{\text{ARR}}$& -0.11 & 0.04 &  7.00&0.03 &0.86&-0.12&0.11\\
%			$p\text{ value}$& 0.08 &	0.34 &	$0.00^*$&	0.62 &$0.00^*$&$0.01^*$&0.30\\
%			$\beta_{\text{ARR-h}}$& -0.09& 0.05 & -0.51 &  0.21&0.47 &-0.16&-0.32\\
%			$p$ \text{value}& 0.32 &	0.42 &	$0.00^*$ &	$0.01^*$ &$0.00^*$&$0.02^*$&$0.03^*$\\
%		    $\beta_{\text{ARR-C}}$& -0.09& 0.05 & -0.51 &  0.21&0.47 &-0.16&-0.32\\
%			$p\text{ value}$& 0.32 &	0.42 &	$0.00^*$ &	$0.01^*$ &$0.00^*$&$0.02^*$&$0.03^*$\\
%			Ref &\cite{tai2003differential}&&&\cite{du2018downregulation}&&\cite{shangguan2012inhibition}
%		\end{tabular}	
%		\label{table:pvalue}
%	\end{threeparttable}
%\end{table}

\begin{table}[H]
	\center
	\begin{threeparttable}[flushleft]
		\caption{{Estimates and the corresponding $p$-values calculated by six methods, where $p$-values are computed using the asymptotic normal distributions.}}
		\begin{tabular}[c]{lccccccc}
			Genes&\textit{MT1E}&\textit{ARHGAP29}&\textit{MALL}&\textit{ANXA3}&\textit{MAL2}&\textit{BAMBI}&\textit{KRT19}\\
			$\widehat{\beta}_{\text{OLS}}$ & 4.22 & -0.32 & -0.13 &  1.48&5.36 &-1.88&4.89\\
			$p\text{ value}$& 0.17 &	0.89 &	0.96 &	0.61 &0.06&0.42&0.32\\
			$\widehat{\beta}_{\text{AHR}}$ & 4.25 & -0.42 & 0.05 &  1.37&5.19 &-1.96&4.96\\
			$p\text{ value}$& $0.00^*$ &	0.12 &	0.86 &	$0.00^*$ &$0.00^*$&$0.00^*$&$0.00^*$\\
			$\widehat{\beta}_{\text{LTS}}$ & -0.12 & 0.01 & -0.01 &  0.03&-0.14 &-0.11&0.03\\
			$p\text{ value}$& 0.12 &	0.91 &	0.91 &	0.59 &0.25&$0.03^*$&0.84\\
			$\widehat{\beta}_{\text{ACLS}}$& -0.09& 0.05 & -0.51 &  0.21&0.47 &-0.16&-0.32\\
			$p$ \text{value}& 0.32 &	0.42 &	$0.00^*$ &	$0.01^*$ &$0.00^*$&$0.02^*$&$0.03^*$\\
			$\widehat{\beta}_{\text{ACLS-h}}$& -0.09& 0.05 & -0.51 &  0.21&0.47 &-0.16&-0.32\\
			$p$ \text{value}& 0.32 &	0.42 &	$0.00^*$ &	$0.01^*$ &$0.00^*$&$0.02^*$&$0.03^*$\\
			$\widehat{\beta}_{\text{ACLS-C}}$& -0.09& 0.05 & -0.51 &  0.21&0.47 &-0.16&-0.32\\
			$p\text{ value}$& 0.32 &	0.42 &	$0.00^*$ &	$0.01^*$ &$0.00^*$&$0.02^*$&$0.03^*$\\
			Reference &&&&$1$&$2$&$3$
		\end{tabular}	
		\label{table:pvalue}
		\begin{tablenotes}[flushleft]
			\footnotesize
			\item $1$, \cite{du2018downregulation}; $2$, \cite{2018MAL2}; $3$, \cite{shangguan2012inhibition}.
		\end{tablenotes}
	\end{threeparttable}
\end{table}

\subsection{Background Recovery in Video Surveillance}\label{sec:7.2}

We examine the proposed method on the video surveillance dataset from \cite{li2004statistical}. This dataset consists of $n=1546$ frames  from a lobby in an office building with illumination changes (switching on/off lights). All frames have resolution $p=128 \times 160=20,480$. {We first convert all frames to gray scale and then  stack each frame as a column of the matrix $(y_1,\ldots, y_n) \in \mathbb{R}^{20,480 \times 1546}$.}  The stationary background can often be modeled by a low rank matrix while the moving foreground items are often treated as outliers. We apply our methods to recovery the stationary background. % while minimizing the effect of outiers by using the capped least squares regression. 

%Video data are often modeled by a low-rank term (stationary background) plus an outlier term (moving foreground items \scolor{or illumination changes}). This section applies the proposed  method for background recovery.
%Denote the columns of $Y$ by $y_1,\ldots,y_n$ with $n=1546$. 
To utilize the low rank structure of the stationary background, we model each background frame, denoted by $z_i$, as
\$
z_i=m+ Us_i +\epsilon_i, ~1\leq i \leq n,
\$
where $m$ is the mean vector, $U\in\RR^{p\times q}$ is the orthonormal basis matrix spanning a $q$-dimensional space, $s_i$ is the coefficient vector, and $\epsilon_i$ is the noise. The observed video frames  $y_i$ can be seen as contaminated versions of $z_i$. To introduce resistance to outliers, we use the adaptive capped least squares regression to estimate the unknowns by solving the following optimization problem
\#\label{eq:ARR_original}
{ \big\{\hat m, \hat U, \{\hat s_i\}_{i=1}^n\big\} =\argmin_{m,U,\{s_i\}_{i=1}^n} \frac{1}{n}\sum_{i=1}^n\ell_\tau (\|y_i-m-Us_i\|_2),\quad\textnormal{s.t. } U^{\T} U=I_q},
\#
which  is then solved by an alternating minimization algorithm. We collect the details of the algorithm in the supplementary material.

{In this data example, we choose $q=10$, $\tau=\hat\sigma  \cdot \sqrt n (\log\log n)^{-1}$ with 
\$
\widehat\sigma= 1.4826\times \median\left\{ y_{ij}-\median(y_{ij}: 1\leq j\leq p): 1\leq i \leq n, 1\leq j\leq p\right\},
\$ 
where $y_i=(y_{i1},\cdots,y_{ip})^{\T}, i=1,\cdots n$.}
%(details and corresponding Algorithm \ref{alg:AM} are shown in Appendix). 
%In this application, we take low-rank parameter $q=10$ and 
We have $\hat\sigma=0.1512$ and $\tau=2.9811$. 
The stationary background and moving foreground items of each frame $i$ are then constructed as $\hat{m}+\hat{U}\hat{s}_i$ and $(y_i-\hat{m}-\hat{U}\hat{s}_i)\cdot 1(\|y_i-\hat{m}-\hat{U}\hat{s}_i\|>\tau)$, respectively.

%\scolor{the low-rank bilinear decomposition model is $y_i=m+Us_i+e_i$, $i=1,...,n$, where $m \in \mathbb{R}^p$ is a location vector, matrix $U \in \mathbb{R}^{p \times q}$ has $q$ orthogonal unit vectors as columns, $\{s_i\}_{i=1}^n \in \mathbb{R}^q$ are principal components, and $\{e_i\}_{i=1}^n$ are i.i.d. random errors.}\scomment{This is such an abrupt logic change. I can not follow.} One standard way to solve this model is to adopt the OLS method to estimate $\{m,U,\{s_i\}\}$. We know from previous section that the OLS method is highly sensitive to the outliers while the proposed method is robust, thus we apply the adaptive resistant regression to this model. We set the low-rank $q=10$ and the adaptive resistant parameter $\tau=0.1$ to decompose $Y$.

%We apply an alternating minimization (AM) algorithm with random initializations to iteratively solve the convex relaxation of problem \eqref{eq:ARR_original} (details and corresponding Algorithm \ref{alg:AM} are shown in Appendix). 

%On a laptop with a 2.9 GHz Core6 Duo processor and 32 GB RAM, it takes 6.8184 seconds for the alternating minimization algorithm implemented by Matlab to converge after nine iterations.

We pick three frames to recover, whose results are collected  in Figure \ref{fig:video}.  The first two rows show the results when the observed frames have a moving person and two standing still people respectively. The third row collects the results for a static scene when some lights are off. For comparison purpose, we also collect the results from ordinary least squares regression,  where we replace the capped least squares loss in \eqref{eq:ARR_original} by the quadratic loss.  For all cases, our proposed method is able to  %captures the moving person as outliers 
recover the stationary background without being affected by the moving person, the static people or the illumination change,  while the ordinary least squares regression fails in every case.  %For the second row, the OLS method has trouble recognizing the static foreground while our method successfully extract the background. In the third example, our method can remove the shadow in the dark environment while the OLS method cannot. 

\iffalse
We also collected the results tested on different low-rank parameter $q=1,2,3$ on Appendix. Because our dataset consists of only $20\%$ images with people in the background, the ordinary least squares method also performs well in the simple model cases with $q=1,2$. However, as the model becomes more complicated, for example, $q \geq 3$, the ordinary least squares method fails to recover the background while our method could recover the stationary background in every cases. 
\fi

%\scolor{describe how we calcuate the resutls. Compare the results obtained by $y_i$ minus outliers and $m+Us_i$ in Rmarkdown.}\scolor{Current point.}

%\begin{figure}[H]
%	\centering
%	\includegraphics[width=5.5in]{video.pdf} 
%	\caption{{Video surveillance: (a) Original Frames. (b)-(c) Background extraction and outliers from robust regression conducted on adaptive resistant regression model. (d)-(e) Background extraction and outliers from the OLS method.}}\label{fig:video}
%\end{figure}

\begin{figure}[H]
	\centering
	\includegraphics[width=5.5in]{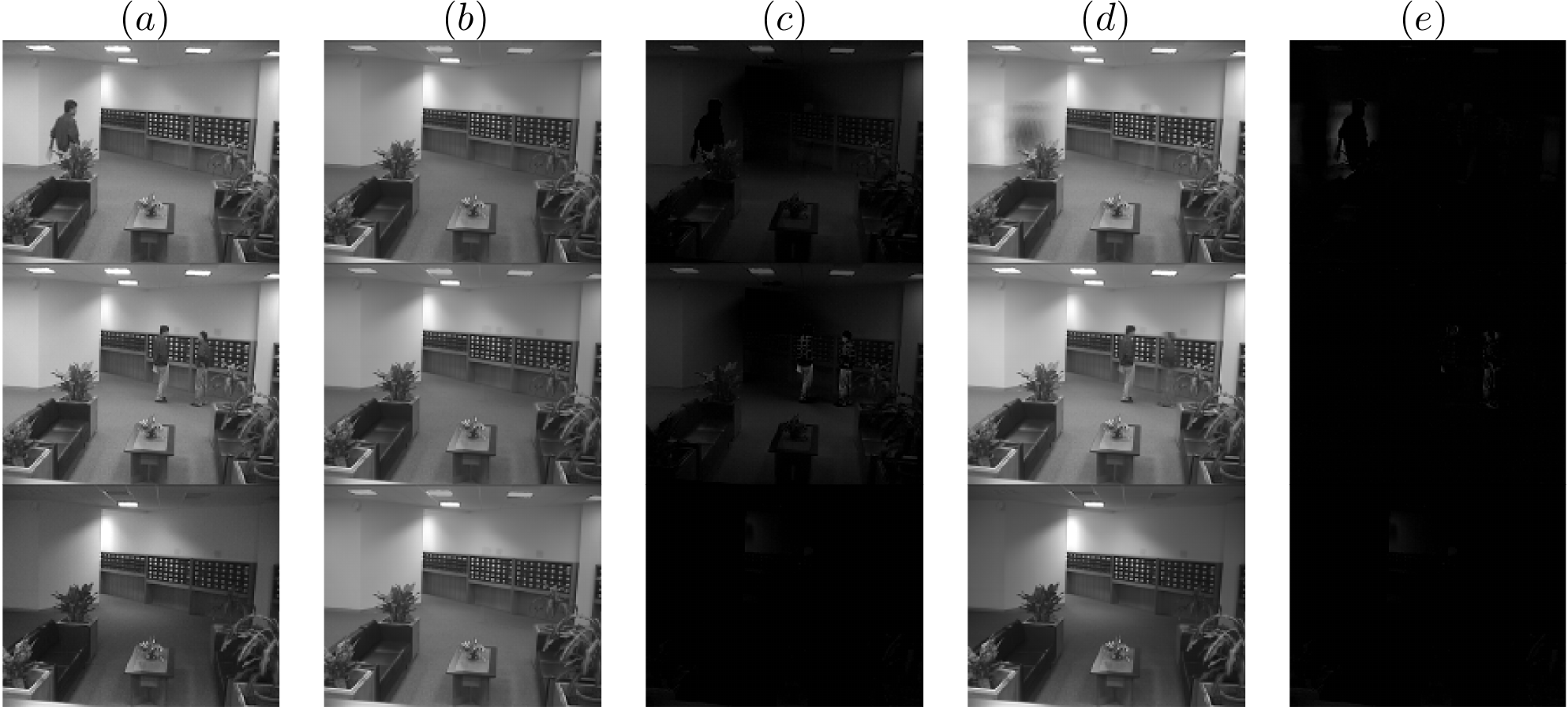} 
	\caption{{Video surveillance: (a) Original Frames. (b)-(c) Background extraction and outliers from robust regression conducted on adaptive capped least squares regression model. (d)-(e) Background extraction and outliers from the ordinary least squares method.}}\label{fig:video}
\end{figure}

\subsection{Blind Image Inpainting}

This section applies the proposed method to blind image inpainting, whose goal is to repair damaged pixels of a given image without knowing the damaged positions as a priori.  %recovering the corrupted image without specifying the exact mask for contaminated areas. 
 We first divide the  damaged image, possibly  after normalization, into $p$ small squared patches, consisting of $\sqrt{n} \times \sqrt{n}$ pixels. We then stack each patch as a column of the signal matrix $Y=(y_1,\ldots,y_p) \in \mathbb{R}^{n \times p}$. The damaged pixels of the given image are often treated as outliers in the signal matrix, that is
 \begin{eqnarray*}
 	y_{ij}=\begin{cases}
 		u_{ij}, & \textnormal{if}\  y_{ij} \text{ is clean}; \\     
 		u_{ij}+o_{ij}, & \textnormal{if}\  y_{ij} \text{ is a damaged pixel.}
 	\end{cases}   
 \end{eqnarray*}
 where $u_{ij}$ is an undamaged pixel and $o_{ij}$ is an outlier. %\scomment{we used $\epsilon_i$ previously as stochastic noises? Can you check the appendix to make sure all are right?}

 Our aim is to apply the proposed method to filtrate  corrupted pixels and recover the image using rest undamaged ones.  We say an undamaged pixel $u_i$ has a sparse representation over a dictionary $D \in \mathbb{R}^{n \times m}$, if we could find a sparse vector $\alpha_i \in \mathbb{R}^{m \times 1}$ such than $u_i \approx D\alpha_i$.  The dictionary $D$ is pre-learned from an undamaged picture, %matrix $Z=[z_1,\cdots,z_m] \in \mathbb{R}^{n\times m}$ 
 and it consists of $m$ basis vectors, referred to as atoms. 
 After learning $D$, we solve $\alpha=(\alpha_1,\ldots, \alpha_p)$ by optimize the empirical ACLS loss with Lasso  penalties 
 \begin{equation}\label{eq:ARR_inpainting}
 \hat{\alpha}=\argmin_{\alpha \in \mathbb{R}^{m \times p}} \sum_{i=1}^p\ell_{\tau}(y_i-D\alpha_i)+\lambda\sum_{i=1}^p \|\alpha_i\|_1,
 \end{equation}
 where $\lambda$ is a regularization parameter.

We test the performance of \eqref{eq:ARR_inpainting} on the gray scale Lena image with resolution $d=256 \times 256=58,564$. We first normalize all pixel values,  8-bit integers $\in [0,255]$,  in the image matrix to real numbers $\in [0,1]$. The columns of the signal matrix is then formed by taking squared patches of size $\sqrt{n} \times \sqrt{n} =15 \times 15$  in a sliding manner. We learn the dictionary $D \in \mathbb{R}^{225 \times 256}$ from the undamaged image signal matrix as in  \cite{2009Online}. 
Figure \ref{fig:dictionary} shows the learned dictionary with 256 atoms.

\begin{figure}[t]
	\centering
	\includegraphics[width=3.5in]{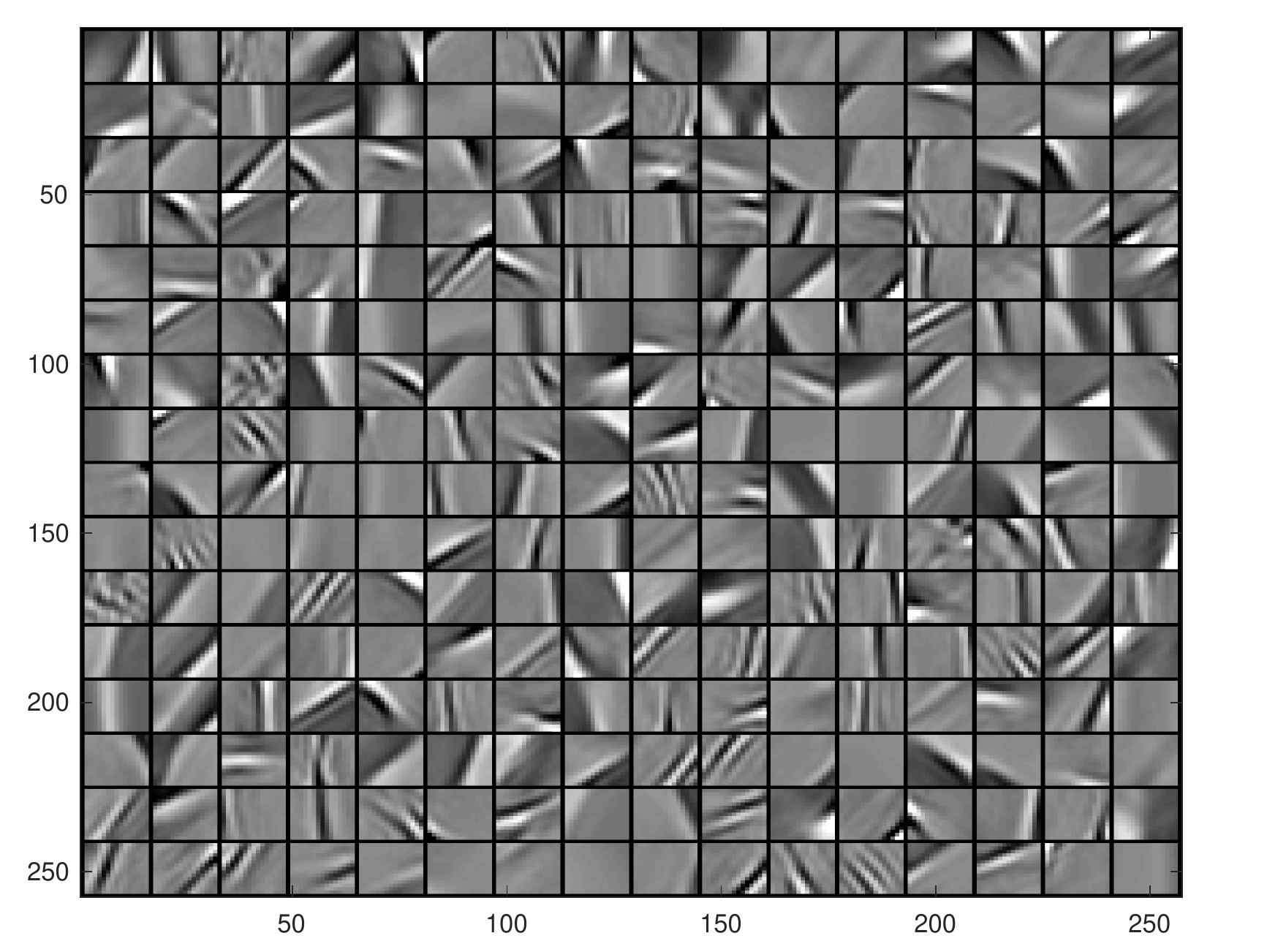} 
	\caption{A dictionary with 256 atoms learned on a natural image.}\label{fig:dictionary}
\end{figure}

		Similar to the video surveillance study, we use alternating minimization algorithm to solve \eqref{eq:ARR_inpainting} with details collected in the supplementary material. The top row of Figure \ref{fig:img} presents  two instances for blind image inpainting.  In the first instance, we contaminate the original Lena image by changing $300$ pixels value to $0$s where $300$ damaged pixels are uniformly distributed across the image shown in the top-left panel of Figure \ref{fig:img}. In the second instance, we contaminate the original image with a manually added curve shown in the top-right panel of Figure \ref{fig:img}. %\scomment{How did you construct good images?}
		
%		\scolor{For both cases, we pick  $\tau=c\cdot\hat{\sigma}\cdot\sqrt{n}(\log{\log{n}})^{-1}$, where $c$ is some constant. In case 1 $\hat{\sigma}=0.7465$ and we pick $c=0.06$ such that $\tau=0.4$; in case 2, $\hat{\sigma}=0.7466$ and we pick $c=0.06$ such that $\tau=0.3$, where $\hat \sigma=\text{median}\{\|y_i-D\alpha_i\|_2^2/p: 1\leq i \leq n\}$ and $\alpha$ is obtained from our method using fixed $\tau=0.4$ in case 1, $\tau=0.3$ in case 2.}

%	
%	\scolor{Let $\bar {y_i}= p^{-1}\sum_{j=1}^p y_{ij}$.  We use the adaptive resistant parameter $\tau=c\cdot\hat\sigma  \cdot \sqrt n (\log\log n)^{-1}$, where $\widehat{\sigma}= \text{median}{\{ \sqrt{\|y_i-\bar y_i\|_2^2/p}: 1\leq i \leq n \}}$, $c$ is some known constant. In case 1, $\widehat{\sigma}= \text{median}{\{ \sqrt{\|y_i-\bar y_i\|_2^2/p}: 1\leq i \leq n \}}=0.1270$, the adaptive resistant parameter $\tau=c\cdot\hat\sigma  \cdot \sqrt n (\log\log n)^{-1}=0.4$ where $c=0.35$. In case 2, $\widehat{\sigma}= \text{median}{\{ \sqrt{\|y_i-\bar y_i\|_2^2/p}: 1\leq i \leq n \}}=0.1269$, the adaptive resistant parameter $\tau=c\cdot\hat\sigma  \cdot \sqrt n (\log\log n)^{-1}=0.3$ where $c=0.27$. 
	We then apply our method to recover. We use $\tau=c\cdot\hat\sigma  \cdot \sqrt n (\log\log n)^{-1}$, where $\widehat\sigma= 1.4826\times \text{median}\{ |y_{ij}-\text{median}(y_{i}:1\leq i\leq p)|: 1\leq i \leq p, 1\leq j \leq n\}$ and $c$ is some constant. In the first instance, $\widehat\sigma=0.0349,$ resulting roughly $\tau=0.4$ with $c=0.77$. In the second instance, $\widehat\sigma=0.0349,$
	resulting  $\tau=0.3$ with $c=0.97$. We recover the signal matrix using $Y\odot(1-\hat{\Delta})+D\hat{\alpha}\odot\hat{\Delta}$, where $\odot$ denotes the Hadamard product, $\hat{\Delta}=[\hat{\delta}_{ij}]$ and $\hat{\delta}_{ij}=1(|y_{ij}-[D\hat{\alpha}_i]_j|>\tau), i=1,\cdots,p,j=1,\cdots,n$.
	
	The restoration results are presented in the bottom row of
    Figure \ref{fig:img}. For both instances, our proposed method is able to repair the damaged image. We calculate the peak signal to noise ratio (PSNR) of estimates for two cases, which can be used to quantitatively evaluate the quality of the restoration results. The PSNR is defined as,
    \$
    \text{PSNR}(\hat{x},x)=10\log_{10}\frac{255^2}{\frac{1}{d}\sum_{i=1}^{\sqrt{d}}\sum_{j=1}^{\sqrt{d}}(\hat{x}_{ij}-x_{ij})^2},
    \$
    where $x_{ij}$ is the intensity value that from 0 (black) to 255 (white) at $\{i,j\}$ pixel of the damaged image, $\hat{x}_{ij}$ is the intensity value at $\{i,j\}$ pixel of the recovered image. The PSNR of estimates for two case 1 and case 2 are 47.4590 and 43.7476, respectively.

		\begin{figure}[t]
		\centering
		\includegraphics[width=3.5in]{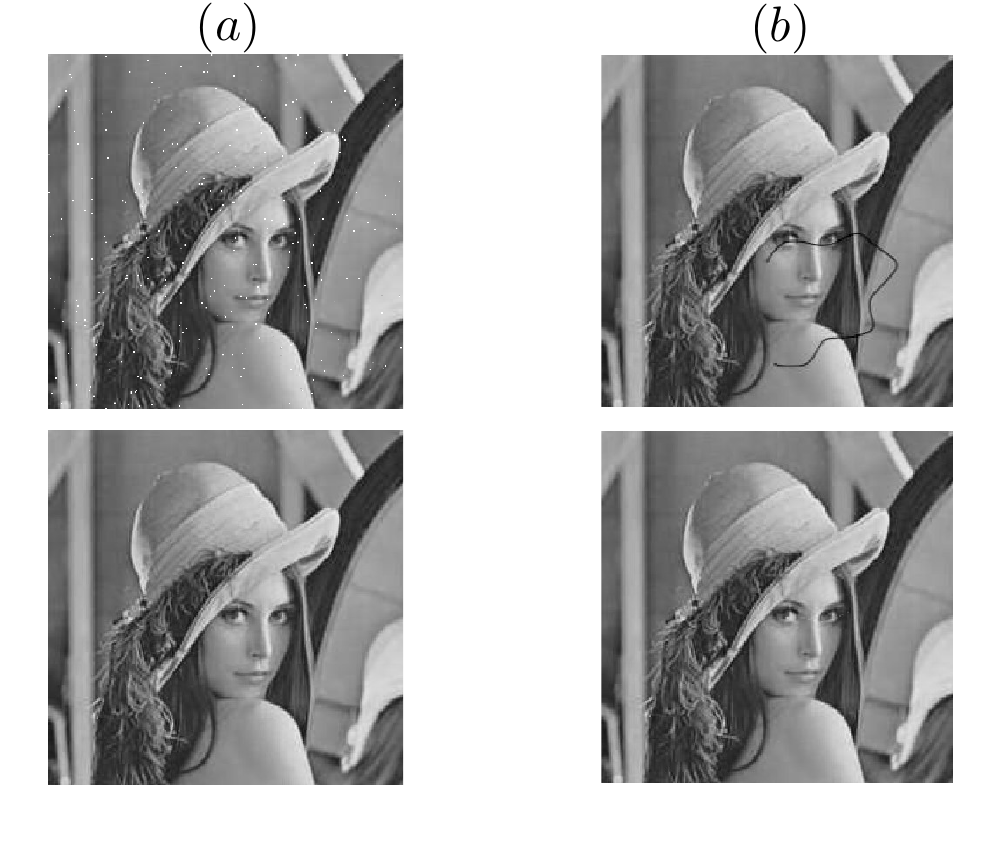} 
		\caption{ \textit{Blind image inpainting: (a) Lena image with random noise and its restoration. (b) Lena image with manually added curve and its restoration.}}\label{fig:img}
	\end{figure}

\section{Discussion}\label{sec:8}
This paper proposes the capped least squares regression with an adaptive resistance parameter, hence the name, adaptive capped least squares regression. The key observation is, by taking the resistant parameter to be data-dependent and at a proper order, the proposed estimator automatically achieves high accuracy and high efficiency. Surprisingly, at the same time, it does not lose resistance: the proposed estimator achieves the maximum breakdown point of $1/2$ asymptotically. Computationally, we formulate the problem as a quadratic mixed integer programming problem which can readily solved by \texttt{CPLEX}. To speed up the computation, we propose a randomized gradient descent algorithm.   Numerical examples lend strong support to our methodology and theory.

\bibliographystyle{ims}
\bibliography{acls}

%\section*{Acknowledgement}
%The authors would like to thank Professor Douglas G. Simpson for helpful discussions, and for pointing out two  related references.  

\newpage
\appendix 
	%Notations for Appendices
	\renewcommand{\theequation}{S.\arabic{equation}}
	\renewcommand{\thetable}{S.\arabic{table}}
	\renewcommand{\thefigure}{S.\arabic{figure}}
	\renewcommand{\thesection}{S.\arabic{section}}
	\renewcommand{\thelemma}{S.\arabic{lemma}}
	\numberwithin{equation}{section}  %%%numberwithin section
	
	\vspace{30pt}
	\noindent{\bf \LARGE Appendix}
	\vspace{10pt}
	\setcounter{page}{1}

    Throughout the appendix, we assume all suprema of functions are measurable; otherwise we shall use the essential supremum instead.

	\iffalse
	{\sf Part \Rom{2}:} {\color{red} This part is wrong when $ a\neq h $. }
	Now, we will prove the other side. We will show for any positive integer $ 1\leq h\leq n$, there exists a $\tau>0 $ such that $ \hat \beta _{\tau} $ is the minimizer of (\ref{LTS}).
	
	We will first show that for any $ \tau^2 \in [r_{(h)}^2(\hat \beta _h), r_{(h+1)}^2(\hat \beta _h)) $, $ \hat \beta _{\tau} $ is the minimizer of (\ref{LTS}) with truncation number $ h $. If $ \hat \beta _{\tau}$ is not  the minimizer of (\ref{LTS}), one has 
	\begin{eqnarray}\label{contradict2}
	Q_h(\hat \beta _{\tau})>Q_h(\hat\beta _h)
	\end{eqnarray}
	We will then prove this leads to a contradiction.

	Define that $a=\sum_{i=1}^n 1(r_{i, \tau}^2(\hat \beta _{\tau})\leq \tau^2)$.  First, if we assume that $a=h$, we would have
	\begin{eqnarray*}
		\ell_{\tau}(\hat \beta _{\tau})&=&\sum_{i=1}^h r^2_{(i), \tau}(\hat \beta _{\tau})+(n-h)\tau^2\\
		&\leq& \sum_{i=1}^h r^2_{(i), h}(\hat \beta _h)+(n-h)\tau^2,
	\end{eqnarray*}
	where the first inequality holds because $ \hat \beta _{\tau} $ is the minimizer of $\ell_{\tau}(\beta )$. Therefore, one has $ Q_h(\hat \beta _{\tau})=\sum_{i=1}^h r^2_{(i), \tau}(\hat \beta _{\tau})\leq  \sum_{i=1}^h r^2_{(i), h}(\hat \beta _h)=Q_h(\hat\beta _h) $.

	We proceed to show that $a$ must be equal to $h$ with our choice of $\tau$. If $a<\tau$, then using the similar argument in Part \Rom{1}, we shall obtain that   $\widehat \beta _a$ is the minimizer of problem \eqref{eq:flathuber}, and thus 
	\$
	XXX. %Not ~working. 
	\$
	\fi

	\section{Proofs for Breakdown Points}
	
	This section collects the proofs for Proposition \ref{pro:huber} and Theorem \ref{thm:2}. 
	
\begin{proof}[Proof of Proposition \ref{pro:huber}]
	The proof of this proposition is taken  from Section 5.13.1 of \cite{maronna2019robust}. For  completeness, we collect it here. Let $g(u)=\sign (x) \left(\tau\wedge |x| \right)$ be the gradient function of the Huber loss function. Then the Huber estimator verifies 
	\#\label{prop:huber.eq1}
g(y_1-x_1^\T \widehat\beta)x_1 + \sum_{i=2}^n g(y_i -x_i^\T\widehat\beta)x_i=0. 
	\#
	Let $y_1$ and $x_1$ tend to infinity in such a way that $y_1/\|x_1\|_2\rightarrow \infty$. If $\widehat\beta$ remained bounded, we would have
	\$
	y_1-x_1^\T \widehat\beta\geq y_1-\|x_1\|_2\|\widehat\beta\|_2=\|x_1\|_2\left (\frac{y_1}{\|x_1\|_2}-\|\widehat\beta\|_2\right)\rightarrow \infty. 
	\$
	Since $g$ is nondecreasing, $g(y_1-x_1^\T \widehat\beta)$ would tend to $\sup g>0$, and hence the first term in \eqref{prop:huber.eq1} would tend to infinity, while the sum would remain bounded. This is a contraction. Thus $\widehat\beta$ has to be unbounded. This finishes the proof. 
	\end{proof}

The proof of Proposition \ref{pro:huber} relies on the estimating equations \eqref{prop:huber.eq1}. To prove the breakdown point for ACLS estimator, we take a more general  route by directly looking at the losses since the capped least squares loss is  not differential. 
	
	\begin{proof}[Proof of Theorem \ref{thm:2}]
		
		%{\sf Case I: Contamination in the response space}. We start with the analysis of break down property in the fixed design case, in comparison with the {\it moving design} setting.  
		%Without loss of generality, we may assume that $\widehat\beta (\cZ)=0.$ 

Let $m=n\varepsilon^*(\widehat\beta , \cZ)$. For every $k\in\NN$, there exists a $\cZ^k\in \cP_m(\cZ)$ such that $ \|\widehat\beta (\cZ^k) \|_2>k$. For simplicity, we write $\beta ^k=\widehat\beta (\cZ^k).$ Without loss of generality, we may assume the first $n-m$ samples in $\cZ^k$ are uncontaminated. And we have $\|\beta ^k\|_2\rightarrow \infty$. Per the compactness of the unit sphere in $\RR^d$, we may also suppose that $\theta^k=\beta ^k/\|\beta ^k\|_2$ converges to some point $\theta^\infty$, passing to a subsequence otherwise. Then we  have 
		\$
		\cL_n\big(\beta ^k,\cZ^k\big)\leq \cL_n\big(\beta^*, \cZ^k\big),
		\$
		or equivalently 
		\$%\label{thm2:eq1}
		\sum_{i=1}^{n-m}\ell_\tau\big(y_i-x_i ^\T\beta ^k\big)+\sum_{i=n-m+1}^n\ell_\tau\big(y_i^k-(x_i^k)^\T\beta ^k\big)\leq \sum_{i=1}^{n-m}\ell_\tau(\epsilon_i)+\sum_{i=n-m+1}^n\ell_\tau(\epsilon_i^k),
		\$
where $\epsilon_i^k=y_i^k - (x_i^k)^\T \beta^*.$
		%Let $\btheta_k=\beta $
		Noting that $0 \leq \ell_\tau(u)\leq \tau^2/2$ for all $u$, the above inequality further reduces to % implies that 
		\$
		\sum_{i=1}^{n-m}\ell_\tau\big(y_i-\|\beta ^k\|_2x_i ^\T\theta^k\big)\leq (n-m) \times\frac{1}{n-m}\sum_{i=1}^{n-m}\ell_\tau(\epsilon_i)+m\tau^2/2. 
		\$
		Under the general position assumption, there are at most $d-1$ samples such that $x _i^\T\theta^\infty=0$. Denote the collection of such observations by  $\cE$ with $|\cE|\leq d-1$. Hence, by letting  $k\rightarrow\infty$, we obtain 
		\$
		(n-m-|\cE|)\tau^2/2+\sum_{i\in\cE}\ell_\tau(y_i)\leq (n-m)c_r+m\tau^2/2,
		\$
		which in turn implies 
		\$
		m\geq \frac{n(1-2c_r/\tau^2)-d+1}{2(1-c_r/\tau^2)}.
		\$
		Taking $\tau\rightarrow \infty$ acquires
		\$
		m\geq \left\lceil\frac{n-d+1}{2}\right\rceil=\left\lfloor\frac{n-d+2}{2}\right\rfloor,
		\$
		as desired. 
	\end{proof}

\section{Proof of Lemma \ref{lemma:sc}}
   
\begin{proof}[Proof of Lemma \ref{lemma:sc}]   
Suppose $y=\mu+\epsilon$. Suppose $|\mu-\mu_\tau^*|\leq \tau/4$, which holds for sufficiency large $\tau$ because $\mu_\tau^*\rightarrow \mu$.
For the loss function $\ell_\tau(x)$, we have 
\$
&\ell_\tau(x+u)-\ell_\tau(x)-\frac{1}{2}u^21(|x|\leq \tau)\\
&=\frac{1}{2}x^2\left(1(|x+u|\leq \tau)-1(|x|\leq \tau) \right) + ux1(|x+u|\leq \tau)\\
& \qquad + {\frac{1}{2}\tau^2\left(1(|x+u|> \tau)-1(|x|> \tau) \right)}  + \frac{1}{2}u^2\left(1(|x+u|> \tau)-1(|x|> \tau) \right)\\
&\geq -\frac{1}{2}x^21(\tau-|u|\leq |x|\leq \tau) -\frac{1}{2}\tau^2 {1(\tau\leq  |x|\leq \tau+|u|)}-\frac{1}{2}u^21(\tau-|u|\leq  |x|\leq \tau)\\
& \qquad + ux1(|x+u|\leq \tau),
\$
which, by taking $x=y-\mu_\tau^*$, $u=\Delta$, and $b=\mu-\mu_\tau^*$, implies 
\#\label{proof:sc:1}
&\ell_\tau(y-\mu_\tau^*+\Delta)-\ell_\tau(y-\mu_\tau^*)-\frac{1}{2}\Delta^21(|y-\mu_\tau^*|\leq \tau) \notag\\
&\geq -\frac{1}{2}(y-\mu_\tau^*)^21(\tau-|\Delta|\leq |y-\mu_\tau^*|\leq \tau) -\frac{1}{2}\tau^2 1(\tau\leq |y-\mu_\tau^*|\leq \tau+|\Delta|) \notag \\
& \qquad -\frac{1}{2}\Delta^21(\tau-|\Delta|\leq |y-\mu_\tau^*|\leq \tau) + \Delta (y-\mu_\tau^*)1(|y-\mu_\tau^*+u|\leq \tau) \notag\\
&= -\frac{1}{2}(y-\mu_\tau^*)^21(\tau-|\Delta|\leq |y-\mu_\tau^*|\leq \tau) -\frac{1}{2}\tau^2 1(\tau\leq |y-\mu_\tau^*|\leq \tau+|\Delta|) \notag\\
& \qquad -\frac{1}{2}\Delta^21(\tau-|\Delta|\leq |y-\mu_\tau^*|\leq \tau) + \Delta b(|y-\mu|\leq \tau) \notag\\
&\qquad + \epsilon\Delta 1(|y-\mu|\leq \tau).
\#
{Suppose $|\Delta|,|b|\leq \tau/4$ such that $|\Delta|+|b|\leq \delta_0 \leq \tau/2$}. Now because $\{\tau-|\Delta|\leq |y-\mu_\tau^*|\leq \tau\}\subseteq \{\tau-|\Delta|-|b|\leq |\epsilon|\leq \tau+|b|\}$, we have 
\$
&\EE (y-\mu_\tau^*)^21(\tau-|\Delta|\leq |y-\mu_\tau^*|\leq \tau)\\
&\leq 2\EE\left\{(\epsilon^2+b^2)1(\tau-|\Delta|-|b|\leq |\epsilon|\leq \tau+|b|)\right\}\\
&\leq \frac{2^{3+\eta}L_0(|\Delta|+|b|)}{\tau^{3+\eta}}+\frac{2^{5+\eta}L_0b^2(|\Delta|+|b|)}{\tau^{5+\eta}}.
\$
Similarly, we have 
\begin{gather*}
\EE \tau^2 1(\tau\leq |y-\mu_\tau^*|\leq \tau +|\Delta|) \leq \frac{L_0(|\Delta|+|b|)}{\tau^{3+\eta}},\\
\EE \Delta^2 1(\tau-|\Delta|\leq |y-\mu_\tau^*|\leq \tau) \leq \frac{2^{5+\eta}L_0\Delta^2(|\Delta|+|b|)}{\tau^{5+\eta}},\\
\EE \Delta b 1(\tau-|\Delta|\leq |y-\mu|\leq \tau) \leq \frac{2^{5+\eta}L_0|\Delta b| (|\Delta|+|b|)}{\tau^{5+\eta}}.
\end{gather*}
For the last term in the right hand side of \eqref{proof:sc:1},  we have 
\$
\EE \epsilon\Delta 1(|y-\mu|\leq \tau)=-\EE \epsilon\Delta 1(|y-\mu|> \tau)\leq \frac{|\Delta|m_{4+\eta}}{\tau^{3+\eta}}.
\$
Let $\PP_n f= n^{-1}\sum_{i=1}^n f_i$. Therefore, taking expectation on both sides of the above inequality and summing it over $i$ yield
\$
&\frac{1}{n}\sum_{i=1}^n\left(\EE\ell_\tau(y_i-x_i^\T\beta_\tau^*-x_i^\T(\beta-\beta_\tau^*))-\EE\ell_\tau(y_i-x_i^\T\beta_\tau^*)\right)\\
&\qquad -\frac{1}{2}(\beta-\beta^*)^\T\frac{1}{n}\sum_{i=1}^n x_ix_i^\T \EE 1(|y_i-x_i^\T\beta_\tau^*|\leq \tau) (\beta-\beta^*) \\
&\geq -C\bigg\{\frac{\PP_n \|x\|_2 (\|\beta-\beta_\tau^*\|_2+\|\beta^*-\beta_\tau^*\|_2)}{\tau^{3+\eta}}+\frac{\PP_n\|x\|_2\|\beta-\beta_\tau^*\|_2}{\tau^{3+\eta}}\\
&\qquad \qquad + \frac{\PP_n \|x\|_2^3\|\beta^*-\beta^*_\tau\|_2^2 (\|\beta-\beta_\tau^*\|_2+\|\beta^*-\beta_\tau^*\|_2)}{\tau^{5+\eta}} \\
&\qquad \qquad + \frac{\PP_n\|x\|_2^3\|\beta-\beta_\tau^*\|_2^2(\|\beta-\beta_\tau^*\|_2+\|\beta^*-\beta_\tau^*\|_2)}{\tau^{5+\eta}}\\
&\qquad \qquad + \frac{\PP_n\|x\|_2^3\|\beta-\beta_\tau^*\|_2\|\beta^*-\beta_\tau^*\|_2(\|\beta-\beta_\tau^*\|_2+\|\beta^*-\beta_\tau^*\|_2)}{\tau^{5+\eta}}\bigg\},
\$
where $C$ only depends on $\eta, \, L_0$, and $m_{4+\eta}$. Now since $|x_i^\T(\beta-\beta_\tau^*)|\leq\tau/4$, we have 
\$
\EE 1(|y_i-x_i^\T\beta_\tau^*|\leq \tau)\geq \EE 1(|\epsilon_i|\leq 3\tau/4),
\$
and thus $\rho_\tau=\lambda_{\min}\left(n^{-1}\sum_{i=1}^nx_ix_i^\T 1(|\epsilon_i|\leq 3\tau/4)\right)>0.$
Now if $\sqrt{d}/\tau^{3+\eta}\leq c \|\beta-\beta_\tau^*\|_2$ and $\sqrt{d^3}/\tau^{5+\eta}\leq c \|\beta-\beta_\tau^*\|_2^{-1}$ for $c$, depending only on $C$ and $\rho_\tau$, sufficiently small, we have 
\$
\frac{1}{n}\sum_{i=1}^n\left(\EE\ell_\tau(y_i-x_i^\T\beta_\tau^*-x_i^\T(\beta-\beta_\tau^*))-\EE\ell_\tau(y_i-x_i^\T\beta_\tau^*)\right)\geq \frac{1}{4}\rho_\tau \|\beta-\beta^*\|_2,
\$
as desired. 

\end{proof}

	\section{Proof for Consistency}
	To prove consistency, we first establish the uniform law of large numbers for the empirical loss, namely, $\sup_{\beta\in \RR^d} \big| (1/n)\sum_{i=1}^n \ell_\tau (y_i -x_i^\T\beta) -  (1/n)\sum_{i=1}^n\EE\ell_\tau(y_i- x_i^\T \beta)\big|$.  
	\begin{theorem}\label{thm:unif}
		With probability at least $1-\delta$, we have  
		\$
		\sup_{\beta\in \RR^d} \left|\frac{1}{n}\sum_{i=1}^n\ell_\tau (y_i -x_i^\T\beta) -\frac{1}{n}\sum_{i=1}^n\EE\ell_\tau(y_i- x_i^\T \beta)\right|\leq  C\tau^2\left(\sqrt{\frac{d}{n}}+ \ \sqrt{\frac{\log {1}/{\delta}}{n}}\right),
		\$
		where $C>0$ is a universal constant. 
	\end{theorem}

	\begin{proof}[Proof of Theorem \ref{thm:unif}]
		We first fix $\tau$ and thus the loss function $\ell_\tau(\cdot)$. Let 
		\$
		\cV=\left\{f_\beta(x,\epsilon)= \epsilon -x^\T( \beta - \beta^*) : \beta\in \RR^d\right\}
		\$
		be a class of functions $\RR^{d} \times \RR \to \RR$.
		Given a function $f \in \cV$, we write  $\PP_n f=n^{-1}\sum_{i=1}^n f(z_i)$ and $\PP f= {n}^{-1}\sum_{i=1}^n\EE f(z_i)$, where $z_i=(x_i, \epsilon_i)$. %'s are i.i.d. random vectors. 
		Moreover, write $\left\|\PP_n -\PP\right\|_\cF=\sup_{f\in \cF}|\PP_nf-\PP f|.$ Under this notation, we have 
		\$
		\sup_{\beta\in \RR^d} \left|\frac{1}{n}\sum_{i=1}^n\ell_\tau (y_i -x_i^\T\beta) -\frac{1}{n}\sum_{i=1}^n\EE\ell_\tau(y_i- x_i^\T \beta)\right|=\sup_{g\in \ell_{\tau}\circ \cV}\left|\PP_n g-\PP g\right|. 
		\$

         Applying Lemma \ref{lemma:1} to the function class $\ell_\tau\circ\cV$, which is $\ell_\tau \circ\cV$ is $\tau^2/2$-bounded, yields that 
		\#\label{eq:unif}
		\left\|\PP_n -\PP\right\|_\cF\leq \EE\|\PP_n-\PP\|_\cF+ \frac{\tau^2}{2 } \sqrt{\frac{2\log(1/\delta)}{n}}, 
		\#
		with probability at least $1-\delta$. We then upper bound the mean $\EE\|\PP_n-\PP\|_\cF$  using Dudley's entropy integral bound as summarized in Lemma \ref{lemma:3}.

		To proceed, we need to control the metric entropy of $\ell_\tau\circ\cV$.  Note first that the functional class $\cV$ is a $d$-dimensional vector space. By Lemma 2.6.15 of \cite{wellner1996weak}, $\cV$ has a VC-subgraph dimension at most $\nu=\vc(\cV)\leq d+2.$  Here the notion of VC-subgraph dimension is defined in \cite{wellner1996weak}. We rewrite the loss function $\ell_\tau(u)$ as
		\$
		\ell_\tau(u)=\ell_\tau^+(u)+\ell_\tau^-(u), ~\text{where}~\ell_\tau^+(u)=\ell_\tau(u)1(u\geq 0) \text{ and } \ell_\tau^-(u)=\ell_\tau(u)1(u<0). 
		\$
		Then,  $\ell_\tau(y-x^\T\beta)$ can be written as the sum of two composite loss functions
		\$
		\ell_\tau(y-x^\T\beta)=\ell_\tau^+\circ f+\ell_\tau^-\circ f, ~\text{where}~f = f_\beta \in \cV . 
		\$
		Since $\ell_\tau^+$ is an increasing function, by Lemma 2.6.15 of \cite{wellner1996weak}, $\ell_\tau^+\circ \cV$  has  VC-subgraph dimension of at most $d+2$. Similarly, $\ell_\tau^-\circ \cV$ has VC-subgraph dimension of most $d+2$. %Now because $\ell_\tau\circ \cV$ is a subspace of $\ell_\tau^+\circ \cV+\ell_\tau^-\circ \cV$, it has VC index at most $d+2$ as well. 
	   Applying Lemma 2.6.7 in \cite{wellner1996weak} yields that, for any probability measure $Q$ and $0<s<1$,
		\$
		N(s\tau^2/2, \ell_\tau^+\circ\cV, L_2(Q))&\leq K (d+2) (16e)^{d+2}\left(\frac{1}{s}\right)^{2(d+1)},\\
		N(s\tau^2/2, \ell_\tau^-\circ\cV, L_2(Q))&\leq K (d+2) (16e)^{d+2}\left(\frac{1}{s}\right)^{2(d+1)} ,
		\$
		where $K>0$ is a universal constant. By elementary calculations, 
		\$
		N(s\tau^2, \ell_\tau^+\circ\cV+\ell_\tau^-\circ\cV, L_2(Q))&\leq N(s\tau^2/2, \ell_\tau^+\circ\cV, L_2(Q))\times N(s\tau^2/2, \ell_\tau^+\circ\cV, L_2(Q))\\
		&\leq \left\{K (d+2) (16e)^{d+2}\left(\frac{1}{s}\right)^{2(d+1)}\right\}^2=: M(d,s).
		\$
		Since $\ell_\tau\circ\cV$ is a subclass of $\ell_\tau^+\circ\cV+\ell_\tau^-\circ\cV$, the covering number of $\ell_\tau\circ\cV$ is smaller than that of $\ell_\tau^+\circ\cV+\ell_\tau^-\circ\cV$. 
		If $d\geq 1$, there exists some universal constant $C_1$ such that $1+\log M(d,s)\leq C_1 d \log(e/s)$, for any $0<s<1$. Applying Lemma \ref{lemma:3} with $F=\tau^2/2$ acquires%Plugging the above bound with $Q=\PP_n$ into the upper bound for Rademarcher complexity acquires 
		\$
		\EE\bigl\{ \sqrt{n}\|\PP_n-\PP\|_\cF\bigr\}
		%&\leq  C\int_0^{\tau^2/2} \sqrt{1+\log N (s, \ell_\tau\circ\cV, L_2(\PP_n))}\d s\\
		&\leq C\tau^2\int_0^{1} \sqrt{1+\log N (s\tau^2, \ell_\tau\circ\cV, L_2(\PP_n))}\d s\\
		&\leq C \tau^2\sqrt{C_1d}\int_0^{1} \sqrt{\log(e/s)} \, \d s\leq C_2 \tau^2\sqrt{d},
		\$
		where $C_2$ is a universal constant. 
		%Since the right hand side does not depend on $Z_n^1$, this implies
		%\$
		%\cR_n(\ell_\tau\circ\cV)\leq C_2n^{-1/2}\tau^2 \sqrt{d}.
		%\$ 
		Together with \eqref{eq:unif}, we obtain that, for any $n\geq 1$ and $0\leq \delta\leq 1$, 
		\$
		\left\|\PP_n -\PP\right\|_\cF\leq 2C_2\tau^2\sqrt{\frac{d}{n}}+ \frac{\tau^2}{2 } \sqrt{ \frac{2\log(1/\delta )}{n}},
		\$
		with probability at least $1-\delta$. This completes the proof. 
	\end{proof}

    Now we are ready to prove Theorem~\ref{thm:consistency}.

		\begin{proof}[Proof of Theorem \ref{thm:consistency}]
		%The proof is a direct application of Theorem 5.7 of \cite{van2000asymptotic} by combining Theorem \ref{thm:unif} and Assumption \ref{ass:sep}. 
     To begin with, we have the following basic inequality 
		\$
		0\leq \cL(\widehat \beta_\tau) -\cL(\beta^*_\tau)& = \cL(\widehat\beta_\tau) - \cL_n(\widehat\beta_\tau) +\cL_n(\widehat\beta_\tau) - \cL_n(\beta^*_\tau) +\cL_n(\beta^*_\tau)-\cL(\beta^*_\tau)\\
		&\leq \cL(\widehat\beta_\tau) - \cL_n(\widehat\beta_\tau) +\cL_n(\beta^*_\tau)-\cL(\beta^*_\tau),
		\$
		where the last inequality is due to the fact that $\cL_n(\widehat\beta_\tau) \leq \cL_n(\beta^*_\tau).$  Rearranging the terms gives 
		\$
		0  & \leq -\PP\bigl\{ \ell_\tau(y-x^\T\beta^*_\tau)-\ell_\tau(y-x^\T\widehat\beta_\tau) \bigr\}  \\
		&\leq -(\PP_n-\PP)\bigl\{ \ell_\tau(y-x^\T\widehat\beta_\tau)-\ell_\tau(y-x^\T\beta^*_\tau) \bigr\}   \\
		&\leq 2\sup_{\beta\in \RR^d}(\PP_n-\PP)\ell_\tau(y-x^\T\beta)\\
		&\lesssim \tau^2  \sqrt{\frac{d + \log(1/\delta )}{n}}   ~\text{with probability} \geq 1-\delta\\
		&\to 0  ~\mbox{ as }~ n \to \infty,
		\$
		provided that $\tau^4 d \rightarrow o(n)$, or equivalently $\tau= \ro((n/d)^{1/4})$. 
		%We refer to the first row of the above inequalities as the basic inequality. 
		By Assumption \ref{ass:sep}, for any fixed $\varepsilon>0$, there exists some $\eta>0$ such that 
		\$
		\inf_{\beta: \|\beta- \beta_\tau^*\|_2\geq \varepsilon} \cL(\beta)>\cL(\beta_\tau^*) -\eta,
		\$
		which further implies
		\$
		\PP\bigl(  \big\|\widehat\beta_\tau-\beta^*_\tau\big\|_2\geq \varepsilon \bigr) \leq \PP\bigl\{ \cL(\widehat\beta_\tau)>\cL(\beta^*_\tau)-\eta \bigr\} \rightarrow 0. 
		\$
		This proves $ \big\|\widehat\beta_\tau-\beta^*_\tau\big\|_2   \overset{\mathbb{P}}{\to}  0$. Taking into account the fact that   $\beta_\tau^*\rightarrow \beta^*$ as $\tau\rightarrow\infty$, we obtain
		$
		\widehat\beta_\tau\overset{\mathbb{P}}{\to}\beta^* 
		$, as claimed.
	\end{proof}

\section{Proof for Asymptotic Normality}	

This section presents the proof of Theorem \ref{thm:asyn}. We first need two results on the bias and convergence rate of the ACLS estimator. 

\begin{lemma}[Bias]\label{thm:bias}
  Suppose that Assumptions \ref{ass:anti} -- \ref{ass:design} hold and $m_{4+\eta}<\infty$.  Then 
\$
\big\|\beta^*_\tau - \beta^*  \big\|_2\lesssim \kappa_\ell^{-1} d^{1/2} \tau^{-(3+\eta)} .
\$
Provided that  $\tau\gtrsim n^{1/(6+2\eta)}$, we have 
$
\|\beta^*_\tau - \beta^*  \|_2\lesssim \kappa_\ell^{-1}\sqrt{{d}/{n}}.
$
\end{lemma}

\begin{theorem}[Convergence Rate]\label{thm:cr}
  Suppose that Assumptions \ref{ass:anti} -- \ref{ass:design} hold with $m_{4+\eta}$ exists. Moreover, assume that  $n\gtrsim d^3$ and $\tau\geq n^{\frac{1}{2(1+\eta)}}$.  If $\widehat\beta_\tau \overset{\PP  }{\to}\beta_\tau^*$, then
\$
\|\widehat\beta_\tau-\beta^*_\tau\|_2=\rO_{\PP}\bigl( \kappa_\ell^{-1} \sqrt{d/n} \bigr).
\$
Consequently, $\|\widehat\beta_\tau-\beta^* \|_2=\rO_{\PP} ( \kappa_\ell^{-1} \sqrt{d/n}   )$.
\end{theorem}

    \begin{proof}[Proof of Theorem \ref{thm:asyn}]
By Theorem \ref{thm:cr}, we have $\|\widehat\beta_\tau-\beta^*\|_2=\rO_{\PP} ( (d/n)^{1/2} )$. 
For every $\beta \in \RR^d$, write $t=\beta-\beta^*$, and define 
\$
	 \Sigma_\tau(t)= \frac{1}{n}\sum_{i=1}^n\EE \{x_ix_i^\T 1(|\epsilon_i-x_i^\T t |\leq \tau)  \} , \\
  D_n(t)= \frac{1}{n}\sum_{i=1}^n \{ \ell_\tau(\epsilon_i - x_i^\T t ) - \ell_\tau(\epsilon_i)  \}  .
\$ 
Similarly to the bias analysis, we can show that, for any $t$ that satisfies $\| t \|_2 \leq r \lesssim (d/n)^{1/2}$,  a neighborhood of $0$,
	\$
	\EE   D_n(t)
	=\frac{1}{2}t^\T\Sigma_\tau(t) t+\rO\biggl(\frac{\|t\|_2\PP_n\|x\|_2}{\tau^{3+\eta}} \biggr)
	= \frac{1}{2}t^\T\Sigma_\tau(t) t+\ro\bigl(\|t\|_2^2\bigr) 
	\$
as long as $\tau^{-3-\eta}\PP_n \|x\|_2=\ro (\|t\|_2)$, which holds if $n=\ro(\tau^{6+2\eta})$ and  is implied by $\tau\propto n^{1/(2+\eta)}$.  
For $t$ in a shrinking neighborhood of $0$, let  $\tau>0$ satisfy $\max_{1\leq i \leq n} | x_i^\T t  | \leq \tau/2$. Then, we have
\$
&\|\Sigma_\tau(t)-\Sigma_\tau\|_2\\
&=\Biggl\| {\frac{1}{n}\sum_{i=1}^n \EE \bigl\{x_ix_i^\T 1(|\epsilon_i-x_i^\T t |\leq \tau)\bigr\}-  \frac{1}{n} \sum_{i=1}^n\EE \bigl\{x_ix_i^\T 1(|\epsilon_i |\leq \tau)\bigr\}}\Biggr\|_2\\
&{\leq}\Biggl\|  \frac{1}{n} \sum_{i=1}^n  x_ix_i^\T \, \EE \bigl\{ {1(\tau-|x_i^\T t|\leq |\epsilon_i |\leq \tau +|x_i^\T t|)}\bigr\}\Biggr\|_2\\
&\leq  \left\|  \frac{1}{n} \sum_{i=1}^n x_ix_i^\T \, \PP (\tau/2\leq |\epsilon_i |\leq 3\tau/2) \right\|_2\leq \frac{2m_1}{\tau } \|\Sigma\|_2\rightarrow 0, 
\$
%\rcomment{I'm not sure if the equality holds if $\epsilon_i=\tau$ and $x_i^{\T}t=-\tau/2$. Do we have extra condition on $\epsilon_i$? \scolor{S: fixed by changing to inequality.}}
provided that $\tau\rightarrow \infty$. It follows that
\$
\EE   D_n(t)
	= \frac{1}{2}t^\T\Sigma_\tau t+\ro\bigl( \|t\|_2^2 \bigr).
\$

Next, following a similar  proof of Theorem \ref{thm:cr}, we obtain the following maximal inequality
\$
\EE \sup_{\|t\|_2 \leq r} |D_n(t)- \EE   D_n(t) | \lesssim  r\sqrt{\frac{d}{n}},
\$
which further implies
\$
| D_n(t)-\EE D_n(t)| =\rO_{\PP} \bigl(  r (d/n)^{1/2} \bigr)
\$
uniformly over $t$ satisfying $\| t \|_2 \leq r$. Direct calculation yields 
	\$
	& \EE\{ \ell_\tau(\epsilon_i-x_i^\T t) - \ell_\tau(\epsilon_i) \} \\
	&=\EE\left\{\frac{1}{2}\epsilon_i^21(|\epsilon_i-x_i^\T t|\leq \tau)-\frac{1}{2}\epsilon_i^21(|\epsilon_i|\leq \tau)\right\} -\EE\bigg\{\epsilon_ix_i^\T t{1(|\epsilon_i-x_i^\T t|\leq \tau)}\bigg\}\\
	&\quad +\EE \bigg\{(x_i^\T t)^21(|\epsilon_i-x_i^\T t|\leq \tau)\bigg\}+\EE\left\{\frac{1}{2}\tau^21(|\epsilon_i-x_i^\T t|>\tau)-\frac{1}{2}\tau^21(|\epsilon_i|>\tau)\right\}.  
	\$
	Define the remainder
	\$
	R_i(t)  =\ell_\tau(\epsilon_i-x_i^\T t) -\ell_\tau(\epsilon_i)+{\psi_\tau(\epsilon_i)}x_i^\T t. 
	\$
By the definition of $\ell_\tau$, we have 
	\$
	R_i(t)&= \frac{1}{2}\epsilon_i^2\left\{1(|\epsilon_i-x_i^\T t|\leq \tau)-1(|\epsilon_i|\leq \tau)\right\} 
	+\epsilon_ix_i^\T t\left({1(|\epsilon_i|\leq  \tau)- 1(|\epsilon_i-x_i^\T t|\leq  \tau)   }\right)\\
	&\quad +(x_i^\T t)^21(|\epsilon_i-x_i^\T t|\leq \tau)+\frac{1}{2}\tau^2\left\{1(|\epsilon_i-x_i^\T t|>\tau)-1(|\epsilon_i|>\tau)\right\},
	\$
	which has an envelope $F: \RR^{d+1} \to \RR$ given by 
	\$
	F(z_i)&=\frac{1}{2}\epsilon_i^2 1(\tau-|x_i^\T t|\leq |\epsilon_i|\leq \tau+|x_i^\T t|)
	+{|\epsilon_ix_i^\T t| 1(\tau-|x_i^\T t|\leq \epsilon_i\leq \tau+|x_i^\T t|)}\\
	&\quad +(x_i^\T t)^21(|\epsilon_i-x_i^\T t|\leq \tau)+\frac{1}{2}\tau^21(\tau-|x_i^\T t|\leq |\epsilon_i|\leq \tau+|x_i^\T t|)
	\$
for $z_i = (x_i, \epsilon_i)$.		 For $\|t\|_2\leq r\leq \tau/(2\|x_i\|_2)\wedge 1$ and $\tau\geq 1$, 
	\$
	\EE F^2&\lesssim \EE\left\{ \epsilon_i^4 1(\tau-|x_i^\T t|\leq |\epsilon_i|\leq \tau+|x_i^\T t|)
	+\epsilon_i^2 (x_i^\T t)^21(\tau-|x_i^\T t|\leq \epsilon_i\leq \tau+|x_i^\T t|)\right\} \\
	&\quad +\EE\left\{(x_i^\T t)^41(|\epsilon_i-x_i^\T t|\leq \tau)+\tau^41(\tau-|x_i^\T t|\leq |\epsilon_i|\leq \tau+|x_i^\T t|)\right\}\\
	&\lesssim \frac{L_0|x_i^\T t|}{\tau^{1+\eta}}+ \frac{L_0|x_i^\T t|^2}{\tau^{3+\eta}} + (x_i^\T t)^4 \lesssim 
	 \|t\|_2^2, 
	\$
provided that  $\tau^{1+\eta}\gtrsim  \|x_i\|_2/\|t\|_2$, $\tau^{2+\eta}\gtrsim \|x_i\|_2^2  $ and $\|x_i\|_2^4\|t\|_2^2\lesssim1$, all of which hold if $n\gtrsim d^3$ and $\tau\gtrsim n^{1/(2+2\eta)}$.

Now we consider to bound  the local fluctuation uniformly 
	\$
	\sup_{t: \|t\|_2\leq r}\Biggl|\frac{1}{\sqrt{n}}\sum_{i=1}^n\left\{R_i(t)-\EE R_i(t)\right\}\Biggr|. 
	\$
	Similar to the proof of Theorem \ref{thm:unif}, we   write
	\$
	R_i(t)=\ell_\tau^+(\epsilon_i-x_i^\T t)+\ell_\tau^-(\epsilon_i-x_i^\T t)+\psi_\tau(\epsilon_i)x_i^\T t -\ell_\tau(\epsilon_i).
	\$
	This implies that the covering number for the class of functions $R_i(t)$, denoted by $\cG_r$, is upper bounded by
	\$
	N(s \|F\|_{Q,2}, \cG_r, \|\cdot\|_{Q,2})
	&\leq \left\{K (d+2) (16e)^{d+2}\left(\frac{1}{s}\right)^{2(d+1)}\right\}^3.
	\$
	Recall that $\PP f=\EE \PP_n f$. Applying the maximum inequality in Lemma \ref{lemma:5} with $\psi(x)=x^{1+\eta/4}$ and $\lambda=1$ acquires 	
	\$
	B:=\sqrt{\EE \max_{1\leq i \leq n} F^2(z_i)}\lesssim  ( \PP F^{2+\eta/2})^{\frac{1}{2+\eta/2}}= :n^{\frac{1}{2+\eta/2}}\|F\|_{\PP, 2+\eta/2}. 
	\$
	Applying Lemma \ref{lemma:4} with $\delta=\sqrt{\PP F^2}\lesssim  \|t\|_2$ yields 
	\$
	&\EE\sup_{ \|t\|_2\leq r}  \sqrt{n} |\PP_n R(t)-\PP R(t) | \\
	&\lesssim  J(\delta, F,\cF)\|F\|_{\PP,2}+\frac{BJ^2(\delta,  F,\cF)}{\delta^2\sqrt{n}}\\
	&\lesssim  \|F\|_{\PP,2} \left\{ J(\delta, F,\cF)+\frac{BJ^2(\delta,  F,\cF)}{\delta^2\sqrt{n}}\right\}\\
	&\lesssim  \|F\|_{\PP,2} \left\{ J(\delta, F,\cF)+n^{-\frac{\eta}{8+2\eta}}\frac{\|F\|_{\PP,2+\eta/2}J^2(\delta,  F,\cF)}{\|F\|_{\PP, 2}\delta^2}\right\}\\
	&\lesssim \|F\|_{\PP,2}J(\delta, F,\cF). 
	\$
	Similar to the proof of Theorem \ref{thm:unif}, we shall have  $J(\delta, F,\cF)\lesssim \|F\|_{\PP,2}\sqrt{d}$. This implies that uniformly in a shrinking neighborhood of  $0$,
	\$
	\sqrt{n} |\PP_n R(t)-\PP R(t)|=\rO_{\PP} \bigl(  \|t\|_2^2\sqrt{d}\bigr)=\ro_{\PP} \bigl(  \|t\|_2\bigr),
	\$
    if $n\gg d^2$ which is implied by $n\gtrsim d^3$. Hence
	\$
	D_n(t)= \EE D_n(t) + t^\T\PP_n\psi_\tau(\epsilon)x+\ro_{\PP}\bigl( \|t\|_2/\sqrt{n}\bigr) .
	\$
	Let $t_n= \widehat\beta_\tau -\beta^*$. Thus 
	\$
D_n(t_n)&= \EE D_n(t_n)  + t_n^\T\PP_n\psi_\tau(\epsilon)x+\ro_{\PP} \bigl(\|t_n\|_2/\sqrt{n}\bigr)\\
&= \frac{ 1+\ro_{\PP} (1)  }{2}\left\| \Sigma_\tau^{1/2} t_n- {\Sigma_\tau^{-1/2}\PP_n\psi(\epsilon)x}\right\|_2^2 - \frac{\|\Sigma_\tau^{-1/2}\PP_n\psi(\epsilon)x\|_2^2}{2}+ \ro_{\PP} (n^{-1}).
	\$
By the optimality of $\hat \beta_\tau$,  $D_n( t_n) \leq D_n (\Sigma_\tau^{-1/2}\PP_n\psi_\tau (\epsilon)x )$. Putting together the pieces, we conclude that %\scomment{The $\ro_{\PP}(\|t_n\|_2/\sqrt{n})$ is absored in the cross term.}
	\$
	&\frac{1 +\ro_{\PP}(1)}{2}\left\| \Sigma_\tau^{1/2} t_n-{\Sigma_\tau^{-1/2}\PP_n\psi(\epsilon)x}\right\|_2^2 - \frac{\|\Sigma_\tau^{-1/2}\PP_n\psi(\epsilon)x\|_2^2}{2}+ \ro_{\PP} (n^{-1}) \\
	&\leq -\frac{\|\Sigma_\tau^{-1/2}\PP_n \psi_\tau(\epsilon)  x \|_2^2}{2}+ \ro_{\PP} (n^{-1}) ,
	\$
	or equivalently
	\$
	\sqrt{n}(\widehat\beta_\tau-\beta^*)= \frac{1}{\sqrt{n}}\sum_{i=1}^n \Sigma_\tau^{-1}\psi_\tau(\epsilon_i)x_i+\Delta_n,
	\$
	where $\|\Delta_n\|_2=\ro_\PP(1).$
	
\end{proof}

\subsection{Proof of Technical Results}	
This section collects the proofs for Lemma \ref{thm:bias} and Theorem \ref{thm:cr}. 	
	\begin{proof}[Proof of Lemma \ref{thm:bias}]
 Because $\beta^*_\tau\rightarrow \beta^*$ as $\tau\rightarrow \infty$, for any $r>0$, we have $\|\beta^*-\beta^*_\tau\|_2\leq r$ for sufficiently large $\tau$.   For simplicity, for a fixed predictor $x$ with $\|x\|_2\leq C_x d^{1/2}$, write $\mu= x^\T\beta$, $\mu^*=x^\T\beta^*$ and $\mu^*_\tau=x^\T\beta^*_\tau$.  Recall that $\epsilon$ denotes the model error  satisfying  $\EE\epsilon=0$. By direction calculations, we have
	\$
&\EE\ell_\tau\left( \epsilon\right) -\EE\ell_\tau\left(\epsilon+\mu^*-\mu^*_\tau\right)\\
	&=\EE\left\{\frac{1}{2}\epsilon^21(|\epsilon|\leq \tau)\right\} +\EE \left\{\frac{1}{2}\tau^21(|\epsilon|> \tau)\right\} \\
	&\quad - \left[\EE\left\{\frac{1}{2}(\epsilon+\mu^*-\mu^*_\tau)^21(|\epsilon+\mu^*-\mu^*_\tau|\leq \tau)\right\} +\EE\left\{\frac{1}{2}\tau^21(|\epsilon+\mu^*-\mu^*_\tau|>\tau)\right\}\right]\\
	&={\EE\left\{\frac{1}{2}\epsilon^2 \big(1(|\epsilon|\leq \tau)-1(|\epsilon+\mu^*-\mu^*_\tau|\leq \tau)\big)\right\}} {-{\EE\bigg\{\epsilon(\mu^*-\mu^*_\tau)1(|\epsilon+\mu^*-\mu^*_\tau|\leq \tau)\bigg\}}}\\
	&\quad -{\EE \bigg\{\frac{1}{2}(\mu^*-\mu^*_\tau)^21(|\epsilon+\mu^*-\mu^*_\tau|\leq \tau)\bigg\}}+{\EE\left\{\frac{1}{2}\tau^2\big(1(|\epsilon|>\tau)-1(|\epsilon+\mu^*-\mu^*_\tau|>\tau)\big)\right\}}\\
&=\Rom{1}{+\Rom{2}}+\Rom{3}+\Rom{4}. 
	\$
In what follows, we bound terms $\Rom{1}$ - $\Rom{4}$ respectively. 

Note that $|\mu^* - \mu_\tau^*| \leq C_x d^{1/2} r$.
Take $r>0$ sufficiently small such that $\Delta_0=\mu^*-\mu^*_\tau\in \big\{u\in \RR: |u|\leq \delta_0 \leq (m_{4+\eta}/m_1)^{1/(3+\eta)
}\big\}$.   We start with the first term. Assumption \ref{ass:anti} implies 
	\$
	\Rom{1}&\leq \EE\left\{\frac{1}{2}\epsilon^21(|\epsilon|\leq \tau)-\frac{1}{2}\epsilon^21(|\epsilon|\leq \tau-|\Delta_0|)\right\}\leq   \frac{L_0 |\Delta_0|}{2\tau^{3+\eta}}. 
	\$
	For term \Rom{2}, using the fact that $\EE \epsilon=0$,   
%    \$
%    \{|\epsilon+\mu^*-\mu_\tau^*|>\tau\}\subseteq \{\tau-|\Delta_0|\leq \epsilon\leq \tau+|\Delta_0|\},
%    \$
    we obtain
	\$
	\Rom{2}
	&=\EE\bigg\{\epsilon(\mu^*-\mu^*_\tau)1(|\epsilon+\mu^*-\mu^*_\tau|\leq \tau)\bigg\}=-\EE\bigg\{\epsilon(\mu^*-\mu^*_\tau)1(|\epsilon+\mu^*-\mu^*_\tau|> \tau)\bigg\}\\
	&\leq |\Delta_0|\frac{\EE\big\{|\epsilon||\epsilon+\Delta_0|^{3+\eta}1(|\epsilon+\Delta_0|>\tau)\big\}}{\tau^{3+\eta}}  \\
	&\leq |\Delta_0|\frac{ 2^{3+\eta}\big\{\EE(|\epsilon|^{4+\eta}1(|\epsilon+\Delta_0|>\tau))+|\Delta_0|^{3+\eta}\EE |\epsilon|\}}{\tau^{3+\eta}}\\
	&\leq \frac{2^{3+\eta} |\Delta_0|(m_{4+\eta}+ \delta_0^{3+\eta}m_1 )}{\tau^{3+\eta}}\leq \frac{2^{4+\eta} |\Delta_0|m_{4+\eta}}{\tau^{3+\eta}},
	\$	
where the last inequality uses the fact that $\delta_0^{3+\eta}m_1\leq m_{4+\eta}$, and the second inequality follows from the inequality   $(x+y)^a\leq 2^ax^a +2^ay^a$ for any $x,y,a\geq 0$ with the convention $0^0=1$. 
For term \Rom{4}, we have %we have 
	\$
	\Rom{4}&\leq \EE\left[ \frac{1}{2}\tau^2 \left\{1( |\epsilon|> \tau ) - 1(|\epsilon|> \tau+|\Delta_0|)\right\} \right]\leq \frac{L_0 |\Delta_0|}{2\tau^{3+\eta}}.
	%\Rom{4}=\EE\left\{\frac{1}{2}\tau^21(|\epsilon+\mu^*-\mu|>\tau)-\frac{1}{2}\tau^21(|\epsilon|>\tau)\right\}
	\$
Combining the bounds for terms \Rom{1}, \Rom{2}, \Rom{4} and moving \Rom{3} to the left-hand side,  we obtain
\$
&\EE\ell_\tau\left( \epsilon\right) -\EE\ell_\tau\left(\epsilon+\mu^*-\mu^*_\tau\right) +\frac{1}{2}\EE\left\{(\mu^*-\mu^*_\tau)^21(|\epsilon+\mu^*-\mu^*_\tau|\leq \tau)\right\}\\
&\leq \frac{L_0 |\Delta_0|}{\tau^{3+\eta}}+\frac{2^{4+\eta} |\Delta_0|m_{4+\eta}}{\tau^{3+\eta}}\leq  C \frac{|\Delta_0|}{\tau^{3+\eta}},
\$
where $C>0$ is a constant depending only on $L_0$ and $\eta$. 

Summing up the above inequalities over $\epsilon_1,\ldots, \epsilon_n$, and by Assumption \ref{ass:inv}, we conclude that
\$
&\kappa_\ell \left\|\beta^*-\beta^*_\tau\right\|_2^2+\frac{1}{2n}\sum_{i=1}^n\EE\bigl\{(x_i^\T\beta^*-x_i^\T\beta^*_\tau)^2 1(|\epsilon_i+ x_i^\T(\beta^*-\beta^*_\tau)|\leq \tau)\bigr\}\\
&\leq C \frac{\|\beta^*_\tau-\beta^*\|_2\PP_n\|x\|_2}{\tau^{3+\eta}}\leq C'\frac{ d^{1/2} \|\beta^*_\tau-\beta^*\|_2}{\tau^{3+\eta}}.
\$
This leads to the claimed bound immediately. 
\end{proof}

\begin{proof}[Proof of Theorem \ref{thm:cr}]
%In the following proof we assume $m_{4+\eta}$ exists.
	Since $\widehat\beta_\tau \overset{\mathbb{P} }{\to}\beta_\tau^*$, we have $\|\widehat\beta_\tau-\beta_\tau^*\|_2\leq r$ for any $r>0$ with probability approaching one.  Therefore, using local strong convexity and the basic inequality, we have
	\$
	\kappa_\ell \|\widehat\beta_\tau-\beta^*_\tau\|_2^2  
	&\leq \EE \bigl\{ \ell_\tau(y-x^\T\widehat\beta_\tau)-\ell_\tau(y-x^\T\beta^*_\tau)\bigr\} \\
	&\leq (\PP_n-\PP)\bigl\{  \ell_\tau(y-x^\T\widehat\beta_\tau)-\ell_\tau(y-x^\T\beta^*_\tau)\bigr\}\\
	&\leq \sup_{\beta}(\PP_n-\PP)\big\{\ell_\tau(y-x^\T\beta)-\ell_\tau(y-x^\T\beta^*_\tau)\big\} 
	\$
	with probability approaching one.
	
In what follows, we tighten the upper bound in the proof of Theorem \ref{thm:unif} by considering a localized function class with $\beta$ falling in a local neighborhood of $\beta^*_\tau.$  For $(x,\epsilon)\in \RR^d \times \RR$ and $y=x^\T \beta^* + \epsilon$, define 
	\$
	f_\beta(x,\epsilon)=\ell_\tau(y- x^\T\beta)-\ell_\tau(y- x^\T\beta^*_\tau). 
	\$
	For sufficiently small $\delta>0$,  we consider the following localized function class
	\$
	\cF_\delta= \{f_\beta(x,\epsilon)=\ell_\tau(y- x^\T\beta)-\ell_\tau(y- x^\T\beta^*_\tau): \|\beta-\beta^*_\tau\|_2\leq \delta\}. 
	\$
	We first identify an envelope function $F$ for the function class $\cF_\delta$. Note that, for any $\beta$ satisfying $\|\beta-\beta^*_\tau\|_2 \leq \delta$, $|x^\T(\beta-\beta^*_\tau)|\leq \|x\|_2\delta$ and $|x^\T(\beta^*-\beta^*_\tau)|\leq \|x\|_2 b$, where $b : =\|\beta^*-\beta^*_\tau\|_2$ denotes the bias. Once again, we write $\mu= x^\T\beta$, $\mu^*=x^\T\beta^*$, $\mu^*_\tau=x^\T\beta^*_\tau$, and moreover, $\Delta_1=\mu^*-\mu$ and $\Delta_2=\mu^*-\mu^*_\tau$ such that $\Delta_1,\Delta_2 \in \big\{x: |x|\leq \delta_0 \leq \tau/2\big\}$. 
	Following a similar argument as in the proof of Theorem \ref{thm:bias}, we obtain 
	\$
	& \ell_\tau \left(y-\mu\right)-\ell_\tau\left(y-\mu^*_\tau\right)= \ell_\tau\left( \epsilon+\mu^*-\mu\right) -\ell_\tau\left(\epsilon+\mu^*-\mu^*_\tau\right)\\
	&=\frac{1}{2}\epsilon^21(|\epsilon+\mu^*-\mu|\leq \tau)-\frac{1}{2}\epsilon^21(|\epsilon+\mu^*-\mu^*_\tau|\leq \tau)\\
	&\quad +{\epsilon(\mu^*-\mu)1(|\epsilon+\mu^*-\mu|\leq \tau)-\epsilon(\mu^*-\mu^*_\tau)1(|\epsilon+\mu^*-\mu^*_\tau|\leq \tau)}\\
	&\quad +{\frac{1}{2}(\mu^*-\mu)^21(|\epsilon+\mu^*-\mu|\leq \tau)-\frac{1}{2}(\mu^*-\mu^*_\tau)^21(|\epsilon+\mu^*-\mu^*_\tau|\leq \tau)}\\
	&\quad+{\frac{1}{2}\tau^21(|\epsilon+\mu^*-\mu|>\tau)-\frac{1}{2}\tau^21(|\epsilon+\mu^*-\mu^*_\tau|>\tau)}\\
	&\leq {\frac{1}{2}\epsilon^21\left(\tau-|\Delta_1|\vee|\Delta_2|\leq |\epsilon|\leq \tau+|\Delta_1|\vee|\Delta_2|\right)}\\
        &\quad + \max\left\{ |\epsilon\Delta_1|1(|\epsilon+\Delta_1|> \tau), |\epsilon\Delta_2|1(|\epsilon+\Delta_2|> \tau)\right\}\\
        &\quad +\frac{1}{2} \max\left\{{\Delta_1^21(|\epsilon+\Delta_1|\leq \tau), \Delta_2^21(|\epsilon+\Delta_2|\leq \tau)}\right\}\\
        &\quad + \frac{1}{2}{\tau^21(\tau- |\Delta_1|\vee|\Delta_2|\leq |\epsilon|\leq \tau + |\Delta_1|\vee|\Delta_2|)}:= F(z). 
	\$
	where $z=(x ,\epsilon)$. Set $\delta_1=|\Delta_1|$, $\delta_2= |\Delta_2|$, and recall that $\delta_1\vee \delta_2\leq \delta_0\leq \tau/2$. We use the standard empirical process notation that $\PP f= \EE \PP_n f$ for any measurable function $f$. Then, the envelope function $F$ satisfies
	\$
	\PP F^2&\leq \PP\big\{\epsilon^4 1(\tau -\delta_0 \leq |\epsilon|\leq \tau+\delta_0 )+ 4\epsilon^2\delta_0 ^2 1(|\epsilon|>\tau-\delta_0 )\\
	&\qquad\quad +\delta_0 ^4+{\tau^4} 1(\tau-\delta_0 \leq |\epsilon|\leq \tau+\delta_0 ) \big\},
	\$
	where we use the inequality that $(\sum_{j=1} ^{4}a_j)^2\leq 4\sum_{j=1}^k a_j^2$. Recall that $m_k=\EE |\epsilon|^k$.   Applying Markov inequality acquires
	\begin{gather*}
	\PP\big\{\epsilon^4 1(\tau-\delta_0\leq |\epsilon|\leq \tau+\delta_0)\big\}\leq \frac{\PP\{ |\epsilon|^{4+\eta}1(\tau-\delta_0\leq |\epsilon|\leq \tau+\delta_0)\}}{\tau^\eta}\leq \frac{L_0\PP_n\|x\|_2(\delta + b)}{\tau^{1+\eta}}, \\
4\PP \{\epsilon^2\delta_0 ^2 1(|\epsilon|>\tau-\delta_0 )\}\leq \frac{2^{4+\eta}m_{4+\eta}\PP_n\|x\|_2^2 (\delta+ b)^2}{\tau^{2+\eta}}, \\
 \PP\delta_0^4\leq  (\delta+ b)^4 \PP_n\|x\|_2^4, \ 
\PP\big\{\tau^4 1(\tau-\delta_0\leq |\epsilon|\leq \tau+\delta_0)\big\}\leq \frac{2^{4+\eta}L_0 \PP_n\|x\|_2(\delta + b)}{\tau^{1+\eta}}.
	\end{gather*}
	Putting together the pieces, we obtain
	\$
	\PP F^2
	&\leq  4\Big\{  \frac{L_0 \PP_n\|x\|_2(\delta+ b)}{\tau^{1+\eta}}+   \frac{2^{4+\eta}m_{4+\eta}e_2^2 (\delta+b)^2}{\tau^{2+\eta}}\\
    &\qquad + (\delta+b)^4\PP_n\|x\|_2^4+  \frac{2^{4+\eta}L_0 (\delta +b) \PP_n\|x\|_2}{\tau^{1+\eta}} \Big\}\leq C' \delta^2 , 
	\$
	provided that 
    \#\label{eq:cond.thm.cr}
    b\lesssim\delta,~ \tau^{-(1+\eta)}\PP_n\|x\|_2\lesssim \delta, ~\tau^{-2-\eta} \PP_n\|x\|_2^2 \lesssim 1,~\delta^2\PP_n\|x\|_2^4\lesssim 1,
    \#
    where $C'$ only depends on $\eta, L_0, C_x$ and $m_{4+\eta}$. 
	
	Using a similar argument as in the proof of Theorem \ref{thm:unif} gives 
	\$
	&\EE \left\{\sup_{\|\beta-\beta_\tau^*\|_2\leq\delta} \left|(\PP_n-\PP)\left(\ell_\tau(y-x^\T \beta)-\ell_\tau(y-x^\T\beta_\tau^*)\right) \right|\right\}\\
	&\lesssim \frac{1}{\sqrt n}\|F\|_{L_2(\PP)}\int_0^{1} \sqrt{1+\log N(2s\|F\|_{L_2(\PP_n)}, \cF_\delta, L_2(\PP_n))}\d s. 
	\$
	The covering number of the function class $\cF_\delta$ can be bounded as
	\$
	N(2s\|F\|_{L_2(\PP_n)}, \cF_\delta, L_2(\PP_n))&\leq N(2s\|F\|_{L_2(\PP_n)}, \ell_\tau^+\circ\cV+\ell_\tau^-\circ\cV, L_2(\PP_n))\\
	&\leq  \left\{K (d+2) (16e)^{d+2}\left(\frac{1}{s}\right)^{2(d+1)}\right\}^2.
	%N(s\tau^2, \ell_\tau^+\circ\cV+\ell_\tau^-\circ\cV, L_2(Q))&\leq N(s\tau^2/2, \ell_\tau^+\circ\cV, L_2(Q))\times N(s\tau^2/2, \ell_\tau^+\circ\cV, L_2(Q))\\
	%&\leq \left\{K (d+2) (16e)^{d+2}\left(\frac{1}{s}\right)^{2(d+1)}\right\}^2=: M(d,s).
	\$
	Hence, Lemma \ref{lemma:3} yields
	\$
	&\EE \left\{\sup_{\|\beta-\beta_\tau^*\|_2\leq\delta} \left|(\PP_n-\PP)\left(\ell_\tau(y-x^\T \beta)-\ell_\tau(y-x^\T\beta_\tau^*)\right) \right|\right\}
	\leq  C\delta \sqrt{\frac{d}{n}} ,
	\$
    for some constant $C$   depending only on $\eta, L_0, C_x$ and $m_{4+\eta}$.

    \iffalse
    By Assumption \ref{ass:inv}, we have for any $0<\delta<r$
	\$
	\kappa_\ell  \|\widehat\beta_\tau-\beta^*_\tau\|_2^2
	&\leq \PP \ell_\tau(y-x^\T\widehat\beta_\tau)-\PP \ell_\tau(y-x^\T\beta^*_\tau)\\
	&\leq\PP \left\{\sup_{f\in\cF_\delta} \left|(\PP_n-\PP)f \right|\right\}\lesssim\frac{C^{1/2}\delta\sqrt{d} }{\sqrt{n}}=:\phi_n(\delta).
	\$
    \fi
	
	Now we are ready to establish the convergence rate for $\widehat\beta_\tau$. Let $\phi_n(\delta):= C\delta (d/n)^{1/2}$. To prove $\|\widehat\beta_\tau-\beta^*_\tau\|_2=\rO_{\mathbb P}(\delta_n)$ for some $\delta_n$, it suffices to show that
	\$
 \PP\left(\|\widehat\beta_\tau-\beta^*_\tau\|_2>2^M \delta_n\right)\rightarrow 0  ~\text{as}~ M\rightarrow\infty. 
	\$
	For $\phi_n(\cdot)$ defined above, we have $\phi_n(c\delta) =  c\phi_n(\delta) $ for all $c, \delta>0$. It then follows from the basic inequality that%Take $j$ large enough such that
	\$
	\PP\bigl(\|\widehat\beta_\tau-\beta^*_\tau\|_2>2^M \delta_n\bigr)
	& =\sum_{j>M} \PP\left(2^{j-1}\delta_n<\|\widehat\beta_\tau-\beta^*_\tau\|_2\leq 2^j\delta_n\right)\\
	&\leq\sum_{j>M} \PP \left(\kappa_\ell  2^{2j-2}\delta_n^2\lesssim \sup_{f\in\cF_\delta}(\PP_n-\PP)f, \|\widehat\beta_\tau-\beta^*_\tau\|_2\leq 2^j\delta_n\right) \\
	&=\sum_{j>M}\PP \left( \sup_{f\in\cF_{2^j\delta_n}}(\PP_n-\PP)f\gtrsim 2^{2j-2}\delta_n^2 \kappa_\ell \right)\\
	&\leq \sum_{j>M}\frac{1}{\kappa_\ell  2^{2j-2}\delta^2_n}\EE \left\{\sup_{f\in\cF_{2^j\delta_n}} \left|(\PP_n-\PP)f \right|\right\}\\
	&\leq 4\sum_{j>M}\frac{\phi(2^j\delta_n)}{\kappa_\ell  2^{2j}\delta^2_n}=\frac{4}{\kappa_\ell } \sum_{j>M} \frac{\phi(2^j\delta_n)}{2^{2j}\delta^2_n}\\
	&\leq \frac{4}{\kappa_\ell } \frac{\phi(\delta_n)}{\delta_n^2}\sum_{j>M} \frac{1}{2^{j}}. 
	\$
    We choose $\delta_n$ in a way that $\phi_n(\delta_n) = \kappa_\ell \delta_n^2$, implying $\delta_n = C \kappa_\ell^{-1}(d/n)^{1/2}$. Consequently,
    \$
    \PP\bigl(\|\widehat\beta_\tau-\beta^*_\tau\|_2>2^M \delta_n\bigr) 
    \leq 4\sum_{j>M} 2^{-j} \to 0 ~\mbox{ as }~ M \to \infty , 
    \$
    which in turn implies $\|\widehat\beta_\tau-\beta^*_\tau\|_2 = \rO_{\mathbb{P}} (\kappa_\ell^{-1} (d/n)^{1/2})$, as claimed.

    It remains to verify that \eqref{eq:cond.thm.cr} holds under the conditions of the theorem. By Lemma \ref{thm:bias}, $b\lesssim \delta_n$ if $\tau\gtrsim n^{1/(6+2\eta)}$. Moreover, $\tau^{-(1+\eta)}\PP_n \|x\|_2\lesssim \delta_n$ if $\tau\gtrsim n^{1/(2+2\eta)}$; $\tau^{-2-\eta}\PP_n\|x\|_2^2\lesssim 1$ if $\tau\gtrsim d^{1/(2+\eta)}$;   $\delta_n^2\PP_n\|x\|_2^4\lesssim 1$ if $n\gtrsim d^3$. Thus \eqref{eq:cond.thm.cr} is satisfied under the constraints $n\gtrsim d^3$ and $\tau\gtrsim n^{1/(2+2\eta)}$. This completes the proof. 
	\end{proof}

	\section{Inequalities for Empirical Processes}
	
	%%% Tail Probability and Rademacher complexity
	
	The first lemma, which is a direct consequence of the bounded differences inequality, provides a concentration inequality for the suprema of bounded empirical processes.
	
	\begin{lemma}\label{lemma:1}
		Let $\cF$ be  a class of measurable functions $f: \cX \rightarrow \RR $ that are uniformly $b$-bounded, that is, $\sup_{x\in \cX}|f(x) | \leq b$. 
		Then, for any $n\geq 1$ and $0\leq \delta\leq 1$, 
		\$
		\left\|\PP_n -\PP\right\|_\cF\leq \EE\|\PP_n-\PP\|_\cF+ b \sqrt{\frac{2\log(1/\delta)}{n}} 
		\$
		with probability at least $1-\delta$. 
	\end{lemma}

	\iffalse
	%%%%Talagrand
	We need the following version of Talagrand inequality which is due to Bousquet (2003). 
	\begin{lemma}\label{lemma:1.pre}
		Let $\cF$ be  a separable  class of functions $f: \cX \rightarrow \RR $ that is uniformly $1$-bounded. Assume all functions  $f$ in $\cF$ are measurable, square integrable with $\EE f(Z_i)=0$. Let $\Delta=\sup \sum_{i=1}^n f(Z_i)$. Let $\sigma^2=\sup_{f\in \cF} \EE f^2(Z_i).$  Then for $v=n\sigma^2+2\EE \Delta$ and for any $0\leq \delta\leq 1$, we have 
		\$
		\Delta\leq \EE \Delta+\sqrt{2 v\log\frac{1}{\delta}}+\frac{1}{3}{\log\frac{1}{\delta}}
		\$
		with probability at least $1-\delta$. 
	\end{lemma}

	\begin{proof}[Proof of Lemma \ref{lemma:1.pre}]
		See Theorem 7.3 in Bousquet (2003). 
	\end{proof}
\fi

	Let $Z_n^1=(Z_1,\ldots, Z_n)^\T$, where $Z_i$'s are identically and independently distributed random variables.  For a function class $\cF$, 
	define the  empirical Rademacher complexity $\cR(\cF(Z_1^n))$ and the Rademacher complexity $\cR_n(\cF)$ as
	\$
	&\cR(\cF(Z_1^n)/n)\coloneqq \EE_{\varepsilon}\left[\sup_{f\in \cF} \left|\frac{1}{n} \sum_{i=1}^n \varepsilon_i f(Z_i)\right|\right], \\
	&\cR_n(\cF)\coloneqq \EE_{Z}\left[\cR(\cF(Z_1^n)/n)\right]=\EE_{Z,\varepsilon}\left[\sup_{f\in \cF} \left|\frac{1}{n} \sum_{i=1}^n \varepsilon_i f(Z_i)\right|\right]. 
	\$
		Define  Dudley's entropy integral as
	\$
	J(\delta, F,\cF):=\int_0^\delta\sqrt{1+\log\sup_Q N( s\|F\|_{Q,2}, \cF, \|\cdot\|_{Q,2})}  \, \d s.
	\$
	
	\begin{lemma}\label{lemma:3}
		Let $F$ be an envelope for the class $\cF$ with $\PP F^2<\infty$. Then 
		\begin{gather*}
		\EE\sup_{f\in\cF} \left(\sqrt{n} |\PP_nf-\PP f |\right)\leq C \|F\|_{L^2(\PP)}J(1, F,\cF),
		\end{gather*} 
		\text{where} $C$ is a universal constant.
		%\begin{gather*}
		%J(1, F,\cF):=\int_0^1\sqrt{1+\log\sup_Q N((\epsilon\|F\|_{Q,2}, \cF, \|\cdot\|_{Q,2}))}\d s,
		%\end{gather*}
	\end{lemma}
	\begin{proof}[Proof of Lemma \ref{lemma:3}]
		We first bound the expectation of suprema  of the empirical process by the Rademacher complexity. By a symmetrization argument, we obtain 
		\$
		\EE \|\PP_n-\PP\|_\cF\leq  2\cR_n(\cF). 
		\$
		It remains to bound the right hand side.  For a general class of functions $\cF$, applying  Lemma  \ref{lemma:2} with  $T=\cF(Z_n^1)\coloneqq \big\{(f(Z_1),\ldots, f(Z_n)): f\in \cF\big\}$
		gives 
		\#\label{thm:eq:1}
		\EE \left[\sup_{f\in \cF} \left|\frac{1}{\sqrt n}\sum_{i=1}^n\varepsilon _i f(Z_i)\right|\right]\leq  C\EE \int_0^{D_\cF} \sqrt{\log N (s, \cF(Z_n^1)\cup \{0\}, d_n)}\d s,
		\#
		where $D_\cF=\sup_{f\in\cF}\sqrt{\PP_n f^2}$. 
		Let $L_2(\PP_n)$ refer to the pseudometric on $\cF$ given by
		\$
		(f,g)\mapsto  \sqrt{\frac{1}{n}\sum_{i=1}^n \big(f(X_i)-g(X_i)\big)^2}. 
		\$
		Rewrite \eqref{thm:eq:1} as
		\$
		\EE \left[\sup_{f\in \cF} \left|\frac       {1}{\sqrt n}\sum_{i=1}^n\varepsilon _i f(Z_i)\right|\right]
		&\leq  C \EE \int_0^{D_\cF} \sqrt{\log N (s, \cF\cup \{0\}, L_2(\PP_n))}\d s\\
		&\leq  C\EE\int_0^{D_\cF} \sqrt{1+\log N (s, \cF, L_2(\PP_n))}\d s.
		\$
		%Now taking $\cF=\ell_\tau\circ \cV$ which is uniformly $\tau^2/2$-bounded, we have $D_{\ell_\tau\circ\cV}\leq \tau^2/2$, and thus 
		%\$
		%\EE \left[\sup_{f\in \ell_\tau\circ\cV} \left|\frac       {1}{\sqrt n}\sum_{i=1}^n\varepsilon _i f(Z_i)\right|\right]
		%&\leq  C\EE \int_0^{\tau^2/2} \sqrt{1+\log N (s, \ell_\tau\circ\cV, L_2(\PP_n))}\d s.
		%\$
		
		%We shall first write
		%\$
		%N(s, \cF(Z_1,\ldots, Z_n)\cup \{0\}, d_n)= N(s, \cF\cup \{0\}, L^2(\PP_n)). 
		%\$
		Now because $F$ is an envelope for the class $\cF$, that is 
		\$
		\sup_{f\in \cF}|f(x)|\leq F(x),~\text{for every}~ x\in \cX,
		\$
		we have $\sup_{f\in \cF}\sqrt{\PP_n f^2}\leq \sqrt{\PP_n F^2}$. Thus 
		\$
		\EE\sup_{f\in\cF} \left(\sqrt{n} |\PP_nf-\PP f |\right)&\leq2 \EE \left[\sup_{f\in \cF} \left|\frac  {1}{\sqrt n}\sum_{i=1}^n\varepsilon _i f(Z_i)\right|\right]\tag{Symmetrization}\\
		&\leq  C\EE \int_0^{D_\cF} \sqrt{1+\log N (s, \cF, L_2(\PP_n))}\d s\\
		&\leq C\EE \int_0^{\sqrt{\PP_n F^2}} \sqrt{1+\log N (s, \cF, L_2(\PP_n))}\d s\\
		&\leq C\EE \sqrt{\PP_n F^2}\int_0^{1} \sqrt{1+\log N (s\sqrt{\PP_n F^2}, \cF, L_2(\PP_n))}\d s\\
		&\leq C\sqrt{\PP F^2}\int_0^{1} \sqrt{1+\log \sup_{\QQ} N (s\sqrt{\QQ F^2}, \cF, L_2(\QQ))}\d s. 
		\$
	\end{proof}

	The next lemma bounds the expectation of localized empirical process \citep{chernozhukov2014gaussian}, which sharpens the bound obtained by directly applying Lemma \ref{lemma:2}.

	\begin{lemma}\label{lemma:4}
		Suppose that $\|F\|_{\PP,2}<\infty.$ Let $\sigma^2>0$ be any positive constant such that $\sup_{f\in \cF} \PP f^2\leq \sigma^2\leq \|F\|_{\PP, 2}$. Let $\delta=\sigma/\|F\|_{\PP,2}$. Define $B=\sqrt{\EE\max_{1\leq i\leq n}F^2(X_i)}$. Then 
		\$
		\EE \bigl( \sqrt{n}\left\|\PP_n-\PP\right\|_\cF \bigr) \leq C\Biggl\{ J(\delta, F,\cF)\|F\|_{\PP,2}+\frac{BJ^2(\delta, F, \cF)}{\delta^2\sqrt{n}}\Biggr\},
		\$
		where $C>0$ is a universal constant.
	\end{lemma}

	%%%%Another maximum inequality
	We need another useful maximal inequality. 
	\begin{lemma}\label{lemma:5}
		Let $\psi: R\mapsto R^+$ be a convex function that is strictly increasing on $R^+$. Let $X_1,\ldots, X_n$ be $n$ random variables. Then
		\$
		\EE \max_{1\leq i\leq n} X_i\leq \inf_{\lambda>0}\frac{1}{\lambda}\psi^{-1}\left(\sum_{i=1}^n\EE \psi(\lambda X_i)\right). 
		\$
		
	\end{lemma}
	
	\subsection{Technical Lemmas}

	%%%%% A maximum inequality
	For $s, t\in \RR^n$, define the metric $d_n(s,t)$ as
	\$
	d_n(s,t)\coloneqq n^{-1/2}\|s-t\|_2=\sqrt{\frac{1}{n}\sum_{i=1}^n (s_i-t_i)^2},
	\$
	a re-scaled Euclidean metric.
	
	\begin{lemma}\label{lemma:2}
		Let $\varepsilon_1,\ldots, \varepsilon_n$ be i.i.d. Rademacher random variables. Suppose $T\subseteq \RR^n$ and consider the stochastic process $\{X_t: t\in T\}$ given by 
		$
		X_t\coloneqq n^{-1/2} \sum_{i=1}^n{\varepsilon_i t_i}. 
		$
		Then 
		\$
		\EE \left[\sup_{t\in T} \left|\frac{1}{\sqrt n}\sum_{i=1}^n\varepsilon_i t_i\right|\right]\leq C\EE \left[ \int_0^{D} \sqrt{\log N (s, T\cup \{0\}, d_n)}\d s \right],
		\$
		where $D=\max_{t\in T} \sqrt{\sum_{i=1}^n t_i^2/n}$ and $C>0$ is a universal constant. 
	\end{lemma}
	
	\begin{proof}[Proof of Lemma \ref{lemma:2}]
		
		By Hoeffding's inequality, for every $u\geq 0,$
		\$
		\PP\left(|X_t-X_s|\geq u \right)\leq 2\exp\left(\frac{-nu^2}{2\sum_{i=1}^n (s_i-t_i)^2}\right)=2\exp\left(\frac{-u^2}{2d_n^2(s,t)} \right),
		\$
		so that $\{X_t, t\in T\}$ is a sub-Gaussian process with the metric $d_n$. Since $T\subset \RR^n$ is naturally separable and  the map $t\mapsto n^{-1/2}\sum_{i=1}^n\varepsilon_i t_i$ is linear and continuous in $t$, $\{X_t, t\in T\}$ is separable. Therefore applying the Dudley's entropy integral bound to the Rademacher complexity, we get  
		\$
		\EE \left[\sup_{t\in T} \left|\frac{1}{\sqrt n}\sum_{i=1}^n\varepsilon_i t_i\right|\Bigg| Z_n^1\right]\leq  C\int_0^{D} \sqrt{\log N (s, T\cup \{0\}, d_n)}\d s,
		\$
		where 
		$
		D=\diam \left(T\cup \{0\}\right)=\max_{t\in T}\sqrt{\sum_{i=1}^nt_i^2/n}.  
		$
		Applying expectation with respect to $Z_i$'s on both sides finishes the proof. 
		%\scomment{current point. \today} 
	\end{proof}

\section{Video Surveillance}	
\subsection{An Alternating Minimization Algorithm for Problem \eqref{eq:ARR_original}}\label{section: AM}
We reformulate the optimization problem \eqref{eq:ARR_original} as
\begin{align}\label{eq:ARR_2}
\min_{m,U,\{s_i\}_{i=1}^n,\{\delta_i\}_{i=1}^n} \quad& \frac{1}{2n} \Bigg (\sum_{i=1}^{n}\|(y_i- m-Us_i)\odot [(1-\delta_i)1_p]\|_2^2+\tau^2\sum_{i=1}^n \delta_i \Bigg ),\\
\text{s.t. }\quad&U^{\T}U=I_q, \, \delta_i=1(\|y_i-m-Us_i\|_2>\tau) \text{ for } 1\leq i\leq n, \nonumber
\end{align}
where $1_p$ denotes the $p \times 1$ vector of all ones, $\odot$ denotes the Hadamard product, $\delta_i=1(\|y_i-m-Us_i\|_2>\tau)$ is a decision variable taking values 0 or 1.
To optimize \eqref{eq:ARR_2}, we develop an alternating minimization algorithm to    iteratively update the parameters  in the order of $m \rightarrow \{s_i\} _{i=1}^n \rightarrow U \rightarrow \delta_i$.   %$\delta_i $ being fixed, and  update $\delta_i$ using feasibility conditions $\delta_i=1(\|y_i-m-Us_i\|_2>\tau)$.
We first fix $\delta_i \in \{0,1\},\, 1\leq i\leq n$ and  run one round of updates of $m, U, s_i$. We then  update $\delta_i$ using feasibility conditions $\delta_i=1(\|y_i-m-Us_i\|_2>\tau)$. This finishes one round. We then run multiple steps until convergence. We derive the details in what follows. 

We derive the updating rules for $m, U,$ and $s_i$. Optimizing  \eqref{eq:ARR_2} with respect to $m, U, s_i$ reduces to
\begin{align}\label{eq:ARR_4}
\min_{m,U,\{s_i\}_{i=1}^n} \quad& \frac{1}{2n} \sum_{i=1}^{n}\|(y_i- m-Us_i)\odot [(1-\delta_i)1_p]\|_2^2,\\%+\tau^2\sum_{i=1}^n \delta_i \Bigg ),\\
\text{s.t. }\quad&U^{\T}U=I_q. \nonumber
\end{align}
%{which can then be solved by one-round alternating minimization in the order of $m \rightarrow \{s_i\} _{i=1}^n \rightarrow U$.} %to iteratively solve \eqref{eq:ARR_4} and update $\delta_i$, which is to cyclically update $m \rightarrow \{s_i\} _{i=1}^n \rightarrow U \rightarrow \delta_i$ until convergence.} 
%Thus the full algorithm runs  in the order of $m \rightarrow \{s_i\} _{i=1}^n \rightarrow U \rightarrow \delta_i$. %and can be seen as a modified version from  Algorithm 1 of \cite{mateos2012robust}.  %Starting at iteration $0$,
 At iteration $k\geq 0$, with fixed $\delta_i^k, U^k$ and $s^k_i$, we first take the partial derivative of the objective function in \eqref{eq:ARR_4} with respect to $m$, set it to zero, and obtain 
\$
  m\odot\frac{1}{n}\sum_{i=1}^n  (1-\delta_i^k)1_p=\frac{1}{n}\sum_{i=1}^n (y_i-U^{k}s_i^{k})\odot [(1-\delta_i^k)1_p].
\$
We then  update $m^{k+1}$ as the solution to  the above equations
 \$
 m^{k+1}=\sum_{i=1}^n (y_i-U^{k}s_i^{k})\odot [(1-\delta_i^k)1_p] \oslash \sum_{i=1}^n [(1-\delta_i^k)1_p],
 \$
where $\oslash$ denotes the element-wise division. 
Similarly, 
\iffalse we take the partial derivative with respect to $s_i$ and set this to zero,
\$
\frac{\partial}{\partial s_i}&\frac{1}{2n}\sum_{i=1}^n\left\{[(y_i-m)\odot(1-\delta_i^k)1_p]^{\T}[(y_i-m)\odot(1-\delta_i^k)1_p]\\
&\quad-2[(y_i-m)\odot(1-\delta_i^k)1_p]^{\T}[Us_i\odot(1-\delta_i^k)1_p]\\
&\quad+s_i^{\T}U^{\T}\odot[(1-\delta_i^k)1_p]^{\T}Us_i\odot(1-\delta_i^k)1_p\right\},\\
&=\frac{1}{n}\left\{-[U^{\T}\odot(1-\delta_i^k)1_p](y_i-m)\odot(1-\delta_i^k)1_p+[U^{\T}\odot(1-\delta_i^k)1_p][Us_i\odot(1-\delta_i^k)1_p]\right\}=0.
\$ 
\fi
we update $s_i^{k+1}$ as
\$
s_i^{k+1}=(U^k)^{\T}[(y_i-m^{k+1})\odot (1-\delta_i^k) 1_p]. 
\$
We then optimize  $U$ while fixing other variables to their up-to-date values. Write $Y=(y_1,\ldots, y_n)^\T\in \RR ^{n\times p}$. Let 
\$
\tilde{Y}_o^{k+1}=Y-1_n (m^{k+1})^{\T}-O^{k+1},
\$
 where $O^{k+1}\coloneqq(o_1^{k+1},...,o_n^{k+1})^{\T}$ such that  $o_i^{k+1}=(y_i-m^{k+1}-U^ks_i^{k+1})\odot \delta_i^k, \;i=1,...,n$. We update $U^{k+1}$ as
\#
\label{eq:ARR_svd}
U^{k+1}=\argmin_{U \in \mathbb{R}^{p \times q}} \frac{1}{2n}\|\tilde{Y}_o^{k+1}-S^{k+1}U^{\T}\|_{\text{F}}^2,\quad\textnormal{s.t. } U^{\T} U=I_q.
\#
We need the following lemma to obtain a closed form update of $U^{k+1}$. 
\begin{lemma}\label{lem:ARR_svd}
	For $Y\in \mathbb{R}^{n \times p}$ and $S \in \mathbb{R}^{n \times q}$, let  $LDR^{\T}=$svd$(Y^{\T}S)$ be the SVD decomposition of $Y^\T S$ with $L\in \mathbb{R}^{p \times q}, D\in \mathbb{R}^{q \times q}, R\in \mathbb{R}^{q \times q}$.
    Then $W=LR^{\T}$ is the solution to the following constrained least squares problem
	\#
{W}=\argmin_{U \in \mathbb{R}^{p \times q}} \frac{1}{2n}\|Y-SU^{\T}\|_{\text{F}}^2,\quad\textnormal{s.t. } U^{\T} U=I_q.
	\# \end{lemma}
Applying lemma \ref{lem:ARR_svd}, we update $U^{k+1}$ as the  multiplication of the left and right singular matrices of $(\tilde{Y}_0^{k+1})^{\T}S^{k+1}$, that is    $U^{k+1}=L^{k+1}(R^{k+1})^\T$ with $L^{k+1}D^{k+1}(R^{k+1})^{\T}=\text{svd}((\tilde{Y}_o^{k+1})^{\T}S^{k+1})$.  To finish the cycle, we last update $\delta_i^{k+1}=1(\|y_i-m^{k+1}-U^{k+1}s_i^{k+1}\|_2>\tau)$. We {initialize the full algorithm with  $\delta_i^0=0$, for $i=1,\cdots,n$ and $U^0=0_p0_q^{\T}$, where $0_p$ denotes a $p \times 1$ vector of all zeros}, and  repeat the above steps  until convergence, that is, $|\cL_{n,\tau}(y_i-m^{k+1}-U^{k+1}s_i^{k+1})-\cL_{n,\tau}(y_i-m^k-U^ks_i^k)|\leq \varepsilon_\text{opt}$ for some pre-specified optimization error $\varepsilon_\text{opt}$.  We use $\varepsilon_\text{opt}=10^{-5}$ in our experiments.  Algorithm \ref{alg:AM} summarizes the pseudo code of the full algorithm. 

%\rcolor{The current algorithm works well without randomized initialization. We have also carried out experiments  with initializing $s_i$ at 0 and randomly initializing $m$. The results are similar.} \scomment{try this.}

\begin{algorithm}[!t]
	\caption{An alternating minimization (AM) algorithm for problem \eqref{eq:ARR_original} using CLS.   }\label{alg:AM}
	\begin{algorithmic}[1]
		\STATE{\textbf{Algorithm}: $\{\hat{m},\hat{U},\{\hat{s}_i\}_{i=1}^n,\{\hat{\delta}_i\}_{i=1}^n\} \leftarrow \text{AM}\big(Y, \tau, \varepsilon_{\text{opt}} )$.}	
		\STATE{\textbf{Initialization}: $U^0=0_p0_q^{\T}$, $\delta_i=0,\, 1\leq i\leq n.$}
		\STATE{\textbf{Input}: $\tau>0$}
		\STATE{\textbf{For }$k\geq 0$ \textbf{until} $|\cL_{n,\tau}(y_i-m^{k+1}-U^{k+1}s_i^{k+1})-\cL_{n,\tau}(y_i-m^k-U^ks_i^k)|\leq\varepsilon_\text{opt}$ \textbf{do}} 
		\begin{flalign*}
		&m^{k+1}=\sum_{i=1}^n (y_i-U^{k}s_i^{k})\odot [(1-\delta_i^k)1_p] \oslash \sum_{i=1}^n [(1-\delta_i^k)1_p],\\
		&s_i^{k+1}=(U^k)^{\T}[(y_i-m^{k+1})\odot (1-\delta_i^k) 1_p],\, 1\leq i \leq n,\\
		&U^{k+1}=L^{k+1}(R^{k+1})^{\T},\\ %\text{ where } L^{k+1}D^{k+1}(R^{k+1})^{\T}=\text{svd}((\tilde{Y}_o^{k+1})^{\T}S^{k+1}), S^{k+1}=(s_1^{k+1},\cdots,s_n^{k+1})^{\T}\\
		&\delta_i^{k+1}=1(\|y_i-m^{k+1}-U^{k+1}s_i^{k+1}\|_2>\tau),\, 1\leq i\leq n.
		\end{flalign*}
		{\bf end for}
		\STATE{{\bf Output}: $\{\hat{m},\hat{U},\{\hat{s}_i\}_{i=1}^n,\{\hat{\delta}_i\}_{i=1}^n\}=\{m^{k+1},U^{k+1},\{s_i^{k+1}\}_{i=1}^n, \{\delta_i^{k+1}\}_{i=1}^n\}$.}
	\end{algorithmic}
\end{algorithm}	
\subsubsection{Proof of Lemma \ref{lem:ARR_svd}}

	\begin{proof}[Proof of Lemma \ref{lem:ARR_svd}]
	$W$ solves
	\$ 
	&\quad\argmin_{U^\T U=I_q} \langle Y-SU^{\T}, Y-SU^{\T} \rangle \\
	&=\argmin_{U^\T U=I_q} (\|Y\|_{\text{F}}^2+\|SU^{\T}\|_{\text{F}}^2-2 \langle Y, SU^{\T} \rangle)\\
	&=\argmax_{U^\T U=I_q} \langle Y, SU^{\T} \rangle      \tag{$\|SU^{\T}\|_{\text{F}}^2=\text{tr}(US^{\T}SU^{\T})=\text{tr}(S^{\T}S)$}\\
	&=\argmax_{U^\T U=I_q} \langle U, Y^{\T}S \rangle\\
	&=\argmax_{U^\T U=I_q} \text{tr}(U^{\T} Y^{\T}S)\\
	&=\argmax_{U^\T U=I_q} \text{tr}( U^{\T} LDR^{\T} )   \tag{$LDR^{\T}=$svd$(Y^{\T}S)$}\\
	&=\argmax_{U^\T U=I_q} \text{tr}(R^{\T}U^{\T}  LD).
	\$
	Let $\tilde{U}=UR$, we have
	$$
	\text{tr}(R^{\T}U^{\T}  LD)=\text{tr}(\tilde{U}^{\T}  LD).
	$$
	Since $D$ is a diagonal matrix with entries $\geq 0$,  $\text{tr}(\tilde{U}^{\T}  LD)$ is maximized when the diagonal entries of $\tilde{U}^{\T}  L$, $\tilde{u}_i^{\T}l_i, i=1,\cdots q$, are non-negative and maximized. By Cauchy-Schwartz inequality, the maximum is achieved when $\tilde{U}=L$, which can be done by setting  $U=\tilde{U}R^{\T}=LR^{\T}=W$. This completes the proof. %Therefore, $\argmax_{U \in \mathbb{R}^{p \times q}} \text{tr}(R^{\T}U^{\T}  LD)=LR^{\T}$.
\end{proof}

\subsection{The Ordinary Least Squares}
The ordinary least squares method optimizes
\#\label{eq:ARR_ols}
\argmin_{m,U,\{s_i\}_{i=1}^n} \frac{1}{n}\sum_{i=1}^n\ell (y_i-m-Us_i),\quad\textnormal{s.t. } U^{\T} U=I_q,
\#
where $\ell(x)=\|x\|_2^2/2$.

Taking  the derivative of the objective function in \eqref{eq:ARR_ols} with respect to $s_i$ we obtain %\scomment{Change this and all that follows.}
\$
&\frac{\partial}{\partial s_i}\frac{1}{2n}\sum_{i=1}^n\left\{(y_i-m)^{\T}(y_i-m)-2(y_i-m)^{\T}Us_i+s_i^{\T}U^{\T}Us_i\right\},\\
&=\frac{1}{n}\left\{-U^{\T}(y_i-m)+U^{\T}Us_i\right\}.\$ 
Setting it to $0$ and using $U^{\T}U=I_q$, we obtain $s_i=U^{\T}(y_i-m)$.
Plugging $s_i=U^{\T}(y_i-m)$ into  \eqref{eq:ARR_ols} acquires 
\#\label{eq:ols_relax}
\argmin_{U,m} \frac{1}{2n}\sum_{i=1}^n\|(y_i-m)-UU^{\T}(y_i-m)\|_2^2,\,\textnormal{s.t. } U^{\T} U=I_q.
\#
To further simplify the above optimization problem,  we take the derivative of the  objective function in the above display  with respect to $m$, set it to zero and obtain 
%$$
%\frac{\partial}{\partial m}\frac{1}{2n}\sum_{i=1}^n[y_i-m-U^{\T}(y_i-m)]^{\T}[y_i-m-U^{\T}(y_i-m)]=-\frac{1}{n}\sum_{i=1}^n[y_i-m-UU^{\T}(y_i-m)],
%$$
%Setting this to $0$, we obtain
\begin{equation}\label{eq:eq1}
(I_p-UU^{\T})(\frac{1}{n}\sum_{i=1}^n y_i-m)=0,
\end{equation}
to which  $m=\sum_{i=1}^n y_i/n=: \bar y$ is a solution.
%After we updating $m=Y^{\T}1_n/n=\bar{y}$, where $Y=(y_1,\cdots,y_n)$, the question remains to solve for $U$ by optimizing
Using this, \eqref{eq:ols_relax} further reduces to
\#
\label{eq:svd}
\argmin_{U} \frac{1}{2n}\sum_{i=1}^n\|(y_i-\bar{y})-UU^{\T}(y_i-\bar{y})\|_2^2,\quad\textnormal{s.t. } U^{\T} U=I_q.
\#
We need the following lemma to obtain a closed-form update for $U$. 
	\begin{lemma}\label{lem:svd}
	Suppose $q\leq p$, the solution to \eqref{eq:svd} is $U=R(:,1:q)$ where $ LDR^{\T}=\text{svd}(\tilde{Y})$ with  $\tilde{Y}=Y-1_n\bar{y}^{\T}$ and $R(:,1:q)$  consists of the first $q$ columns in $R$.
\end{lemma}
Applying lemma \ref{lem:svd}, we obtain $U=R(:,1:q)$ as the solution to \eqref{eq:svd}. Finally, we could use $m$ and $U$ to update $s_i=R(:,1:q)^{\T}(y_i-\bar{y})$. Therefore, a solution to  \eqref{eq:ARR_ols} is 
\$
\{\hat{m},\hat{U},\hat{s_i}\}=\left \{\sum_{i=1}^n y_i/n, R(:,1:q), R(:,1:q)^{\T}(y_i-\bar{y})\right\}.
\$
We mention that the alternating optimization algorithm only needs one sweep and thus is very fast. 

\subsubsection{Proof of Lemma \ref{lem:svd}}	\begin{proof}[Proof of Lemma \ref{lem:svd}]
	We rewrite the objective function of \eqref{eq:svd}, ignoring the factor $1/2n$, as
	\$
	&\quad \sum_{i=1}^n\|(y_i-\bar{y})-UU^{\T}(y_i-\bar{y})\|_2^2\\
	&=\|\tilde{Y}(I_p-UU^{\T})\|_{\text{F}}^2\\
	&=\text{tr}((I_p-UU^{\T})^{\T}\tilde{Y}^{\T}\tilde{Y}(I_p-UU^{\T}))\\
	&=\text{tr}((I_p-UU^{\T})^{\T}RD^{\T}L^{\T}LDR^{\T}(I_p-UU^{\T})\tag{$LDR^{\T}=\text{svd}(\tilde{Y})$}\\
	&=\text{tr}((I_p-UU^{\T})^{\T}RD^{\T}DR^{\T}(I_p-UU^{\T}))\\
	&=\text{tr}(RD^{\T}DR^{\T})-\text{tr}(UU^{\T}RD^{\T}DR^{\T}UU^{\T}).
	\$
	Taking $\tilde{U}=R^{\T}U$, we obtain
	\$
	&\quad\text{tr}(UU^{\T}RD^{\T}DR^{\T}UU^{\T})\\
	&=\text{tr}(U^{\T}UU^{\T}RD^{\T}DR^{\T}U)\\
	&=\text{tr}(\tilde{U}^{\T}D^{\T}D\tilde{U}). \tag{$U^{\T}U=I_q$}
	\$
To proceed, we need the following  lemma.
	\begin{lemma}[Von Neumann's trace inequality]
	Let  $A$, $B$ be  real  $n\times n$ matrices with singular values, $\alpha_1 \geq \cdots \alpha_n$, $\beta_1 \geq \cdots \geq \beta_n$, respectively. Then 
	\$
	|\text{tr}(AB)| \leq \sum_{i=1}^n \alpha_i \beta_i.
	\$
	\end{lemma}
	Let $\alpha_1$, $\beta_i$ be the singular values of $\tilde{U}\tilde{U}^{\T}$ and $D^{\T}D$ in descending order, respectively. Since $\tilde{U}^{\T}\tilde{U}$ is a projection matrix,  $\alpha_1=\cdots=\alpha_q=1, \alpha_{q+1}=\cdots=\alpha_{p}=0$. $D^{T}D$ is a diagonal matrix with non-negative descending diagnoal entries, thus $\beta_i=d_{i,i}^2$, where $d_{i,i}$ is the $(i,i)$th diagonal entry of $D$. We use Von Neumann's trace inequality to find a upper bound of $\text{tr}(\tilde{U}^{\T}D^{\T}D\tilde{U})$,
	\$
	\text{tr}(\tilde{U}^{\T}D^{\T}D\tilde{U})&=\text{tr}(\tilde{U}\tilde{U}^{\T} D^{\T}D)\\
	&\leq \sum_{i=1}^p \alpha_i \beta_i \tag{Von Neumann's trace inequality}\\
	&=\sum_{i=1}^q d_{i,i}^2,
	\$
	The last equality holds when $$\tilde{U}\tilde{U}^{\T}=\begin{bmatrix}
	I_q & 0 \\
	0& 0 
	\end{bmatrix}, U=R\tilde{U}=R(:,1:q),$$ which makes  $\alpha_1,\cdots,\alpha_q$ the first $q$ diagonal entries of $\Sigma$, where $M\Sigma N^{\T}=\text{svd}(\tilde{U}\tilde{U}^{\T})$. Therefore, the equality holds when $U=R(:,1:q)$.
\end{proof}
\subsection {Results of Different Low-rank Parameters}
Figures \ref{fig:video_q=1}-\ref{fig:video_q=3} show video surveillance results with different low-rank parameters $q=1,2,3$. The ordinary least squares method is able to recover the stationary background in the first two instances when $q=1$ or $q=2$, but fails  when $q=3$. Moreover, the ordinary least squares cannot recover the background with illumination changes even when $q=1$. This indicates the nonrobustness for ordinary least squares. Our method, on the other hand, recovers the background for all cases with $q=1,2,3$. 

\begin{figure}[H]
	\centering
	\includegraphics[width=5.5in]{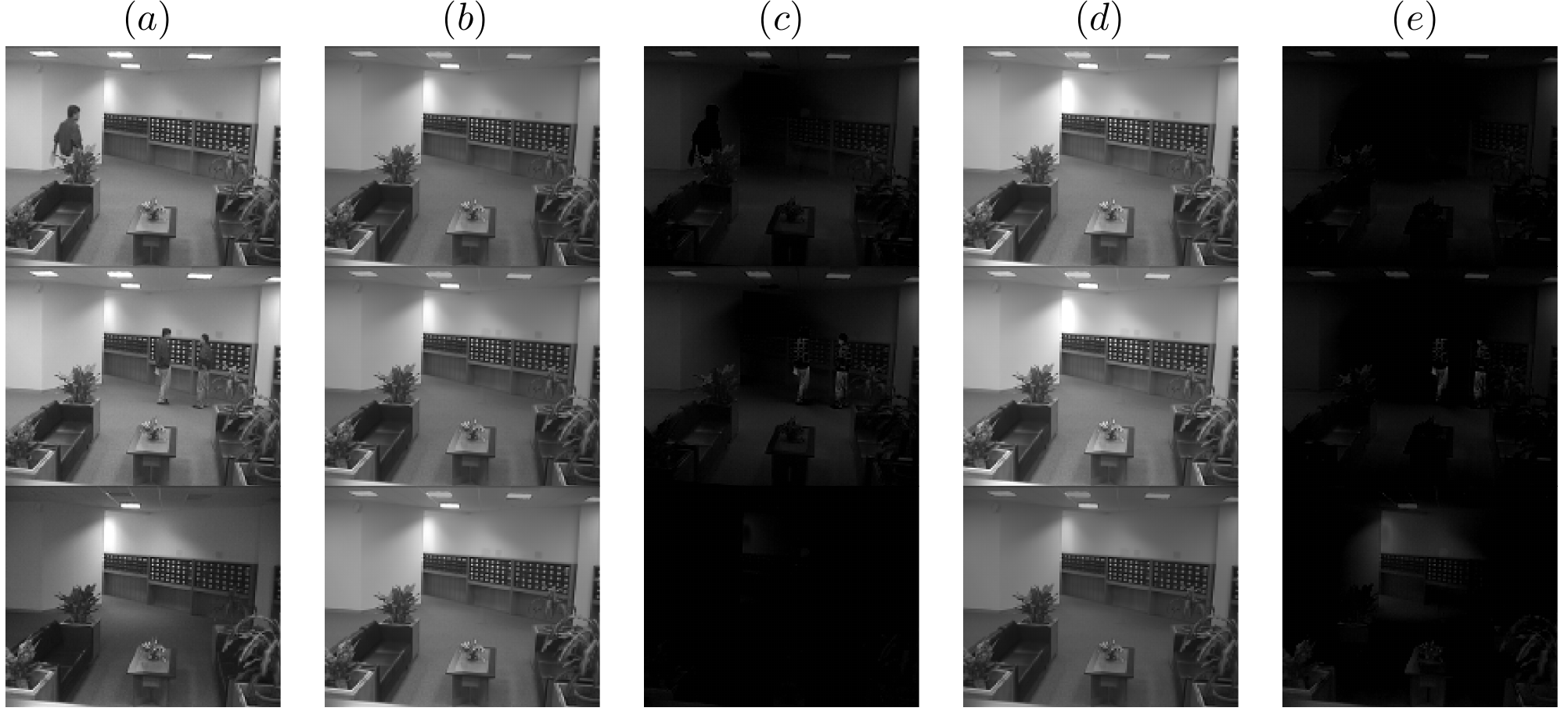} 
	\caption{{Video surveillance with low-rank parameter $q=1$: (a) Original Frames. (b)-(c) Background extraction and outliers from robust regression conducted on adaptive capped least squares regression model. (d)-(e) Background extraction and outliers from the ordinary least squares method.}}\label{fig:video_q=1}
\end{figure}

\begin{figure}[H]
	\centering
	\includegraphics[width=5.5in]{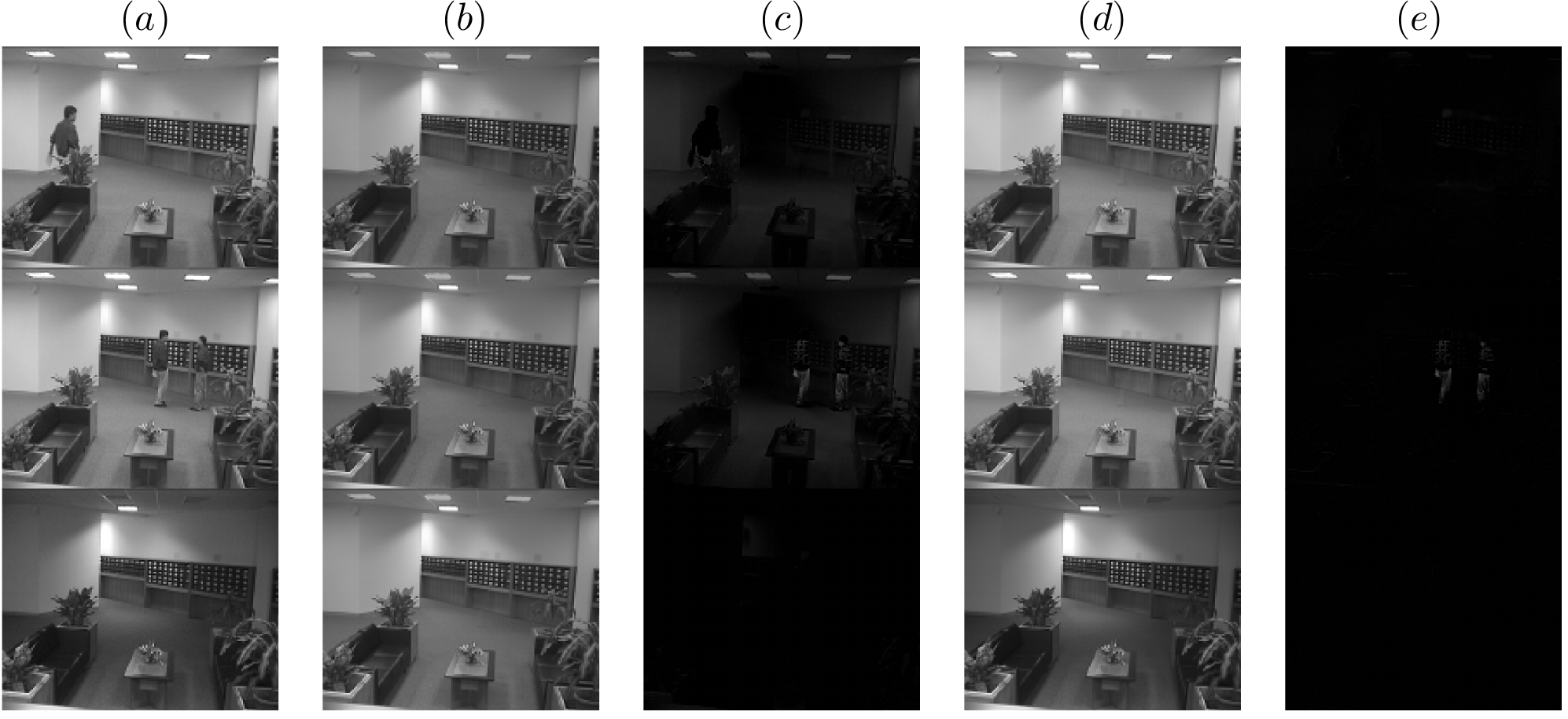} 
	\caption{{Video surveillance with low-rank parameter $q=2$: (a) Original Frames. (b)-(c) Background extraction and outliers from robust regression conducted on adaptive capped least squares regression model. (d)-(e) Background extraction and outliers from the ordinary least squares method.}}\label{fig:video_q=2}
\end{figure}

\begin{figure}[H]
	\centering
	\includegraphics[width=5.5in]{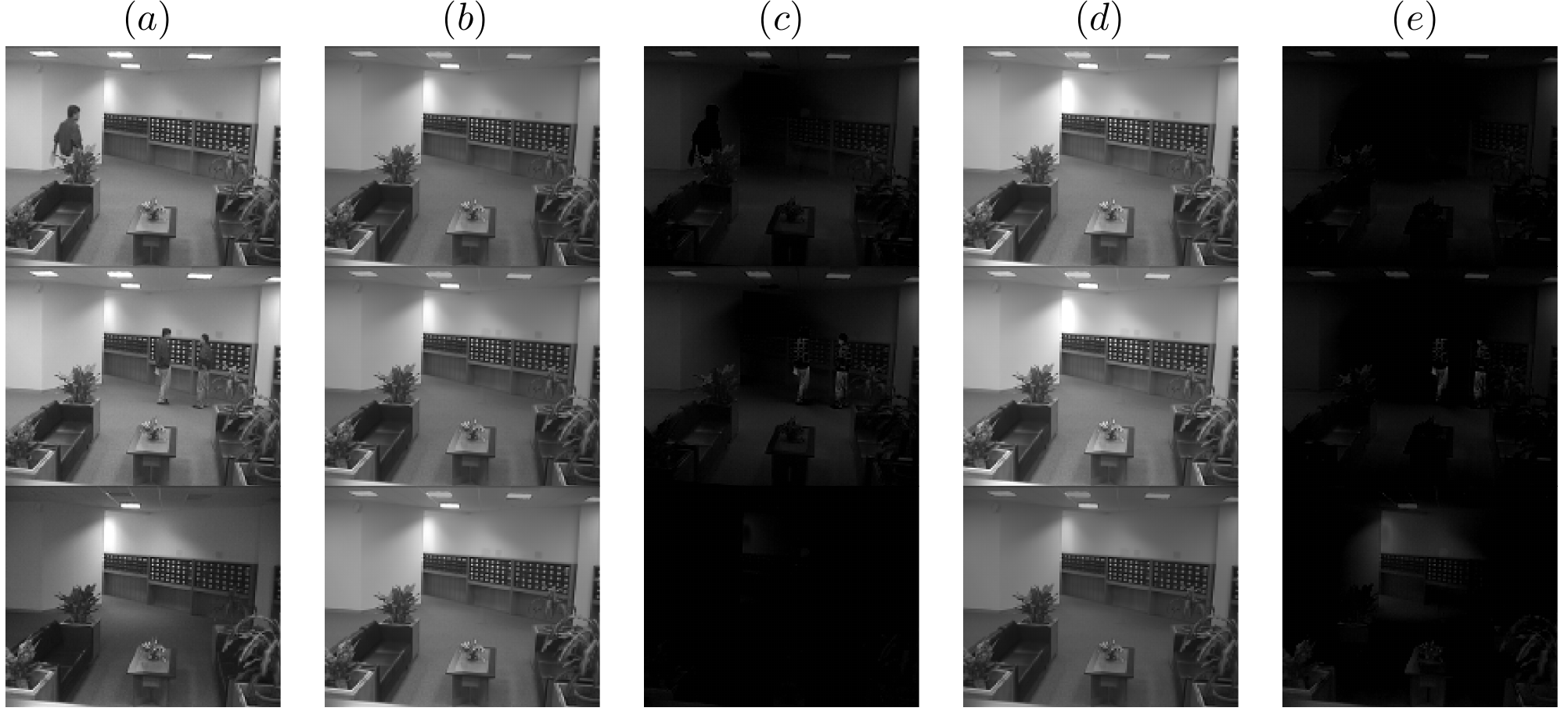} 
	\caption{{Video surveillance with low-rank parameter $q=3$: (a) Original Frames. (b)-(c) Background extraction and outliers from robust regression conducted on adaptive capped least squares regression model. (d)-(e) Background extraction and outliers from the ordinary least squares method.}}\label{fig:video_q=3}
\end{figure}

\section{Blind Image Inpainting}
We reformulate \eqref{eq:ARR_inpainting} as
\begin{align}\label{eq:inpainting2}
\argmin_{\alpha \in \mathbb{R}^{m \times p}, \{\delta_{i}\}_{i=1}^p} &\sum_{i=1}^{p}\Bigg(\frac{1}{2}\|(y_i-D\alpha_i)\odot (1_n-\delta_i)\|_2^2+\lambda\|\alpha_i\|_1 +\frac{1}{2}\tau^2\sum_{j=1}^n \delta_{ij} \Bigg ),\\
\text{s.t. }\quad&\delta_{ij}=1(|y_{ij}-[D\alpha_i]_j|>\tau), \;i=1,...,p, j=1,...,n,\nonumber
\end{align}
where $D$ is a dictionary matrix,  $\lambda$ is the regularization parameter, $\delta_{i}=(\delta_{i1},\cdots,\delta_{in})^{\T}$, $\delta_{ij}$ is a decision variable taking values 0 or 1, $[D\alpha_i]_j$ is the $j$th element of $D\alpha_i$, and $\odot$ denotes the Hadamard product. %We want to find a sparse vector $\hat{\alpha}\in \mathbb{R}^{m \times p}$ that solves \eqref{eq:inpainting2}. %\scomment{the dimensions do not match0Š70†70"400Š70†70"470Š70†70"400Š70†70"430Š70†70"400Š70†70"42 where did you define $\alpha$?}

% first rewrite \eqref{eq:inpainting2} in a more compact form. 
To solve \eqref{eq:inpainting2}, we develop an alternating minimization algorithm  to iteratively update $\alpha$ and $\delta_i$ in the order of $\alpha\rightarrow \delta_i$. Let %\scomment{what is the dimension of $Y$? Can you make it consistent?}
\$
f(\alpha, \Delta; Y, D,\lambda)= \frac{1}{2}\|(Y-D\alpha)\odot (1_n1_p^{\T}-\Delta)\|_{2,1}^2+\sum_{i=1}^{n}\lambda\|\alpha_i\|_1,
\$
where $\Delta=(\delta_{ij})$ and $\|\cdot\|_{2,1}^2$ is the $\ell_1$-norm of the row-wise $\ell_2$-norms. %\coloneqq %\|(\|(y_1-D\alpha_1)\odot (1_p -\delta_1)\|_2^2,...,\|(y_n-D\alpha_n)\odot (1_p-\delta_n)\|_2^2)^{\T}\|_1 $. 
For a fixed $\Delta$, optimizing \eqref{eq:inpainting2} with respect to $\alpha$  becomes 
\begin{equation}\label{eq:inpainting3}
\argmin_{\alpha \in \mathbb{R}^{m \times p}}f(\alpha, \Delta; Y, D,\lambda).
\end{equation}
For a fixed $\alpha$, we update $\delta_{ij}$ using the feasibility conditions $\delta_{ij}=1(|y_{ij}-[D\alpha_i]_j|>\tau), \, 1\leq i\leq p,\, 1\leq j\leq n$.
 %solve a relaxed version of \eqref{eq:inpainting2} with fixed $\delta_{ij} \in \{0,1\}, i=1, \cdots ,n, j=1,\cdots, p$, and to 

%Similarly, we adopt an alternating minimization algorithm to cyclically update $\alpha_i \rightarrow \delta_{ij}$. 
The full algorithm goes as follows. We  initialize $\Delta^0=1_n1_p^{\T}$. At iteration $k\geq 0$, we first apply Lasso to update 
\$
\alpha^{k+1}=\argmin_{\alpha \in \mathbb{R} ^{m \times p} } f(\alpha, \Delta^k; Y, D,\lambda),
\$then we update $\delta^{k+1}_{ij}=1(|y_{ij}-[D\alpha^{k+1}_i]_j|>\tau)$ and repeat the above steps until convergence, that is, $\Delta^{k+1}=\Delta^k$. We use the output $Y \odot (1-\Delta^{k+1})+D\alpha^{k+1} \odot \Delta^{k+1}$ for the signal matrix restoration. Algorithm \ref{alg:ARR_Lasso} summarizes the pseudo code. %\scomment{Again does randomized initialization work?}  

		\begin{algorithm}[!t]
		\caption{An alternating minimization algorithm for blind image inpainting using CLS.}\label{alg:ARR_Lasso}
		\begin{algorithmic}[1]
			\STATE{\textbf{Algorithm}: $\{\hat{\alpha}, \hat{\Delta}\} \leftarrow \text{CLS-inpainting}\big(Y, D,\tau)$}\\		
			\STATE{\textbf{Initialization}: $\Delta^0=1_n1_p^{\T}$}
			\STATE{\textbf{Input}: $\tau>0, D, \lambda$}\\
			\STATE{\textbf{for} $k=0,1,\cdots, \textbf{until } \Delta^{k+1} = \Delta^k$ }  ${\bf do}$\\
			$\alpha^{k+1}=\argmin_{\alpha \in \mathbb{R} ^{m \times p} }f(\alpha, \Delta^k; Y, D,\lambda)$,\\
            $\delta_{ij}^{k+1}=1(|y_{ij}-[D\alpha^{k+1}_i]_j|>\tau)\;\text{for}\; i=1,\cdots p, j=1,\cdots n$.\\
			\textbf{end for}\\
			\STATE{{\bf Output}: $\{\hat{\alpha}, \hat{\Delta}\}=\{\alpha^{k+1}, \Delta^{k+1}\}$}
		\end{algorithmic}
	\end{algorithm}

\end{document}